\documentclass[preprint,12pt,authoryear]{elsarticle}

\usepackage{amsmath,amssymb,amsfonts,amsthm,mathtools,bm}
\usepackage{booktabs}
\usepackage{graphicx}
\usepackage{enumitem}

\usepackage{geometry}
\usepackage{setspace}
\usepackage{microtype}
\usepackage{threeparttable}
\allowdisplaybreaks
\usepackage[colorlinks=true,linkcolor=blue,citecolor=blue,urlcolor=blue]{hyperref}
\journal{Journal of Econometrics}

\usepackage[thicklines]{cancel}
\usepackage{xcolor}

\theoremstyle{plain}
\newtheorem{theorem}{Theorem}[section]         
  
\newtheorem{lemma}{Lemma}[section]             
\newtheorem{corollary}{Corollary}[section] 
\newtheorem{assumption}{Assumption}[section] 
\theoremstyle{definition}
    
\theoremstyle{remark}
\newtheorem{remark}{Remark}[section]         

\newcommand{\E}{\mathbb{E}}
\newcommand{\Pbb}{\mathbb{P}}
\newcommand{\R}{\mathbb{R}}
\newcommand{\1}{\mathbf{1}}
\newcommand{\F}{\mathcal{F}}

\newcommand{\Var}{\operatorname{Var}}

\newcommand{\sgn}{\operatorname{sgn}}
\newcommand{\toP}{\xrightarrow{\mathbb{P}}}
\newcommand{\toD}{\xrightarrow{d}}

\newcommand{\iid}{\stackrel{\mathrm{iid}}{\sim}}
\newcommand{\ind}{\perp\!\!\!\perp}
\newcommand{\abs}[1]{\left|#1\right|}
\newcommand{\norm}[1]{\left\|#1\right\|}
\newcommand{\set}[1]{\left\{#1\right\}}
\newcommand{\paren}[1]{\left(#1\right)}
\newcommand{\brac}[1]{\left[#1\right]}

\newcommand{\dd}{\,\mathrm{d}}
\newcommand{\T}{\mathsf{T}}
\newcommand{\an}{a_n}
\newcommand{\mcS}{\mathcal{S}}

\begin{document}

\begin{frontmatter}

\title{Adaptive Test for Jump}

\author{Huifang Ma$^1$ and Long Feng$^1$\\
$^1$Nankai University}
\begin{abstract}
We develop an adaptive jump test for discretely observed high-frequency semimartingales by combining the A\"it-Sahalia--Jacod ratio statistic \citep{AitSahaliaJacod2009} and the Lee--Mykland extreme-return statistic \citep{LeeMykland2008} with the Cauchy combination rule. Allowing stochastic It\^o drift, volatility, and leverage, we show asymptotic independence under the continuous-path null and dense local alternatives, yielding an analytically calibrated test with closed-form power; under finite-activity jumps, the test is consistent. We also extend the method to additive microstructure noise. Simulations show that the combined procedure performs well under both dense and sparse alternatives and is typically best overall.
\end{abstract}


\end{frontmatter}

\section{Introduction}

Detecting jumps from high-frequency observations remains a central problem in financial econometrics.
The distinction between continuous and discontinuous return variation matters for volatility
measurement, derivative pricing, risk management, and the interpretation of how information is
incorporated into prices. Over the last two decades this problem has generated a large literature on
nonparametric and semiparametric jump tests for discretely observed processes.

The literature on high-frequency jump detection can be roughly grouped into four strands.

The first strand consists of specification-based and realized-variation approaches. \citet{AitSahalia2002} proposed an early specification test for assessing whether discretely observed data are consistent with an underlying diffusion, thereby providing a benchmark for separating diffusive dynamics from jump behavior. In the realized-variation literature, \citet{BarndorffNielsenShephard2004} introduced realized bipower variation as a jump-robust measure of integrated variance; \citet{HuangTauchen2005} interpreted jump tests in a Hausman-type framework and studied the contribution of jumps to total return variance; \citet{BarndorffNielsenShephard2006} established asymptotic theory for bipower-variation-based jump tests; and \citet{AndersenBollerslevDiebold2007} showed that separating continuous and jump variation improves volatility measurement and forecasting.

The second strand develops local, ratio-based, threshold, and truncation methods for nonparametric jump detection. \citet{LeeMykland2008} proposed a locally standardized test that identifies intraday jump times and jump sizes. \citet{AitSahaliaJacod2009} developed a general nonparametric test for jumps in discretely observed It\^o semimartingales. \citet{Mancini2009} introduced threshold methods for separating diffusion and jump increments and for volatility estimation in the presence of jumps. \citet{LeeHannig2010} extended jump detection to L\'evy jump-diffusion models and allowed for separate analysis of small and large jumps. \citet{PodolskijZiggel2010} proposed truncation-based tests that apply both in standard semimartingale settings and in the presence of market microstructure noise. \citet{CorsiPirinoReno2010} combined thresholding with bipower variation to reduce finite-sample bias and to study the effect of jumps on volatility forecasting.

The third strand extends the analysis to multivariate and dependent jump arrivals. \citet{JacodTodorov2009} developed tests for distinguishing common jumps from disjoint jumps in multidimensional processes. \citet{DungeyErdemliogluMateiYang2018} introduced tests for mutually exciting jumps and linked them empirically to flight-to-safety and flight-to-quality episodes. \citet{Kwok2024} further studied serial dependence in jump arrivals by proposing a robust test for autocorrelated jump occurrences.

A fourth strand provides comparative assessments and broader syntheses of the literature. \citet{DumitruUrga2012} compared leading nonparametric jump tests across different sampling frequencies, volatility dynamics, jump sizes, and noise levels. \citet{AitSahaliaJacod2012JEL} reviewed a unified framework for decomposing high-frequency returns into continuous and jump components and for measuring jump activity, while \citet{AitSahaliaJacod2014} offered a comprehensive treatment of high-frequency financial econometrics more broadly. \citet{ManeesoonthornMartinForbes2020} provided updated comparative guidance on the performance of high-frequency jump tests under volatility jumps, microstructure noise, and different sampling frequencies.

The present paper is built around two leading univariate tests. The first is the ratio statistic of
\citet{AitSahaliaJacod2009}, which aggregates power variation across sampling frequencies and is
naturally effective when jump evidence is spread over many increments. The second is the standardized
maximum statistic of \citet{LeeMykland2008}, which focuses on the largest local excursion and is
naturally effective when jumps are isolated and pronounced. These two statistics therefore target
different configurations of alternatives. The evidence in \citet{DumitruUrga2012} and
\citet{ManeesoonthornMartinForbes2020} shows that no single jump test dominates uniformly across
dense and sparse regimes or across sampling frequencies.

Our starting point is that the A\"it-Sahalia--Jacod and Lee--Mykland procedures should be viewed as complementary rather than competing. The A\"it-Sahalia--Jacod statistic aggregates evidence across many increments and is naturally sensitive to alternatives in which jump evidence is spread over time, whereas the Lee--Mykland statistic is driven by local extremes and is naturally sensitive to alternatives in which jumps are sparse and concentrated. This motivates us to study the relation between the two statistics directly. Our main theoretical result shows that, after the appropriate marginal normalization, the two statistics are asymptotically independent. 

In this respect, our paper is connected to a growing literature in high-dimensional hypothesis testing that studies the asymptotic independence between $\ell_q$-type and $\ell_\infty$-type statistics, or, more broadly, between sum-type and max-type statistics; see the classical probability references \citet{Hsing1995,HoHsing1996} and the more recent contributions \citet{LiXue2015,HeEtAl2021,FengEtAl2022,FengEtAl2024,WangEtAl2024}. Related adaptive $L_p$-type procedures include \citet{XuEtAl2016}. The present problem is nevertheless fundamentally different from the classical high-dimensional setting. That literature typically starts from i.i.d.\ high-dimensional vectors, or from coordinate-array models with dependence across coordinates, whereas here both statistics are nonlinear functionals of the same discretely observed It\^o semimartingale, with overlapping high-frequency increments, stochastic volatility, leverage, and, in the noisy extension, market microstructure frictions. As a result, the required asymptotic independence cannot be imported from existing high-dimensional arguments and must instead be proved directly from the primitive continuous-time model for $X$. This constitutes the main theoretical contribution of the paper. 

Building on this result, we propose a Cauchy combination test in the spirit of \citet{LiuXie2020}. We do not stop at the continuous-path null: under dense local jumps we again establish asymptotic independence, while under fixed finite-activity jumps and their noisy counterparts we derive the corresponding joint limits needed for the combined procedure and its consistency. Consequently, the proposed test remains analytically tractable across qualitatively different jump regimes. The simulation evidence supports this message: the A\"it-Sahalia--Jacod statistic is stronger when jump evidence is dense, the Lee--Mykland statistic is stronger when jumps are sparse and isolated, and the Cauchy combination inherits the strengths of both and delivers the most robust overall performance across the alternatives considered in the paper.

A second challenge appears when the observation frequency becomes very high: market microstructure
noise changes the marginal asymptotics of classical jump statistics. Early noise-aware contributions
include \citet{FanWang2007} and \citet{JiangOomen2008}. The general pre-averaging toolkit is
developed by \citet{JacodLiMyklandPodolskijVetter2009} and clarified by
\citet{MyklandZhang2016}. Direct noise-robust extensions of the two statistics used in this paper are
\citet{AitSahaliaJacodLi2012} for the sum-type ratio test and \citet{LeeMykland2012} for the
max-type test. More recent work extends the literature to robust spot-volatility-based jump
detection, endogenous sampling, and one-sided order-book noise; see \citet{Sun2024},
\citet{LiLiNolteNolteYu2026}, and \citet{BibingerHautschRistig2026}. The finite-sample implications
of ultra-high-frequency frictions and spurious detections are emphasized by
\citet{ChristensenOomenPodolskij2014}, \citet{BajgrowiczScailletTreccani2016}, and
\citet{ManeesoonthornMartinForbes2020}.

Our noisy extension follows the same adaptive logic as the latent-process part of the paper. We keep the continuous-time assumptions on the latent process $X$ unchanged and add only a layer of assumptions on the observation error. On the sum side, we adopt the pre-averaged robustification of the A\"it-Sahalia--Jacod statistic developed by \citet{AitSahaliaJacodLi2012}. On the max side, we start from the local-average extreme-value idea of \citet{LeeMykland2012}, but our treatment goes beyond that paper in an essential way. The theory in \citet{LeeMykland2012} is derived under a constant-volatility specification, whereas in our paper the latent price continues to satisfy the same stochastic It\^o semimartingale framework as in the noiseless sections, with volatility itself evolving as an It\^o process and leverage still allowed. For this reason we do not rely on the volatility standardization used in \citet{LeeMykland2012}. Instead, we construct a new feasible Lee--Mykland-type statistic by standardizing local averages with a spot-volatility estimator based on the two-scale realized spot variance of \citet{ZuBoswijk2014}. We then prove that this feasible standardization is asymptotically equivalent to the corresponding oracle normalization, and that the resulting noisy max statistic still satisfies the Gumbel extreme-value limit under the noisy continuous-path null.

This noisy extension is therefore not a routine modification of the latent-process analysis. Unlike the classical i.i.d.\ setting, both components are nonlinear functionals of the same noisy high-frequency It\^o semimartingale, and the proof must simultaneously handle pre-averaging, local averaging, additive sub-Gaussian noise, stochastic volatility, leverage, and feasible spot-volatility estimation. Within this framework we establish asymptotic independence between the pre-averaged AJJ statistic and the new Lee--Mykland-type statistic not only under the noisy null, but also under the noisy alternatives studied in the paper, including the dense local regime and the fixed finite-activity jump regime. This leads to a noise-robust Cauchy combination test with analytic asymptotic calibration. In finite samples, however, the convergence of extreme-value statistics is well known to be slow, and the additional noise layer further complicates size control. To improve finite-sample accuracy, we therefore implement a double bootstrap in the simulation study. The Monte Carlo and empirical evidence both support the same conclusion: the pre-averaged sum statistic is more effective when jump evidence is diffuse, the noisy max statistic is more effective when jumps are sparse and concentrated, and the Cauchy combination inherits the strengths of both and delivers a robust overall procedure for jump detection in the presence of microstructure noise.

The remainder of the paper is organized as follows. Section \ref{sec:model} introduces the process
model, the sampling scheme, and the two marginal jump statistics. Section \ref{sec:null} studies the
continuous-path null and derives the Cauchy combination test. Section \ref{sec:h1} studies a dense
local alternative and a fixed finite-activity alternative. Section \ref{sec:noise} extends the adaptive
construction to additive microstructure noise. Section \ref{sec:simulation} reports the Monte Carlo
evidence. Section \ref{sec:data} give a real data application. Section \ref{sec:conclusion} concludes. All proofs are collected in the appendix.

\section{Model and Methods}
\label{sec:model}

\subsection{Sampling scheme and primitive assumptions}

Let $T>0$ be fixed. The process $X=\{X_t:0\le t\le T\}$ is observed at the regular times
\begin{align*}
    t_i=i\Delta_n,
\qquad
\Delta_n=\frac{T}{n},
\qquad
 i=0,1,\dots,n,
\qquad
\Delta_n\to 0.
\end{align*}
Write
\begin{align*}
    \Delta_i^n X = X_{t_i}-X_{t_{i-1}}.
\end{align*}
Fix an integer $k\ge 2$ and a power $p>3$. Let
\begin{align*}
    m_n=\lfloor n/k\rfloor.
\end{align*}
Throughout, $B=(B^{(1)},\dots,B^{(d_B)})^\top$ denotes a $d_B$-dimensional standard Brownian
motion, for a fixed integer $d_B\ge 1$, on a filtered probability space
$(\Omega,\F,(\F_t)_{0\le t\le T},\Pbb)$. We write
\begin{align*}
    W=B^{(1)},
\qquad
\bar B=(B^{(2)},\dots,B^{(d_B)})^\top.
\end{align*}
The first component $W$ drives the return innovation. The remaining components in $\bar B$ collect
additional orthogonal coefficient drivers. Leverage is allowed because the same component $W$ may
also enter the It\^o dynamics of $b_t$ and $\sigma_t$.\par

The null hypothesis is imposed directly on the latent process.

\begin{assumption}[Continuous-path null]\label{ass:H0}
Under $H_0$,
\begin{align*}
&\dd X_t = b_t\dd t + \sigma_t\dd W_t, \\
&b_t = b_0+\int_0^t b_s^{(0)}\dd s+\int_0^t \beta_s^\top\dd B_s, \\
&\sigma_t = \sigma_0+\int_0^t \sigma_s^{(0)}\dd s+\int_0^t \gamma_s^\top\dd B_s, 
\end{align*}
where $b^{(0)}$ and $\sigma^{(0)}$ are real-valued c\`adl\`ag locally bounded processes, $\beta$ and
$\gamma$ are $\R^{d_B}$-valued c\`adl\`ag locally bounded processes, and there exist constants
$0<\underline\sigma<\overline\sigma<\infty$ such that
\begin{align*}
    \underline\sigma\le \sigma_t\le \overline\sigma,
\qquad 0\le t\le T,
\qquad \text{a.s.}
\end{align*}
\end{assumption}

Assumption \ref{ass:H0} is a primitive model assumption on $X$ itself. The local smoothness
conditions used by \citet{LeeMykland2008} are not assumed separately. \ref{app:aux}
derives from Assumption \ref{ass:H0} the required local regularity of $(b_t,\sigma_t)$, the
stochastic-Taylor expansion of the increments, and the block martingale representation used below.
The common notation $B$ is only a bookkeeping device: if the drift and volatility coefficients are
driven by several Brownian motions, they can be stacked into a single higher-dimensional Brownian
vector $B$. The leverage case corresponds to a nonzero first component in $\beta_t$ or $\gamma_t$.

\subsection{The A\"it-Sahalia--Jacod sum statistic}

Define the fine- and coarse-scale power variations
\begin{equation}\label{eq:AJvariations}
B_n(p)=\sum_{i=1}^n \abs{\Delta_i^n X}^p,
\qquad
B_n^{(k)}(p)=\sum_{j=1}^{m_n}\abs{\sum_{\ell=1}^k \Delta_{(j-1)k+\ell}^n X}^p.
\end{equation}
The sum-type statistic of \citet{AitSahaliaJacod2009} is the ratio
\begin{equation}\label{eq:AJratio}
R_n^{AJ}=\frac{B_n^{(k)}(p)}{B_n(p)}.
\end{equation}
Under a continuous path, the limit is $k^{p/2-1}$, whereas under fixed jumps the limit is one.\par

Let
\begin{align*}
    m_r = \E\abs{N}^r,
\qquad N\sim N(0,1),
\qquad
A_r = \int_0^T |\sigma_t|^r\dd t,
\end{align*}
and define the block kernel
\begin{equation}\label{eq:Upk}
U_{p,k}(x_1,\dots,x_k)=\abs{x_1+\cdots+x_k}^p-k^{p/2-1}\sum_{\ell=1}^k\abs{x_\ell}^p.
\end{equation}
Set
\begin{equation}\label{eq:tau0}
\varsigma_{p,k}^2 = \Var\paren{U_{p,k}(E_1,\dots,E_k)},
\qquad
\tau_0^2 = \frac{\varsigma_{p,k}^2 A_{2p}}{k m_p^2 A_p^2},
\end{equation}
where $E_1,\dots,E_k\stackrel{i.i.d.}{\sim}N(0,1)$. A feasible estimator of $\tau_0^2$ is
\begin{equation}\label{eq:tau0hat}
\hat\tau_{n,0}^2 = \frac{\varsigma_{p,k}^2}{k m_p^2}
\frac{\hat A_{2p,n}}{\hat A_{p,n}^2},
\qquad
\hat A_{r,n}=m_r^{-1}\Delta_n^{1-r/2}\sum_{i=1}^n \abs{\Delta_i^n X}^r.
\end{equation}
The feasible null statistic is
\begin{equation}\label{eq:AJnullstat}
Z_n^{AJ}=\frac{\Delta_n^{-1/2}\paren{R_n^{AJ}-k^{p/2-1}}}{\hat\tau_{n,0}},
\qquad
p_n^{AJ}=2\paren{1-\Phi\paren{\abs{Z_n^{AJ}}}}.
\end{equation}

\subsection{The Lee--Mykland max statistic}

Let $K_n\to\infty$ satisfy
\begin{equation}\label{eq:Kncond}
\frac{K_n}{(\log n)^3}\to\infty,\qquad K_n\Delta_n(\log n)^3\to 0,\qquad K_n\sqrt{\Delta_n}\to\infty.
\end{equation}
A convenient choic is $K_n=\lfloor \Delta_n^{-\varrho}\rfloor$ with $\varrho\in(1/2,1)$. Define the local
bipower estimator
\begin{equation}\label{eq:localbipower}
\hat V_{n,i}=\frac{\pi}{2(K_n-1)}\sum_{j=i-K_n+2}^{i} \abs{\Delta_j^n X}\,\abs{\Delta_{j-1}^n X},
\qquad i=K_n,\dots,n,
\end{equation}
and the standardized increments
\begin{equation}\label{eq:LML}
L_{n,i}=\frac{\Delta_i^n X}{\sqrt{\hat V_{n,i}}},
\qquad i=K_n,\dots,n.
\end{equation}
The max-type statistic is
\begin{equation}\label{eq:maxstat}
M_n=\max_{K_n\le i\le n}\abs{L_{n,i}}.
\end{equation}
Let
\begin{equation}\label{eq:LMnorm}
C_n=\sqrt{2\log n}-\frac{\log\pi+\log\log n}{2\sqrt{2\log n}},
\qquad
\an=\frac{1}{\sqrt{2\log n}},
\end{equation}
and define the normalized maximum
\begin{equation}\label{eq:xi_n}
\xi_n=\frac{M_n-C_n}{\an}.
\end{equation}
Under the continuous-path null, $\xi_n$ converges to the standard Gumbel law with distribution
function $x\mapsto \exp\{-e^{-x}\}$. The corresponding marginal $p$-value is
\begin{equation}\label{eq:pLM}
p_n^{LM}=1-\exp\set{-e^{-\xi_n}}.
\end{equation}

\subsection{The Cauchy combination statistic}

The equal-weight Cauchy combination statistic of \citet{LiuXie2020} is
\begin{equation}\label{eq:Cauchystat}
\T_n^C=
\frac12\tan\brac{\pi\paren{\frac12-p_n^{AJ}}}
+
\frac12\tan\brac{\pi\paren{\frac12-p_n^{LM}}}.
\end{equation}
The combined $p$-value is
\begin{equation}\label{eq:Cauchyp}
p_n^C=\frac12-\frac{1}{\pi}\arctan\paren{\T_n^C}.
\end{equation}
For $0<\alpha<1$, the upper $\alpha$-quantile of the standard Cauchy law is
\begin{align*}
    c_\alpha=\cot(\pi\alpha).
\end{align*}
The level-$\alpha$ Cauchy combination test rejects when $\T_n^C>c_\alpha$.

%\section{The continuous-path null}
%\label{sec:null}

This section develops the joint null theory directly from Assumption \ref{ass:H0}. No additional
assumption is imposed on the statistics themselves.

The appendix derives two intermediate ingredients from the primitive null model. The first is a block
expansion for the sum statistic. The second is a Gaussian approximation for the max statistic. With
those ingredients available, the joint null theorem can be stated directly in the main text.

\subsection{Asymptotic independence under \texorpdfstring{$H_0$}{H0}}\label{sec:null}

\begin{theorem}[Joint null limit]\label{thm:H0indep}
Suppose Assumption \ref{ass:H0} and \eqref{eq:Kncond} hold. Then, 
% for every $x,y\in\R$,
% \begin{equation*}
% \Pbb\paren{Z_n^{AJ}\le x,\ \xi_n\le y}
% \to
% \Phi(x)\exp\set{-e^{-y}}.
% \end{equation*}
% Equivalently,
\begin{equation*}
\paren{Z_n^{AJ},\xi_n}\toD (Z,\xi),
\end{equation*}
where
\begin{align*}
   Z\sim N(0,1),\qquad \Pbb(\xi\le y)=\exp\{-e^{-y}\},\qquad Z\ind \xi.
\end{align*}
\end{theorem}

\begin{corollary}[Null validity of the Cauchy combination test]\label{cor:nullcauchy}
Under the assumptions of Theorem \ref{thm:H0indep},
\begin{equation*}
\Pbb\paren{p_n^C\le \alpha}\to \alpha,
\qquad 0<\alpha<1.
\end{equation*}
% Equivalently,
% \begin{equation*}\label{eq:nullcauchy}
% \T_n^C\toD C,
% \end{equation*}
% where $C$ is standard Cauchy.
\end{corollary}

%\section{Alternative hypotheses}
%\label{sec:h1}

\subsection{A dense local alternative with many small jumps}\label{sec:h1}

\begin{assumption}[Dense local small jumps]\label{ass:H1dense}
Let $\mu_n$ be a Poisson random measure on $[0,T]\times\R$, independent of $(W,B)$, with
compensator
\begin{equation*}
\nu_n(\dd t,\dd y)=\theta\Delta_n^{-1/2}\dd t\,F(\dd y),
\qquad \theta>0,
\end{equation*}
where $F$ is a probability law on $\R$ with compact support and finite $(2p+2)$-moment. Under
$H_{1,n}^{D}$,
\begin{equation*}
X_t^{(n)}=X_0+\int_0^t b_s\dd s+\int_0^t \sigma_s\dd W_s
+\int_0^t\int_{\R} \sigma_{s-}\sqrt{\Delta_n}\,y\,\mu_n(\dd s,\dd y),
\qquad 0\le t\le T,
\end{equation*}
where $(b,\sigma)$ satisfy the conditions in Assumption
\ref{ass:H0}.
\end{assumption}

For $Y\sim F$ and independent $E_1,\dots,E_k\iid N(0,1)$, define
\begin{equation}\label{eq:dpk}
d_{p,k}=
\E\Big[
\abs{Y+E_1+\cdots+E_k}^p-
\abs{E_1+\cdots+E_k}^p-
 k^{p/2-1}\paren{\abs{Y+E_1}^p-\abs{E_1}^p}
\Big].
\end{equation}
The local mean shift of the AJ statistic is
\begin{equation}\label{eq:muD}
\mu_D=\frac{\theta d_{p,k}}{m_p\tau_0}.
\end{equation}

\begin{theorem}[Dense local alternative]\label{thm:densealt}
Suppose Assumption \ref{ass:H1dense} and \eqref{eq:Kncond} hold. Then
\begin{equation*}
\paren{Z_n^{AJ},\xi_n}\toD (Z_D,\xi),
\end{equation*}
where
\begin{equation*}
Z_D\sim N(\mu_D,1),
\qquad
\Pbb(\xi\le y)=\exp\set{-e^{-y}},
\qquad
Z_D\ind \xi.
\end{equation*}
% Equivalently,
% \begin{equation*}
% \paren{p_n^{AJ},p_n^{LM}}\toD \paren{P_D,U},
% \end{equation*}
% where
% \begin{equation*}
% P_D=2\paren{1-\Phi\paren{\abs{Z_D}}},
% \qquad
% U\sim \mathrm{Unif}(0,1),
% \qquad
% P_D\ind U.
% \end{equation*}
\end{theorem}

\begin{corollary}[Asymptotic power under the dense local alternative]\label{cor:densepower}
Under the assumptions of Theorem \ref{thm:densealt}, the power of the level-$\alpha$ Cauchy
combination test satisfies
\begin{equation*}
\Pbb\paren{p_n^C\le \alpha}\to \beta_D(\alpha;\mu_D),
\end{equation*}
where
\begin{align*}
\beta_D(\alpha;\mu_D)
=
\int_{\R}&
\paren{
\frac12-
\frac{1}{\pi}
\arctan\paren{2c_\alpha-A(z)}
}
\phi(z-\mu_D)\dd z,\\
A(z)=&\tan\brac{\pi\paren{\frac12-2\paren{1-\Phi\paren{\abs{z}}}}},
\end{align*}
$\Phi(\cdot)$ and $\phi(\cdot)$ are the cumulative distribution function and probability density function of the standard normal distribution, respectively.
\end{corollary}

\subsection{Fixed separated finite-activity jumps}

\begin{assumption}[Fixed finite-activity jumps]\label{ass:H1fixed}
Under $H_1^F$,
\begin{equation*}
X_t=X_0+\int_0^t b_s\dd s+\int_0^t \sigma_s\dd W_s+\sum_{q=1}^{N_T}\kappa_q\1\set{\tau_q\le t},
\qquad 0\le t\le T,
\end{equation*}
where $(b,\sigma)$ satisfy the conditions in Assumption
\ref{ass:H0}. Moreover,
\begin{equation*}
N_T<\infty\ \text{a.s.},
\qquad
0<\tau_1<\cdots<\tau_{N_T}<T\ \text{a.s.},
\qquad
\inf_{1\le q\le N_T}\abs{\kappa_q}>0\ \text{a.s.},
\end{equation*}
and
\begin{equation*}
\min_{q\ne r}\abs{\tau_q-\tau_r}>0\ \text{a.s.}
\end{equation*}
\end{assumption}

Define
\begin{equation*}
B_p=\sum_{q=1}^{N_T}\abs{\kappa_q}^p,\qquad
\tau_F^2=
\frac{p^2(k-1)}{B_p^2}
\sum_{q=1}^{N_T}\abs{\kappa_q}^{2p-2}\sigma_{\tau_q}^2.
\end{equation*}
The jump-centered oracle AJ statistic is
\begin{equation}\label{eq:oraclefixedAJ}
\widetilde Z_{n,1}^{AJ}=\frac{\Delta_n^{-1/2}\paren{R_n^{AJ}-1}}{\tau_F}.
\end{equation}

\begin{theorem}[Fixed finite-activity jumps]\label{thm:fixedalt}
Suppose Assumption \ref{ass:H1fixed} and \eqref{eq:Kncond} hold. Then
\begin{equation*}
\paren{\widetilde Z_{n,1}^{AJ},p_n^{LM}}\toD (Z,0),
\qquad Z\sim N(0,1).
\end{equation*}
In particular,
\begin{equation*}
\Pbb\paren{p_n^C\le \alpha}\to 1,
\qquad 0<\alpha<1.
\end{equation*}
\end{theorem}

\begin{remark}
The statistic $\widetilde Z_{n,1}^{AJ}$ is an oracle normalization because $\tau_F^2$ depends on the fixed jump times and sizes. It is used only to describe the centered AJ fluctuation under fixed finite-activity jumps. The consistency statement for the feasible Cauchy combination does not require estimating $\tau_F$: it follows from the feasible Lee--Mykland component, for which $p_n^{LM}\toP0$.
\end{remark}

\section{Extension to additive microstructure noise}
\label{sec:noise}

We now extend the adaptive construction to the case where the observed process is contaminated by
additive market microstructure noise. The continuous-path assumptions on the latent process are not
changed. Under the noisy continuous-path null we keep Assumption \ref{ass:H0}. Under the noisy
alternatives we keep exactly the same latent-process architecture as in Section \ref{sec:h1}; only the
local jump scale in the dense regime is retuned from $\sqrt{\Delta_n}$ to $\Delta_n^{1/4}$ so as to match
the slower $\Delta_n^{-1/4}$ normalization induced by pre-averaging. The genuinely new assumptions in
this section concern the observation error. This keeps the paper's process assumptions aligned across
its noiseless and noisy parts and makes the role of the additional noise layer transparent.\par

\begin{assumption}[Additive microstructure noise]\label{ass:noise}
The observations are
\begin{equation*}
Y_{t_i}=X_{t_i}+\epsilon_i,
\qquad i=0,1,\dots,n,
\end{equation*}
where $\epsilon_0,\epsilon_1,\dots,\epsilon_n$ are i.i.d., independent of $\F_T$, and satisfy
\begin{equation*}
\E(\epsilon_0)=0,
\qquad
\E(\epsilon_0^2)=\omega^2>0,
\qquad
\|\epsilon_0\|_{\psi_2}\le K_\epsilon<\infty.
\end{equation*}
% Equivalently, there exists a finite constant $K_\epsilon^\star$ such that
% \begin{equation*}
% \E\exp(t\epsilon_0)\le \exp\paren{\frac12 K_\epsilon^{\star 2} t^2},
% \qquad t\in\R.
% \end{equation*}
\end{assumption}

The Gaussian law is therefore not imposed on the observation error. The Gaussian distribution enters
this section only through the asymptotic law of the averaged noise blocks. In the pre-averaged sum
component, the relevant noise term is a weighted sum of many i.i.d. sub-Gaussian variables and hence
is asymptotically Gaussian. In the max component, the same averaging mechanism together with a
Cram\'er-type moderate deviation argument yields the Gaussian tail approximation needed for the
Gumbel limit at the extreme-value scale. Whenever squares or products of noise increments appear, we use the standard implication that products and squares of sub-Gaussian variables are sub-exponential, so Bernstein/Rosenthal bounds are applied at the sub-exponential rather than sub-Gaussian level. The sum component follows the pre-averaging construction of
\citet{AitSahaliaJacodLi2012}. The max component follows the local-average extreme-value
construction of \citet{LeeMykland2012}. The main point of the section is that, after marginal
calibration, the two noisy components remain asymptotically independent under the continuous-path
null and under the noisy dense local alternative.

\subsection{The pre-averaged sum statistic}

Fix an even integer $p\ge 4$. Let $g$ and $h$ be continuous weight functions on $\R$, piecewise
$C^1$, supported on $[0,1]$, with
\[
\bar g(2)=\int_0^1 g(u)^2\dd u>0,
\qquad
\bar h(2)=\int_0^1 h(u)^2\dd u>0.
\]
For a generic weight $\phi\in\{g,h\}$ and $q>0$, write
\[
\bar\phi(q)=\int_0^1 \abs{\phi(u)}^q\dd u,
\qquad
\bar\phi'(q)=\int_0^1 \abs{\phi'(u)}^q\dd u.
\]
Let $k_n\to\infty$ satisfy
\begin{equation}\label{eq:noise_kn}
k_n\sqrt{\Delta_n}\to\theta\in(0,\infty),
\qquad
\frac{k_n}{\log n}\to\infty.
\end{equation}
Define
\[
\phi_j^n=\phi\paren{\frac{j}{k_n}},
\qquad
\Delta_j\phi^n=\phi_j^n-\phi_{j-1}^n,
\qquad
j=0,1,\dots,k_n,
\]
with $\phi_0^n=\phi_{k_n}^n=0$. For $i=0,\dots,n-k_n$, set
\begin{equation*}
\bar Y_i^n(\phi)=\sum_{j=1}^{k_n-1}\phi_j^n\Delta_{i+j}^nY,
\qquad
\hat Y_i^n(\phi)=\sum_{j=1}^{k_n}\paren{\Delta_j\phi^n\,\Delta_{i+j}^nY}^2.
\end{equation*}
For integers $q\ge 0$ and $r\ge 0$ with $q+2r=p$, define
\begin{equation*}
V_n(Y,\phi,q,r)=\sum_{i=0}^{n-k_n}\abs{\bar Y_i^n(\phi)}^q\abs{\hat Y_i^n(\phi)}^r.
\end{equation*}
Let $\rho_0^{(p)},\dots,\rho_{p/2}^{(p)}$ be the unique solution of the triangular system
\begin{equation}\label{eq:noise_rho}
\rho_0^{(p)}=1,
\qquad
\sum_{l=0}^{j}2^l m_{2j-2l}\binom{p-2l}{p-2j}\rho_l^{(p)}=0,
\qquad
j=1,\dots,\frac{p}{2},
\end{equation}
and set
\begin{equation*}
V_n(Y,\phi,p)=\sum_{l=0}^{p/2}\rho_l^{(p)}V_n(Y,\phi,p-2l,l).
\end{equation*}
Define
\begin{equation*}
\gamma=\frac{\bar g(2)}{\bar h(2)},
\qquad
\gamma'=\frac{\bar g(p)}{\bar h(p)},
\qquad
\gamma''=\frac{\gamma^{p/2}}{\gamma'},
\qquad
\gamma''>1,
\end{equation*}
The inequality $\gamma''>1$ is a calibration restriction on the pair $(g,h)$, not an automatic
consequence of the preceding smoothness assumptions. It is easy to satisfy. For example, let
$w(u)=\sin(\pi u)$ on $[0,1]$, extended by zero outside $[0,1]$, and take
$g=w$ and $h=w^a$ for any $a>1$. If $L(q)=\log\int_0^1w(u)^q\dd u$, then $L$ is strictly convex and
\[
\frac{\dd}{\dd a}\left\{\frac12L(2a)-\frac1pL(pa)\right\}
=L'(2a)-L'(pa)<0,
\qquad p>2.
\]
Thus $a\mapsto \norm{w^a}_2/\norm{w^a}_p$ is strictly decreasing, so
\[
\gamma''
=\left(\frac{\norm{g}_2/\norm{g}_p}{\norm{h}_2/\norm{h}_p}\right)^p>1.
\]
The noise-robust sum statistic is
\begin{equation*}
R_n^{PA}=\frac{V_n(Y,g,p)}{\gamma'V_n(Y,h,p)}.
\end{equation*}
\ref{app:noise_sum_aux} constructs a block self-normalizer $\hat\varsigma_n^{PA}$ and proves
the central limit theorem
\begin{equation}\label{eq:noise_ZPA}
Z_n^{PA}=\frac{\Delta_n^{-1/4}\paren{R_n^{PA}-\gamma''}}{\hat\varsigma_n^{PA}},
\qquad
p_n^{PA}=2\paren{1-\Phi\paren{\abs{Z_n^{PA}}}}.
\end{equation}

\subsection{The local-average max statistic}

Let $M_n\to\infty$ satisfy
\begin{equation}\label{eq:noise_Mn}
M_n\sqrt{\Delta_n}\to\lambda\in(0,\infty),
\qquad
\frac{M_n}{\log n}\to\infty.
\end{equation}
For $j=0,\dots,n-M_n$, define the local averages $\bar Y_{n,j}$, and for
$j=0,\dots,n-2M_n$, define their adjacent-window difference by
\begin{equation}\label{eq:noise_localavg}
\bar Y_{n,j}=\frac{1}{M_n}\sum_{u=0}^{M_n-1}Y_{t_{j+u}},
\qquad
L_{n,j}=\bar Y_{n,j+M_n}-\bar Y_{n,j}.
\end{equation}
To obtain asymptotically independent windows, we work on the disjoint grid
\begin{equation*}
\mathcal J_n=\set{0,2M_n,4M_n,\dots,2M_n(N_n-1)},
\qquad
N_n=\left\lfloor \frac{n-2M_n}{2M_n}\right\rfloor+1.
\end{equation*}
For $j\in\mathcal J_n$, define the oracle variance proxy
\begin{equation}\label{eq:noise_vnj}
v_{n,j}^2=2\omega^2+\sigma_{t_j}^2\paren{\frac{2}{3}M_n^2\,\Delta_n+\frac{1}{3}\,\Delta_n},
\end{equation}
and the oracle standardized local-average statistic
\begin{equation}\label{eq:noise_Xnj}
X_{n,j}=\frac{\sqrt{M_n}\,L_{n,j}}{v_{n,j}}.
\end{equation}
Only the variance proxy $v_{n,j}$ needs to be estimated. We therefore use a plug-in array based on the
TSRSV spot-variance estimator of \citet{ZuBoswijk2014}. Let
\begin{equation}\label{eq:noise_tsrsv_tuning}
K_n^{SV}=\lfloor c_K\Delta_n^{-2/3}\rfloor,
\qquad
H_n=\lfloor c_h\Delta_n^{-5/6}\rfloor,
\qquad
h_n=H_n\Delta_n,
\end{equation}
where $c_K,c_h\in(0,\infty)$ are fixed constants. With this choice $h_n\asymp\Delta_n^{1/6}$, $(K_n^{SV})^{-2}(h_n\Delta_n)^{-1}\asymp\Delta_n^{1/6}$, and $K_n^{SV}\Delta_n/h_n\asymp\Delta_n^{1/6}$, matching the bias, noise and discretization terms in the filtering TSRSV expansion.  Define the noise-variance estimator
\begin{equation}\label{eq:noise_omegahat}
\hat\omega_n^2=\frac{1}{2n}\sum_{i=1}^n\paren{\Delta_i^nY}^2.
\end{equation}
For $r\ge K_n^{SV}$, let
\begin{equation*}
\operatorname{TSRV}_n(t_r;K_n^{SV})
=
\frac{1}{K_n^{SV}}\sum_{i=K_n^{SV}}^{r}\paren{Y_{t_i}-Y_{t_{i-K_n^{SV}}}}^2
-
\frac{r-K_n^{SV}+1}{K_n^{SV}r}\sum_{i=1}^{r}\paren{Y_{t_i}-Y_{t_{i-1}}}^2,
\end{equation*}
with $\operatorname{TSRV}_n(t_0;K_n^{SV})=0$. For $j\in\mathcal J_n$, let
\begin{equation}\label{eq:noise_tsrsv_sigmahat}
j^{\star}=j\vee H_n,
\qquad
\hat\sigma_{n,j}^2
=
\frac{\operatorname{TSRV}_n(t_{j^{\star}};K_n^{SV})-\operatorname{TSRV}_n(t_{j^{\star}-H_n};K_n^{SV})}{h_n},
\end{equation}
and define
\begin{equation}\label{eq:noise_vhatXhat}
\hat v_{n,j}^2=2\hat\omega_n^2+\hat\sigma_{n,j}^2\paren{\frac{2}{3}M_n^2\,\Delta_n+\frac{1}{3}\,\Delta_n},
\qquad
\hat X_{n,j}=\frac{\sqrt{M_n}\,L_{n,j}}{\hat v_{n,j}}.
\end{equation}
Lemma \ref{lem:noise_tsrsv_feasible} in \ref{app:noise_max_aux} proves that this plug-in array
satisfies
\begin{equation}\label{eq:noise_feas}
\frac{1}{B_{N_n}}\max_{j\in\mathcal J_n}\abs{\hat X_{n,j}-X_{n,j}}\toP 0.
\end{equation}
More generally, any feasible array satisfying \eqref{eq:noise_feas} may be used. The feasible max
statistic is then
\begin{equation}\label{eq:noise_MLA}
\mathcal M_n^{LA}=\max_{j\in\mathcal J_n}\abs{\hat X_{n,j}}.
\end{equation}
Let
\begin{equation*}
A_{N_n}=\sqrt{2\log N_n}-\frac{\log\pi+\log\log N_n}{2\sqrt{2\log N_n}},
\qquad
B_{N_n}=\frac{1}{\sqrt{2\log N_n}},
\end{equation*}
and define
\begin{equation}\label{eq:noise_xi}
\xi_n^{LA}=\frac{\mathcal M_n^{LA}-A_{N_n}}{B_{N_n}},
\qquad
p_n^{LA}=1-\exp\set{-e^{-\xi_n^{LA}}}.
\end{equation}

\subsection{The noisy Cauchy combination test}

For the joint proof we fix a block length $r_n$ such that
\begin{equation}\label{eq:noise_rn}
r_n\to\infty,
\qquad
\frac{r_n}{k_n}\to\infty,
\qquad
\frac{r_n}{M_n}\to\infty,
\qquad
r_n\Delta_n\to 0,\qquad \frac{k_n+M_n}{r_n^2\Delta_n}\to0.
\end{equation}
Equivalently, set $r_n=\Delta_n^{-a}$ with $3/4<a<1$ when $k_n\asymp M_n\asymp \Delta_n^{-1/2}$. Define the noisy Cauchy combination statistic
\begin{equation}\label{eq:noise_Cauchy}
\mathcal T_n^{C,N}
=
\frac12\tan\brac{\pi\paren{\frac12-p_n^{PA}}}
+
\frac12\tan\brac{\pi\paren{\frac12-p_n^{LA}}}
\end{equation}
and the associated combined $p$-value
\begin{equation}\label{eq:noise_Cauchyp}
p_n^{C,N}=\frac12-\frac{1}{\pi}\arctan\paren{\mathcal T_n^{C,N}}.
\end{equation}

\begin{theorem}[Noisy null: marginal limits and asymptotic independence]
\label{thm:noisejoint}
Suppose Assumptions \ref{ass:H0} and \ref{ass:noise} hold, and suppose \eqref{eq:noise_kn},
\eqref{eq:noise_Mn}, and \eqref{eq:noise_rn} hold. Then, 
% for every $x,y\in\R$,
% \begin{equation*}
% \Pbb\paren{Z_n^{PA}\le x,\ \xi_n^{LA}\le y}
% \to
% \Phi(x)\exp\set{-e^{-y}}.
% \end{equation*}
% Equivalently,
\begin{equation*}
\paren{Z_n^{PA},\xi_n^{LA}}\toD (Z,\xi),
\end{equation*}
where
\begin{align*}
   Z\sim N(0,1),\qquad \Pbb(\xi\le y)=\exp\{-e^{-y}\},\qquad Z\ind \xi.
\end{align*}
\end{theorem}

\begin{corollary}[Noisy Cauchy combination]
\label{cor:noisecauchy}
Under the assumptions of Theorem \ref{thm:noisejoint},
% \begin{equation*}
% \mathcal T_n^{C,N}\toD C,
% \qquad
% p_n^{C,N}\toD U(0,1),
% \end{equation*}
% where $C$ is standard Cauchy. Hence, 
\begin{equation*}
\Pbb\paren{p_n^{C,N}\le \alpha}\to \alpha,\qquad 0<\alpha<1.
\end{equation*}
\end{corollary}

\subsection{A noisy dense local alternative with many small jumps}

The noisy analogue of the dense local alternative keeps the same Poissonian latent-process structure
as Section \ref{sec:h1}, but the jump scale has to be retuned to the pre-averaging normalization.
After summation by parts, the pre-averaged noise term has variance of order $\Delta_n^{1/2}$, so the
unscaled noise contribution to $\bar Y_i^n(\phi)$ is of size $\Delta_n^{1/4}$ and
$\Delta_n^{-1/4}\bar Y_i^n(\phi)$ has a non-degenerate limit. A local jump must therefore also enter a
pre-averaging window at the $\Delta_n^{1/4}$ scale if it is to be comparable with the noise-robust
continuous component. The compensator below increases the jump intensity to order
$\Delta_n^{-1/4}$; with the $\Delta_n^{1/4}$ jump size this produces a finite first-order mean shift in
$\Delta_n^{-1/4}(R_n^{PA}-\gamma^{\prime\prime})$. The assumptions on $(b,\sigma)$ are not changed.

\begin{assumption}[Noisy dense local small jumps]\label{ass:noisedense}
Let $\mu_n^N$ be a Poisson random measure on $[0,T]\times\R$, independent of
$(B,\epsilon_0,\dots,\epsilon_n)$, with compensator
\begin{equation*}
\nu_n^N(\dd t,\dd y)=\vartheta\Delta_n^{-1/4}\dd t\,F(\dd y),
\qquad \vartheta>0,
\end{equation*}
where $F$ is a probability law on $\R$ with compact support and finite $(2p+2)$-moment. Under
$H_{1,n}^{D,N}$,
\begin{equation*}
X_t^{(n)}
=
X_0+
\int_0^t b_s\dd s+
\int_0^t \sigma_s\dd W_s+
\int_0^t\int_{\R}\sigma_{s-}\Delta_n^{1/4}y\,\mu_n^N(\dd s,\dd y),
\qquad 0\le t\le T,
\end{equation*}
where $(b,\sigma)$ satisfy the conditions in Assumption
\ref{ass:H0}.
\end{assumption}

For $\phi\in\{g,h\}$, let $E\sim N(0,1)$ and define, for $(y,u)\in\R\times[0,1]$,
\begin{equation*}
D_\phi(y,u)
=
\E\abs{\sqrt{\theta\bar\phi(2)}\,E+y\phi(u)}^p
-
m_p\paren{\theta\bar\phi(2)}^{p/2},\qquad
d_\phi^N
=
\int_{\R}\int_0^1 D_\phi(y,u)\,\dd u\,F(\dd y).
\end{equation*}
Let $\Sigma$ be the asymptotic variance of $\Delta_n^{-1/4}(R_n^{PA}-\gamma'')$ under the null (see Lemma \ref{lem:noise_block_clt}). The local mean shift of the pre-averaged ratio statistic is
\begin{equation}\label{eq:noise_dense_mu}
\mu_{D,N}
=
\frac{\vartheta}{\Sigma^{1/2}\gamma'm_p\theta^{p/2-1}\bar h(2)^{p/2}}
\paren{d_g^N-\gamma'\gamma''d_h^N}.
\end{equation}

\begin{theorem}[Noisy dense local alternative]\label{thm:noisedense}
Suppose Assumptions \ref{ass:noisedense} and \ref{ass:noise} hold, and suppose
\eqref{eq:noise_kn}, \eqref{eq:noise_Mn}, and \eqref{eq:noise_rn} hold. Then
\begin{equation*}
\paren{Z_n^{PA},\xi_n^{LA}}\toD (Z_{D,N},\xi),
\end{equation*}
where
\begin{equation*}
Z_{D,N}\sim N(\mu_{D,N},1),
\qquad
\Pbb(\xi\le y)=\exp\set{-e^{-y}},
\qquad
Z_{D,N}\ind \xi.
\end{equation*}
% Equivalently,
% \begin{equation*}
% \paren{p_n^{PA},p_n^{LA}}\toD \paren{P_{D,N},U},
% \end{equation*}
% where
% \begin{equation*}
% P_{D,N}=2\paren{1-\Phi\paren{\abs{Z_{D,N}}}},
% \qquad
% U\sim \mathrm{Unif}(0,1),
% \qquad
% P_{D,N}\ind U.
% \end{equation*}
\end{theorem}

\begin{corollary}[Noisy asymptotic power under the dense local alternative]
\label{cor:noisedensepower}
Under the assumptions of Theorem \ref{thm:noisedense},
\begin{equation*}
\Pbb\paren{p_n^{C,N}\le \alpha}\to \beta_{D,N}(\alpha;\mu_{D,N}),
\end{equation*}
where
\begin{equation*}
\beta_{D,N}(\alpha;\mu_{D,N})
=
\int_{\R}
\paren{
\frac12-
\frac{1}{\pi}
\arctan\paren{2c_\alpha-A(z)}
}
\phi(z-\mu_{D,N})\dd z,
\end{equation*}
and $A(\cdot)$ is given in Corollary \ref{cor:densepower}.
\end{corollary}

\begin{theorem}[Noisy fixed finite-activity jumps]
\label{thm:noisefixed}
Suppose Assumptions \ref{ass:H1fixed} and \ref{ass:noise} hold, suppose \eqref{eq:noise_kn}
and \eqref{eq:noise_Mn} hold, and suppose that the fixed jumps are not asymptotically aligned with the boundary points of the disjoint local-average grid. More precisely, let $i_{q,n}$ be the unique index satisfying $t_{i_{q,n}-1}<\tau_q\le t_{i_{q,n}}$, put
\[
a_{M_n,u}=u\wedge(2M_n-u),\qquad 1\le u\le 2M_n-1,
\]
and define
\[
D_n=
\max_{1\le q\le N_T}\max_{j\in\mathcal J_n}
 a_{M_n,i_{q,n}-j}
 \1\set{1\le i_{q,n}-j\le 2M_n-1}.
\]
Assume
\begin{equation}\label{eq:noisy_fixed_grid_condition}
\frac{D_n}{\sqrt{M_n\log n}}\toP\infty.
\end{equation}
Then
\begin{equation*}
R_n^{PA}\toP 1,
\qquad
p_n^{LA}\toP 0.
\end{equation*}
Consequently,
\begin{equation*}
\Pbb\paren{p_n^{C,N}\le \alpha}\to 1,
\qquad 0<\alpha<1.
\end{equation*}
\end{theorem}

\begin{remark}
Condition \eqref{eq:noisy_fixed_grid_condition} is only needed for the noisy local-average maximum computed on the single disjoint grid $\mathcal J_n$. It rules out the pathological case in which every fixed jump falls within $O(\sqrt{M_n\log n})$ observation points of a local-average grid boundary. If at least one fixed jump time has a bounded conditional density near its location, then \eqref{eq:noisy_fixed_grid_condition} holds; for example,
$\Pbb(D_n\le M_n^{3/4})=O(M_n^{-1/4})\to0$ and $M_n^{3/4}/\sqrt{M_n\log n}\to\infty$.
The pre-averaged ratio still satisfies $R_n^{PA}\toP1$, but the block self-normalizer in \eqref{eq:noise_ZPA} can be affected by fixed jump windows. The consistency statement is therefore proved from the local-average maximum component, which is sufficient for the Cauchy combination.
\end{remark}

\section{Simulation}
\label{sec:simulation}

This section reports Monte Carlo evidence on size and power. The designs are chosen to mirror the
continuous-path null, the dense local alternatives, and the sparse finite-activity alternatives studied
in the theory. Across all experiments one trading day contains $23{,}400$ one-second observations, and
we normalize the trading day to the unit interval. The latent log-price process is generated from
\begin{align}
\dd X_t &= -\frac12 V_t\dd t + \sqrt{V_t}\,\dd W_t + \dd J_t, \label{eq:sim_X}\\
\dd V_t &= \kappa(\bar\beta-V_t)\dd t + \gamma\sqrt{V_t}\,\dd B_t, \label{eq:sim_V}\\
\dd\langle W,B\rangle_t &= \rho\,\dd t, \qquad t\in[0,1], \label{eq:sim_cov}
\end{align}
with
\[
V_0=\bar\beta=0.16,
\qquad
\kappa=5,
\qquad
\gamma=0.5,
\qquad
\rho=-0.5.
\]
We simulate \eqref{eq:sim_X}--\eqref{eq:sim_cov} on the one-second grid and aggregate to sampling
intervals $\delta\in\{1,5,10,15,30\}$ seconds. For the A\"it-Sahalia--Jacod procedures we use the
same block sizes as in \citet{AitSahaliaJacod2009}: $k\in\{2,3,4\}$ at the 1- and 5-second
frequencies and $k=2$ at the 10-, 15-, and 30-second frequencies. The Lee--Mykland statistic is
implemented with the demeaned and locally standardized return and with window length
$K=\lceil n^{0.6}\rceil$; to improve finite-sample size, its critical values are obtained from a pathwise
parametric bootstrap with $199$ resamples. The Cauchy-combined procedures are denoted by CC-$k$.
The size results are based on $5{,}000$ replications, and the power results are based on $1{,}000$
replications.

\subsection{Size under the continuous-path null}

For the size experiment we set $J_t\equiv 0$ in \eqref{eq:sim_X}. Table
\ref{tab:size_aj_lm_cc} reports empirical rejection frequencies at the $10\%$ and $5\%$ nominal
levels.

\begin{table}[htbp]
\centering
\caption{Empirical size under the continuous-path null}
\label{tab:size_aj_lm_cc}
\small
\begin{tabular}{lcccc}
\toprule
$\delta$ (seconds) & $n$ & Method & 10\% level & 5\% level \\
\midrule
1  & 23,400 & AJ-2 & 0.0970 & 0.0434 \\
1  & 23,400 & AJ-3 & 0.1016 & 0.0472 \\
1  & 23,400 & AJ-4 & 0.0970 & 0.0416 \\
1  & 23,400 & CC-2 & 0.0946 & 0.0438 \\
1  & 23,400 & CC-3 & 0.0990 & 0.0432 \\
1  & 23,400 & CC-4 & 0.0920 & 0.0410 \\
1  & 23,400 & LM   & 0.1056 & 0.0502 \\
\addlinespace
5  & 4,680  & AJ-2 & 0.0988 & 0.0446 \\
5  & 4,680  & AJ-3 & 0.0892 & 0.0388 \\
5  & 4,680  & AJ-4 & 0.0870 & 0.0404 \\
5  & 4,680  & CC-2 & 0.0922 & 0.0420 \\
5  & 4,680  & CC-3 & 0.0842 & 0.0376 \\
5  & 4,680  & CC-4 & 0.0834 & 0.0382 \\
5  & 4,680  & LM   & 0.0916 & 0.0430 \\
\addlinespace
10 & 2,340  & AJ-2 & 0.0972 & 0.0488 \\
10 & 2,340  & CC-2 & 0.0906 & 0.0398 \\
10 & 2,340  & LM   & 0.0852 & 0.0404 \\
\addlinespace
15 & 1,560  & AJ-2 & 0.0930 & 0.0384 \\
15 & 1,560  & CC-2 & 0.0822 & 0.0326 \\
15 & 1,560  & LM   & 0.0832 & 0.0420 \\
\addlinespace
30 & 780    & AJ-2 & 0.0878 & 0.0356 \\
30 & 780    & CC-2 & 0.0764 & 0.0302 \\
30 & 780    & LM   & 0.0814 & 0.0384 \\
\bottomrule
\end{tabular}
\end{table}

The size results are satisfactory throughout. All three procedures are close to nominal levels at the
1- and 5-second frequencies, and the differences across methods are small. As the sampling interval
becomes coarser, LM and the corresponding Cauchy combinations become mildly conservative, but the
departures remain moderate. This makes the subsequent power comparisons informative about
adaptivity rather than about size distortions.

\subsection{Dense local alternatives}

To study the dense regime, we let the jump component in \eqref{eq:sim_X} satisfy
\begin{equation}
\dd J_t^{(n)} = \sqrt{V_{t-}}\int_{\R}\sqrt{\Delta_n}\,y\,\mu_n(\dd t,\dd y),
\qquad
\nu_n(\dd t,\dd y)=\theta\Delta_n^{-1/2}\dd t\,F(\dd y),
\label{eq:sim_dense}
\end{equation}
which matches Assumption \ref{ass:H1dense}. For the one-second experiments we take
$F=N(0,\sigma_Y^2)$ with $\sigma_Y^2\in\{0.5,1\}$. For the 5-, 10-, 15-, and 30-second grids we use
the corresponding Gaussian mark variances reported in Tables
\ref{tab:gaussian-y-var-half-sec-sq-5sec}--\ref{tab:dense-local-power-10-15-30sec}. Table
\ref{tab:gaussian-y-var-half}  report the one-second results,
and Tables \ref{tab:gaussian-y-var-half-sec-sq-5sec}--\ref{tab:dense-local-power-10-15-30sec}
report the coarser-grid results.

\begin{table}[htbp]
\centering
\caption{One-second power under the dense-local alternative with $Y_{i,j}\sim N(0,0.5)$}
\label{tab:gaussian-y-var-half}
\begin{tabular}{c|ccccccc}
\toprule
$\theta$ & AJ-2 & AJ-3 & AJ-4 & LM & CC-2 & CC-3 & CC-4 \\ \hline
\multicolumn{8}{c}{$Y_{i,j}\sim N(0,0.5)$}\\
\midrule
100 & 0.365 & 0.332 & 0.305 & 0.547 & 0.576 & 0.572 & 0.544 \\
200 & 0.521 & 0.523 & 0.474 & 0.676 & 0.749 & 0.749 & 0.718 \\
300 & 0.579 & 0.581 & 0.524 & 0.649 & 0.751 & 0.744 & 0.712 \\
400 & 0.606 & 0.596 & 0.553 & 0.616 & 0.748 & 0.741 & 0.707 \\
500 & 0.598 & 0.546 & 0.534 & 0.569 & 0.718 & 0.704 & 0.693 \\
\midrule
\multicolumn{8}{c}{$Y_{i,j}\sim N(0,1)$}\\ \hline
100 & 0.918 & 0.907 & 0.878 & 0.968 & 0.990 & 0.988 & 0.984 \\
200 & 0.957 & 0.945 & 0.923 & 0.953 & 0.991 & 0.993 & 0.988 \\
300 & 0.941 & 0.939 & 0.927 & 0.922 & 0.987 & 0.991 & 0.990 \\
400 & 0.910 & 0.916 & 0.889 & 0.885 & 0.977 & 0.970 & 0.964 \\
500 & 0.886 & 0.891 & 0.872 & 0.824 & 0.950 & 0.958 & 0.954 \\
\bottomrule
\end{tabular}
\end{table}

\begin{table}[htbp]
\centering
\caption{Five-second power under the dense-local alternative with $Y_{i,j}\sim N(0,12.5)$}
\label{tab:gaussian-y-var-half-sec-sq-5sec}
\begin{tabular}{c|ccccccc}
\toprule
$\theta$ & AJ-2 & AJ-3 & AJ-4 & LM & CC-2 & CC-3 & CC-4 \\
\midrule
100 & 0.702 & 0.718 & 0.703 & 0.676 & 0.820 & 0.833 & 0.827 \\
200 & 0.518 & 0.507 & 0.502 & 0.374 & 0.590 & 0.578 & 0.570 \\
300 & 0.420 & 0.431 & 0.398 & 0.251 & 0.472 & 0.459 & 0.442 \\
400 & 0.401 & 0.372 & 0.377 & 0.204 & 0.413 & 0.395 & 0.392 \\
500 & 0.319 & 0.293 & 0.294 & 0.146 & 0.315 & 0.306 & 0.300 \\
\bottomrule
\end{tabular}
\end{table}

\begin{table}[htbp]
\centering
\caption{Power under the dense-local alternative at coarser sampling frequencies}
\label{tab:dense-local-power-10-15-30sec}
\small
\setlength{\tabcolsep}{4.5pt}
\begin{tabular}{c ccc ccc ccc}
\toprule
& \multicolumn{3}{c}{10-second} & \multicolumn{3}{c}{15-second} & \multicolumn{3}{c}{30-second} \\
& \multicolumn{3}{c}{$Y_{i,j}\sim N(0,50)$} 
& \multicolumn{3}{c}{$Y_{i,j}\sim N(0,112.5)$} 
& \multicolumn{3}{c}{$Y_{i,j}\sim N(0,450)$} \\
\cmidrule(lr){2-4}\cmidrule(lr){5-7}\cmidrule(lr){8-10}
$\theta$ & AJ-2 & LM & CC-2 & AJ-2 & LM & CC-2 & AJ-2 & LM & CC-2 \\
\midrule
100 & 0.438 & 0.281 & 0.483 & 0.306 & 0.174 & 0.322 & 0.117 & 0.079 & 0.107 \\
200 & 0.306 & 0.165 & 0.304 & 0.181 & 0.099 & 0.184 & 0.063 & 0.047 & 0.061 \\
300 & 0.231 & 0.124 & 0.229 & 0.136 & 0.080 & 0.119 & 0.043 & 0.049 & 0.045 \\
400 & 0.187 & 0.110 & 0.184 & 0.113 & 0.070 & 0.107 & 0.033 & 0.043 & 0.038 \\
500 & 0.149 & 0.084 & 0.141 & 0.086 & 0.071 & 0.080 & 0.018 & 0.050 & 0.024 \\
\bottomrule
\end{tabular}
\end{table}

The dense-local designs make the complementarity of AJ and LM transparent. At the one-second
frequency with the weaker Gaussian marks in Table \ref{tab:gaussian-y-var-half}, LM is strongest
only when the dense signal is relatively weak; as $\theta$ increases, AJ catches up and eventually
overtakes LM. When the mark variance is increased to one,
this reversal appears earlier: LM is still slightly ahead at $\theta=100$, AJ is already competitive at
$\theta=200$, and AJ is stronger than LM for $\theta\ge 300$.

The same pattern is even clearer on the 5-, 10-, and 15-second grids. In Tables
\ref{tab:gaussian-y-var-half-sec-sq-5sec}--\ref{tab:dense-local-power-10-15-30sec}, AJ is above LM
for most designs, exactly as one would expect when jump evidence is distributed over many intervals.
On the coarsest 30-second grid all procedures lose power, and LM becomes comparable to or slightly
stronger than AJ in the weakest designs, but the broad message remains unchanged: in dense regimes
the sum statistic is the more informative marginal component. Across all dense-local designs, the
Cauchy-combined procedures are uniformly strongest or very close to strongest. The combined test
therefore retains LM's advantage when the dense signal is weak and inherits AJ's advantage once the
signal becomes more diffuse, which is precisely the adaptive gain delivered by our theory.

\subsection{Sparse alternatives}

To study the sparse regime, we replace the jump component in \eqref{eq:sim_X} by a finite-activity
compound Poisson process,
\begin{equation}
\dd J_t = \int_{\R} y\,\mu(\dd t,\dd y),
\qquad
\nu(\dd t,\dd y)=\lambda\dd t\,G(\dd y),
\qquad
G=N(0,0.05),
\label{eq:sim_sparse}
\end{equation}
while keeping the stochastic-volatility dynamics in \eqref{eq:sim_V}--\eqref{eq:sim_cov} unchanged.
Table \ref{tab:sparse-alternative-power} reports rejection frequencies at the $5\%$ level.

\begin{table}[htbp]
\centering
\small
\caption{Power under the sparse finite-activity alternative with $Y\sim N(0,0.05)$}
\label{tab:sparse-alternative-power}
\begin{tabular}{ccrrrrrrr}
\toprule
$\lambda$ & sec & AJ-2 & AJ-3 & AJ-4 & LM & CC-2 & CC-3 & CC-4 \\
\midrule
1.0 & 1  & 0.633 & 0.619 & 0.620 & 0.632 & 0.640 & 0.637 & 0.636 \\
1.0 & 5  & 0.584 & 0.584 & 0.586 & 0.615 & 0.611 & 0.609 & 0.609 \\
1.0 & 10 & 0.557 & -     & -     & 0.590 & 0.602 & -     & -     \\
1.0 & 15 & 0.516 & -     & -     & 0.571 & 0.580 & -     & -     \\
1.0 & 30 & 0.436 & -     & -     & 0.537 & 0.545 & -     & -     \\
\addlinespace
1.5 & 1  & 0.762 & 0.759 & 0.761 & 0.780 & 0.785 & 0.781 & 0.785 \\
1.5 & 5  & 0.718 & 0.720 & 0.723 & 0.747 & 0.755 & 0.753 & 0.752 \\
1.5 & 10 & 0.687 & -     & -     & 0.720 & 0.734 & -     & -     \\
1.5 & 15 & 0.645 & -     & -     & 0.706 & 0.726 & -     & -     \\
1.5 & 30 & 0.542 & -     & -     & 0.663 & 0.683 & -     & -     \\
\addlinespace
2.0 & 1  & 0.842 & 0.841 & 0.842 & 0.860 & 0.858 & 0.860 & 0.858 \\
2.0 & 5  & 0.815 & 0.815 & 0.804 & 0.834 & 0.846 & 0.840 & 0.838 \\
2.0 & 10 & 0.768 & -     & -     & 0.822 & 0.826 & -     & -     \\
2.0 & 15 & 0.716 & -     & -     & 0.797 & 0.811 & -     & -     \\
2.0 & 30 & 0.604 & -     & -     & 0.757 & 0.775 & -     & -     \\
\addlinespace
2.5 & 1  & 0.891 & 0.893 & 0.891 & 0.905 & 0.905 & 0.905 & 0.904 \\
2.5 & 5  & 0.852 & 0.851 & 0.844 & 0.881 & 0.886 & 0.885 & 0.882 \\
2.5 & 10 & 0.801 & -     & -     & 0.867 & 0.872 & -     & -     \\
2.5 & 15 & 0.760 & -     & -     & 0.846 & 0.856 & -     & -     \\
2.5 & 30 & 0.629 & -     & -     & 0.808 & 0.824 & -     & -     \\
\bottomrule
\end{tabular}
\end{table}

Table \ref{tab:sparse-alternative-power} reverses the dense-local ranking. Under sparse finite-activity
jumps, LM is generally more powerful than AJ, and the advantage becomes more pronounced on the
coarser sampling grids. This is exactly the regime in which isolated extreme returns carry most of the
information, so a max-type procedure is naturally favored. The Cauchy-combined procedures remain
best or near-best throughout, with CC-2 especially robust across frequencies. Taken together with the
dense-local tables, the sparse results make the main contribution of the paper transparent: AJ is the
stronger marginal procedure in dense regimes, LM is the stronger marginal procedure in sparse
regimes, and the Cauchy combination effectively synthesizes the strengths of both.

\subsection{Microstructure noise}

\begin{table}[htbp]
\centering
\small
\setlength{\tabcolsep}{5pt}
\caption{Empirical size with corrected LM bootstrap}
\label{tab:size-corrected-lm-bootstrap}
\begin{tabular}{lllcccccc}
\toprule
Noise & $q$ & sec &
\multicolumn{3}{c}{10\% level} &
\multicolumn{3}{c}{5\% level} \\
\cmidrule(lr){4-6} \cmidrule(lr){7-9}
 &  &  & AJJ & LM & CC & AJJ & LM & CC \\
\midrule
\multicolumn{9}{l}{\textit{Panel A: Gaussian noise}} \\
normal & 0.005 & 5  & 0.1060 & 0.0952 & 0.0998 & 0.0550 & 0.0452 & 0.0476 \\
normal & 0.005 & 10 & 0.1122 & 0.0924 & 0.0986 & 0.0500 & 0.0448 & 0.0408 \\
normal & 0.005 & 15 & 0.1138 & 0.0934 & 0.0962 & 0.0390 & 0.0484 & 0.0348 \\
normal & 0.010 & 5  & 0.1110 & 0.0966 & 0.1024 & 0.0530 & 0.0484 & 0.0450 \\
normal & 0.010 & 10 & 0.1250 & 0.0946 & 0.1028 & 0.0548 & 0.0478 & 0.0476 \\
normal & 0.010 & 15 & 0.1062 & 0.0964 & 0.0896 & 0.0364 & 0.0440 & 0.0308 \\
\midrule
\multicolumn{9}{l}{\textit{Panel B: Standardized $t_8$ noise}} \\
$t_8$ & 0.005 & 5  & 0.1146 & 0.0980 & 0.1042 & 0.0586 & 0.0468 & 0.0508 \\
$t_8$ & 0.005 & 10 & 0.1232 & 0.1006 & 0.1042 & 0.0516 & 0.0494 & 0.0446 \\
$t_8$ & 0.005 & 15 & 0.1036 & 0.1060 & 0.0968 & 0.0354 & 0.0566 & 0.0338 \\
$t_8$ & 0.010 & 5  & 0.1134 & 0.0952 & 0.0968 & 0.0526 & 0.0422 & 0.0422 \\
$t_8$ & 0.010 & 10 & 0.1176 & 0.0968 & 0.1016 & 0.0540 & 0.0462 & 0.0416 \\
$t_8$ & 0.010 & 15 & 0.1166 & 0.0996 & 0.0928 & 0.0386 & 0.0466 & 0.0310 \\
\bottomrule
\end{tabular}
\end{table}

\begin{figure}[htbp]
\centering
\includegraphics[width=0.95\textwidth]{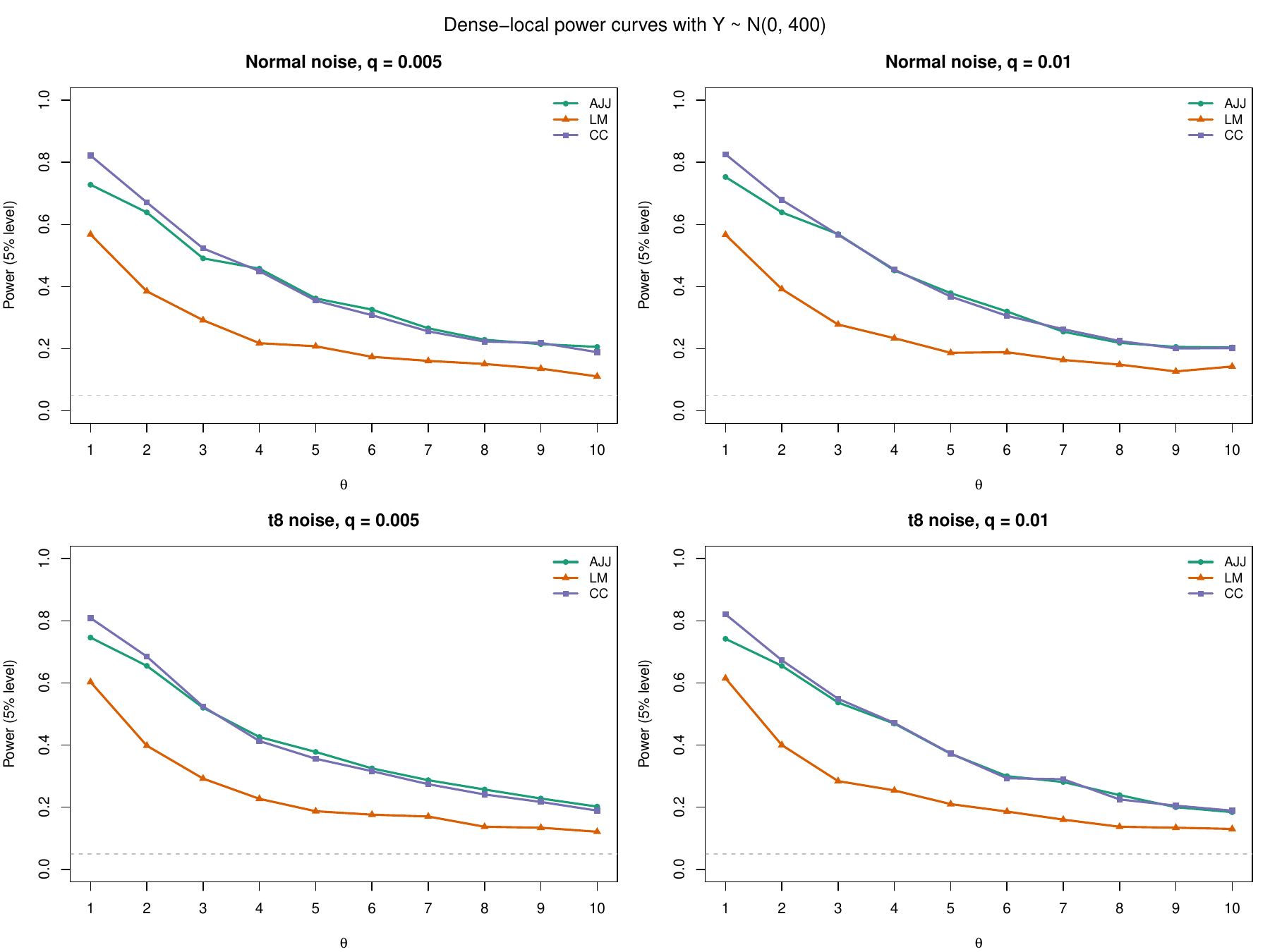}
\caption{Power curves at the $5\%$ level under the dense-local alternative with
$Y_{im}\sim N(0,400)$ and 5 secs.}
\label{fig:dense-power-gaussianY-theta1to10}
\end{figure}

\begin{figure}[htbp]
\centering
\includegraphics[width=0.95\textwidth]{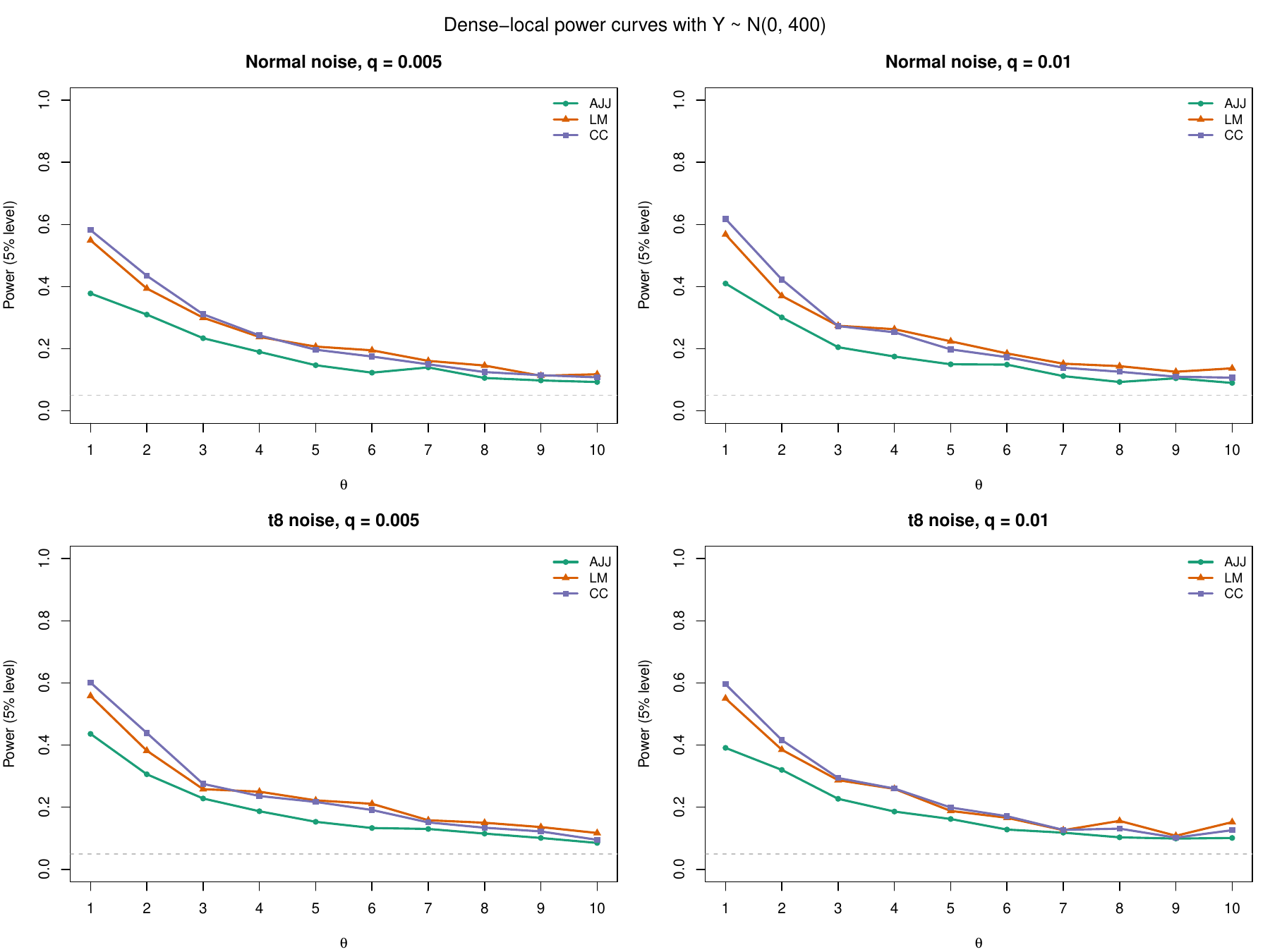}
\caption{Power curves at the $5\%$ level under the dense-local alternative with
$Y_{im}\sim N(0,400)$ and 10 secs.}
\label{fig:dense-power-gaussianY-theta1to10dt10}
\end{figure}

\begin{figure}[htbp]
\centering
\includegraphics[width=0.95\textwidth]{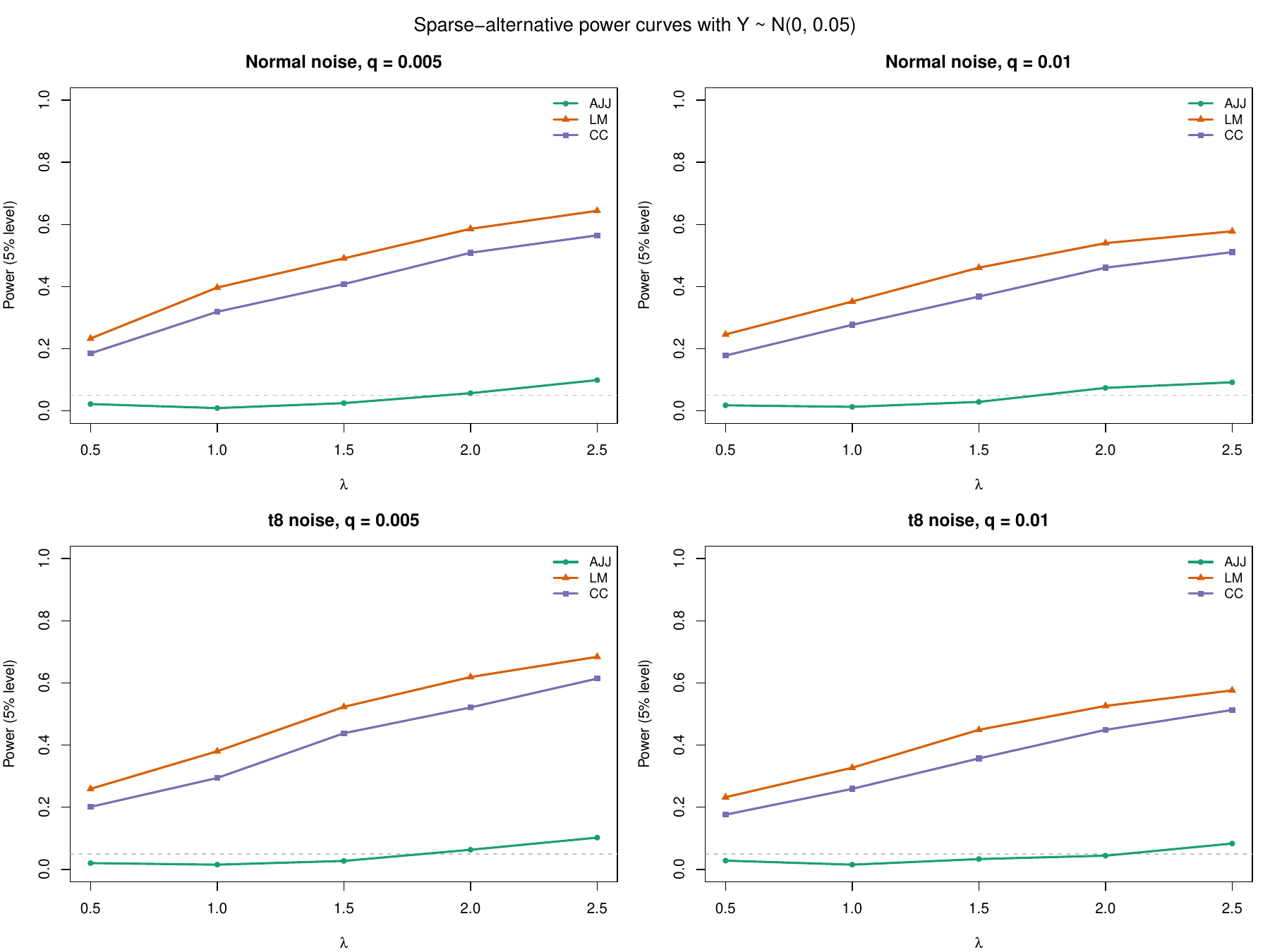}
\caption{Sparse-alternative power at the $5\%$ level, $5$-second sampling}
\label{fig:sparse-noise-dt5}
\end{figure}

\begin{figure}[htbp]
\centering
\includegraphics[width=0.95\textwidth]{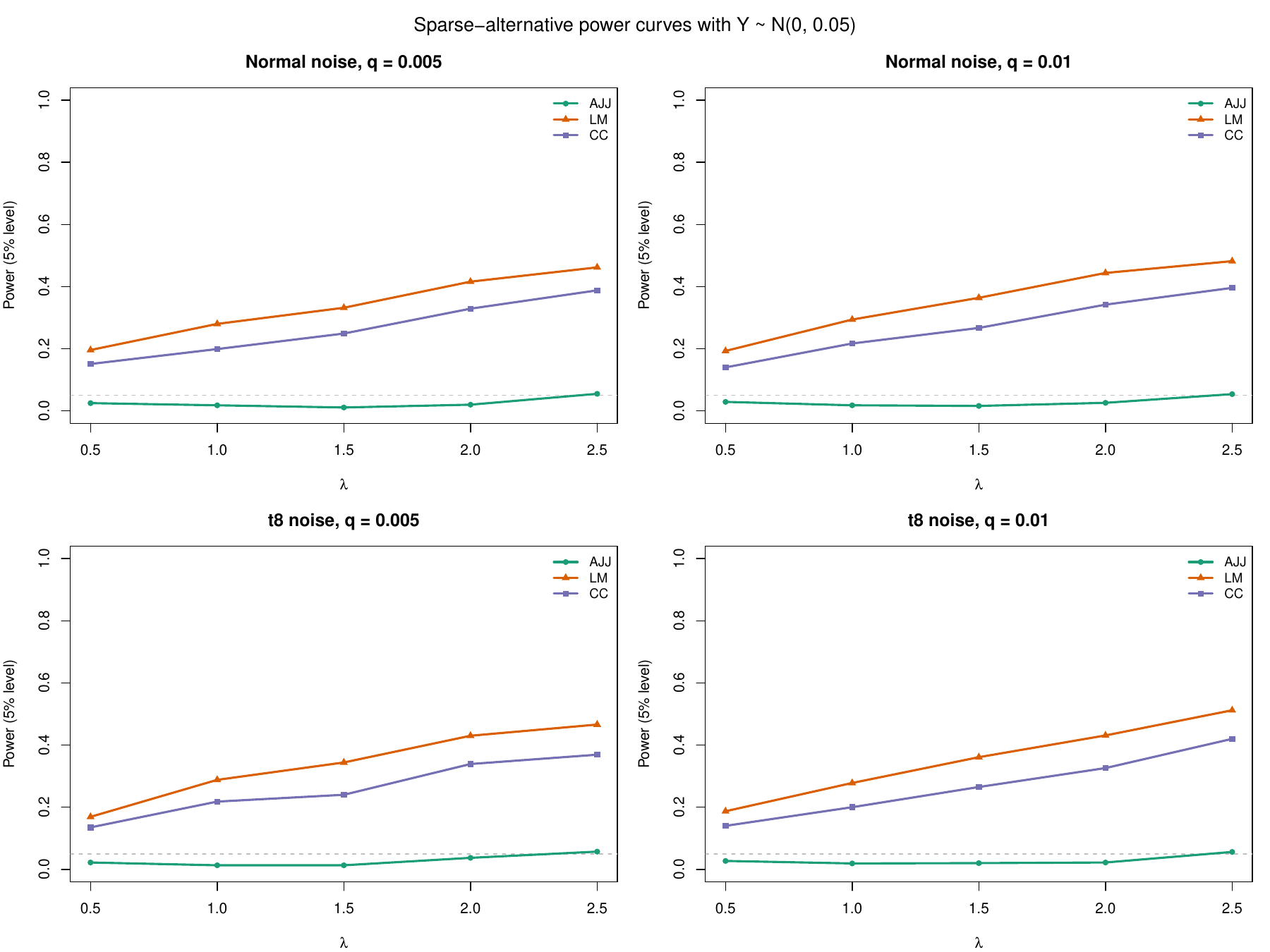}
\caption{Sparse-alternative power at the $5\%$ level, $10$-second sampling}
\label{fig:sparse-noise-dt10}
\end{figure}

We next examine the performance of AJJ, LM, and CC when the observed price is contaminated by market microstructure noise. To keep the comparison with the baseline experiments transparent, we retain the same latent efficient-price dynamics and parameter values as in the previous Monte Carlo designs and only modify the observation equation to
\[
Y_{t_i}=X_{t_i}+\varepsilon_i,
\qquad
t_i=i\Delta_n,
\qquad
\Delta_n=\frac{dt_{\mathrm{sec}}}{23400},
\]
over a fixed horizon of \(T=5\) trading days. We consider \(dt_{\mathrm{sec}}\in\{5,10\}\) and two i.i.d. noise specifications with common variance \(q^2\),
\[
\varepsilon_i\sim N(0,q^2),
\qquad\text{or}\qquad
\varepsilon_i=q\sqrt{\frac{3}{4}}\,T_i,
\qquad
T_i\stackrel{iid}{\sim} t_8,
\]
with \(q\in\{0.005,0.01\}\). All remaining tuning parameters are kept as close as possible to the noise-free experiments: AJJ uses \(k_n=100\) and \(r_n=1000\), while LM is implemented under the i.i.d.-noise specification with \(k=1\), block length \(M_n\) chosen as in \citet{LeeMykland2012}, and spot volatility estimated by the local TSRSV method. The combined procedure CC is computed from the AJJ and LM \(p\)-values exactly as before.

\paragraph{Size under the null.}
We begin with the continuous-path null. Relative to the previous implementation, the main change is the calibration of LM. Instead of the one-step recursive bootstrap with an ex post critical-value adjustment, we now use a double bootstrap. Let \(T_n^{\mathrm{LM}}\) denote the LM statistic computed from the observed sample, and let \(\hat q\) and \(\hat\sigma_t^2\) denote the corresponding estimates of the noise variance and spot volatility. At the first stage, we generate \(B_1\) bootstrap samples from the fitted continuous no-jump model, re-estimate the nuisance quantities on each bootstrap path, and recompute the LM statistic, yielding \(T_n^{*(1)},\ldots,T_n^{*(B_1)}\). The first-stage bootstrap \(p\)-value is
\[
\hat p_n^*
=
\frac{1}{B_1}\sum_{b=1}^{B_1}
\mathbf{1}\!\left\{T_n^{*(b)}\ge T_n^{\mathrm{LM}}\right\}.
\]
For each first-stage sample \(b\), we then generate \(B_2\) second-stage bootstrap samples from the fitted first-stage model, again re-estimate the nuisance quantities, and compute
\[
\hat p_b^{**}
=
\frac{1}{B_2}\sum_{\ell=1}^{B_2}
\mathbf{1}\!\left\{T_n^{**(b,\ell)}\ge T_n^{*(b)}\right\},
\qquad
b=1,\ldots,B_1.
\]
The LM test rejects when \(\hat p_n^*\) is smaller than the empirical \(\alpha\)-quantile of
\(\{\hat p_b^{**}:1\le b\le B_1\}\). At both bootstrap stages, the bootstrap paths are generated from the fitted continuous Gaussian model, with noise variance set equal to the corresponding estimated \(q^2\). Table~\ref{tab:size-corrected-lm-bootstrap} shows that this calibration keeps the size of LM close to the nominal level. More importantly, size is well controlled for all three procedures across the two sampling frequencies, the two noise levels, and both the Gaussian and standardized-\(t_8\) noise designs. In the noisy experiments, the relevant differences across procedures therefore come mainly from power.

\paragraph{Dense alternatives.}
We next return to the dense-local alternative, now under noisy observations. As in the baseline dense-jump design, the latent process contains the interval-level jump component
\[
\Delta J_i
=
\sqrt{v_{t_{i-1}}}\sqrt{\Delta_n}\sum_{m=1}^{N_i}Y_{im},
\qquad
N_i\sim \mathrm{Poisson}(\theta\sqrt{\Delta_n}),
\qquad
Y_{im}\sim N(0,400),
\]
and we report 5\%-level power for \(\theta=1,\ldots,10\). Figures~\ref{fig:dense-power-gaussianY-theta1to10} and \ref{fig:dense-power-gaussianY-theta1to10dt10} reveal a clear frequency effect. At the 5-second frequency, AJJ is uniformly more powerful than LM across the four noisy designs, and CC is typically on or very near the upper envelope. At the 10-second frequency, however, this ranking reverses: LM becomes the stronger marginal procedure, while CC again remains close to the best-performing test. The natural interpretation is that AJJ benefits more from ultra-high-frequency information and is therefore better able to exploit dense alternatives when the sampling grid is very fine. Once the grid is coarsened, that advantage weakens substantially. The ranking is stable across Gaussian and standardized-\(t_8\) noise. As in the benchmark dense design, power is highest for small \(\theta\) and then declines as \(\theta\) increases, because after local normalization the dense jump component becomes progressively harder to distinguish from a smooth continuous perturbation.

\paragraph{Sparse alternatives.}
We finally consider the sparse finite-activity alternative under the same noisy observation scheme. The latent price now contains a compound Poisson jump component
\[
J_t=\sum_{m=1}^{N_t}\xi_m,
\qquad
N_t\sim \mathrm{Poisson}(\lambda t),
\qquad
\xi_m\stackrel{iid}{\sim}N(0,0.05),
\]
and we report 5\%-level power for \(\lambda\in\{0.5,1,1.5,2,2.5\}\). Figures~\ref{fig:sparse-noise-dt5} and \ref{fig:sparse-noise-dt10} deliver a much more stable ordering. LM is the most powerful marginal procedure throughout, AJJ is the weakest, and CC lies between the two but remains much closer to LM than to AJJ. This is exactly what one would expect in a sparse-jump environment, where detection depends on a relatively small number of isolated discontinuities rather than on the accumulation of many small jump increments. As \(\lambda\) increases, all three procedures become more powerful, but the ranking does not change. The same pattern survives the move from Gaussian to standardized-\(t_8\) noise and from 5-second to 10-second sampling, although the coarser frequency shifts all power curves downward.

Taken together, the microstructure-noise experiments point to a simple conclusion. Once LM is calibrated by the double bootstrap, size is satisfactory for AJJ, LM, and CC. In terms of power, AJJ has an advantage in dense alternatives only at very high sampling frequencies, whereas LM is preferable in sparse alternatives and becomes relatively stronger as the sampling interval increases. CC provides the best overall compromise: it closely tracks the better marginal test across designs and is therefore the most reliable default choice in the presence of microstructure noise.

\section{Real data application}
\label{sec:data}

We illustrate the noise-robust procedures on public intraday data from the Kaggle dataset
\emph{Intraday market data} maintained by brtnsmth and the companion data-collection repository
\emph{TOS-RTD-core}. According to the dataset description, the files are generated from the
Thinkorswim real-time data interface, updated at about three-second intervals, written as one CSV file
per day, and captured over broad trading windows that include futures, stocks, and ETFs rather than
only the regular cash session \citep{brtnsmthKaggleIntraday2020,brtnsmthTOSRTD2020}. This
provenance is important for the empirical design. The raw files should be viewed as platform-captured
intraday market snapshots rather than as pre-filtered exchange transaction records. We therefore restrict
attention to the regular U.S. cash session from 09:30:00 to 16:00:00, retain the raw 3-second sampling
frequency, and remove sessions whose intraday path is entirely flat. These choices harmonize the sample
across assets and reduce the influence of overnight futures trading, stale snapshots, and inactive
periods.

Within our local archive we focus on six highly traded equity instruments,
\[
\mathrm{SPY},\qquad \mathrm{QQQ},\qquad \mathrm{IWM},\qquad \mathrm{AAPL},\qquad \mathrm{MSFT},\qquad \mathrm{NVDA}.
\]
The sample runs from January 27, 2020 to April 2, 2026. After restricting to regular trading hours,
there are $1540$ candidate weekday sessions for each asset. We remove $62$ sessions whose intraday
path is constant over the full day, leaving $1478$ valid trading days per asset. Each trading day is
then treated as one testing unit. For every day and every asset we compute the daily $p$-values of the
pre-averaged ratio test $\mathrm{AJJ}$, the Lee--Mykland-type local-average maximum test
$\mathrm{LM}$, and their Cauchy combination $\mathrm{CC}$.

A first empirical question is whether microstructure noise is quantitatively relevant at the 3-second
frequency. Figure~\ref{fig:realdata-qhat} plots the daily estimates $\hat q$ used in the LM
standardization. The mean values of $\hat q$ are $7.60\times 10^{-5}$ for SPY, $1.00\times 10^{-4}$
for QQQ, $1.13\times 10^{-4}$ for IWM, $1.33\times 10^{-4}$ for AAPL, $1.26\times 10^{-4}$ for
MSFT, and $2.31\times 10^{-4}$ for NVDA. These magnitudes are small in absolute terms, but they are
not negligible relative to 3-second returns. This is precisely the frequency range in which a
noise-robust implementation is preferable to the frictionless counterparts.

\begin{figure}[htbp]
\centering
\IfFileExists{alltime_3s_qhat_hist_panels_3x2.pdf}{%
\includegraphics[width=0.85\textwidth]{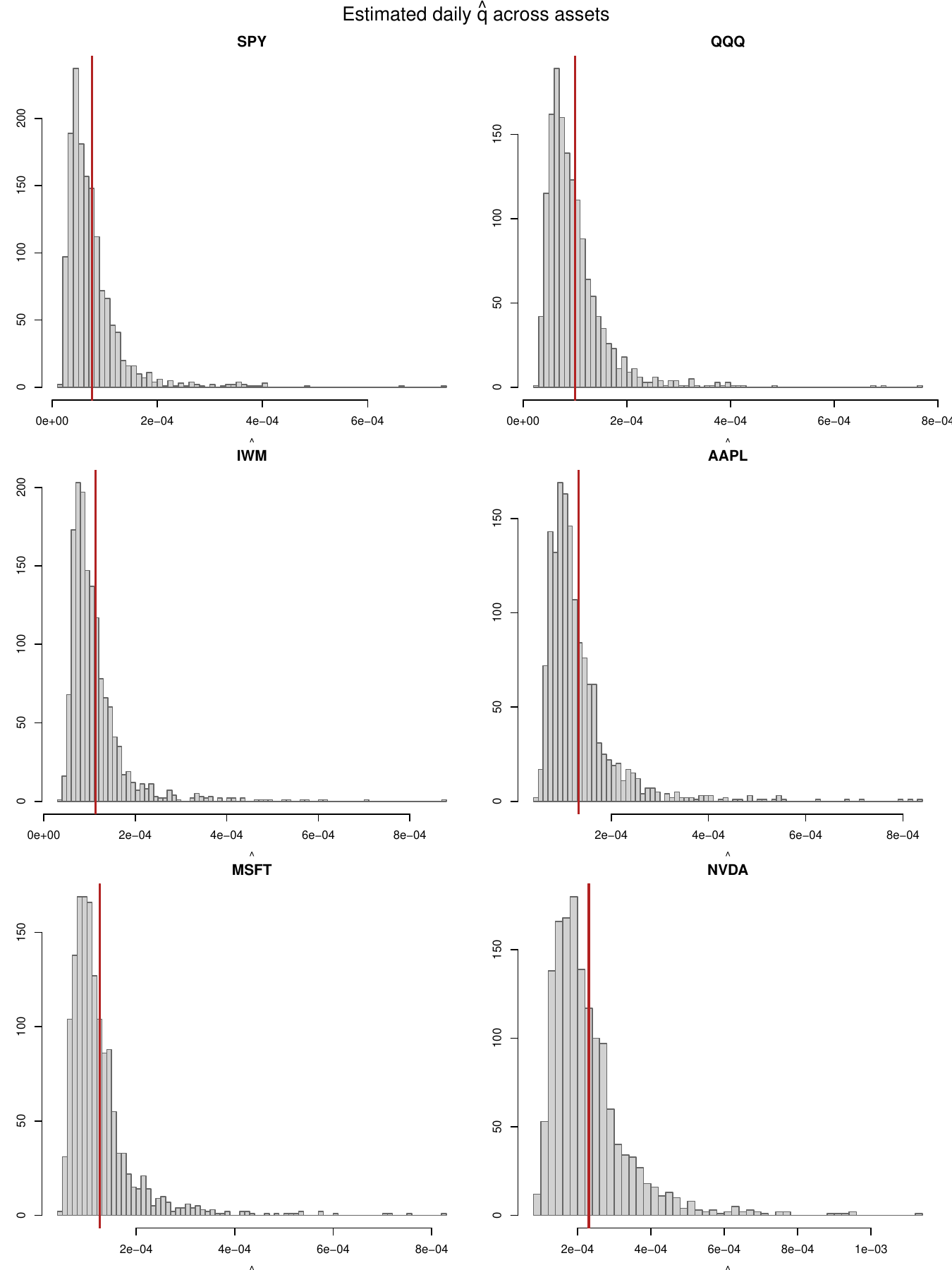}%
}{%
\fbox{\parbox{0.82\textwidth}{\centering Placeholder for the histogram panel of the daily noise estimates $\hat q$ for SPY, QQQ, IWM, AAPL, MSFT, and NVDA.}}%
}
\caption{Histograms of the daily noise estimates $\hat q$ for SPY, QQQ, IWM, AAPL, MSFT, and NVDA.}
\label{fig:realdata-qhat}
\end{figure}

Table~\ref{tab:realdata-rejection} reports the empirical rejection frequencies at the $5\%$ level. The
ranking is stable across all six assets: $\mathrm{LM}>\mathrm{CC}\gg \mathrm{AJJ}$. The LM test rejects
between $51.0\%$ and $61.5\%$ of daily sessions, whereas the pre-averaged ratio test rejects only
between $0.3\%$ and $2.0\%$. The combined test lies systematically between the two components and is
always much closer to LM than to AJJ. This pattern mirrors the sparse-regime logic of the Monte Carlo
section. In the real data, jump evidence appears to be concentrated in a few short intraday windows,
for which a local extreme-value statistic is more sensitive than a statistic that aggregates information
over the whole day. At the same time, the Cauchy combination substantially improves upon the AJJ
component and preserves most of the sparse-regime sensitivity of LM, which is exactly the adaptive gain
our theory is designed to deliver.

\begin{table}[htbp]
\centering
\caption{Empirical rejection frequencies at the $5\%$ level}
\label{tab:realdata-rejection}
\begin{tabular}{lccc}
\toprule
Asset & AJJ & LM & CC \\
\midrule
SPY  & 0.0203 & 0.6150 & 0.5419 \\
QQQ  & 0.0149 & 0.5710 & 0.5061 \\
IWM  & 0.0074 & 0.5169 & 0.4290 \\
AAPL & 0.0074 & 0.5724 & 0.4817 \\
MSFT & 0.0034 & 0.5792 & 0.4973 \\
NVDA & 0.0061 & 0.5101 & 0.4154 \\
\bottomrule
\end{tabular}
\end{table}

To isolate the days with stronger evidence of discontinuities, we apply the Benjamini--Hochberg
procedure to the LM daily $p$-values for each asset separately at false-discovery-rate level $0.2$.
The numbers of selected days are $1172$ for SPY, $1113$ for QQQ, $1037$ for IWM, $1132$ for AAPL,
$1136$ for MSFT, and $1023$ for NVDA. There are $449$ dates that are selected simultaneously for all
six assets. This concentration of detections across the full cross-section suggests that an important
fraction of the discontinuities is market-wide rather than idiosyncratic.

Figures~\ref{fig:realdata-20200309} and \ref{fig:realdata-20221110} display two representative common-
jump days, March 9, 2020 and November 10, 2022. For all six assets, the largest standardized LM
excursion occurs within seconds of the market open. This timing is economically plausible in view of
the source construction of the data: the underlying files are not pre-trimmed to cash-market hours,
so the opening auction and the first minutes of regular trading provide a natural point at which
overnight information is incorporated discontinuously into equity prices. The empirical evidence thus
reinforces the message of the simulation study. Noise should not be ignored at the 3-second frequency,
and the dominant alternatives in these data appear closer to sparse common jumps than to dense diffuse
jump activity.

\begin{figure}[htbp]
\centering
\IfFileExists{intraday_jump_panels_20200309_3x2.pdf}{%
\includegraphics[width=0.9\textwidth]{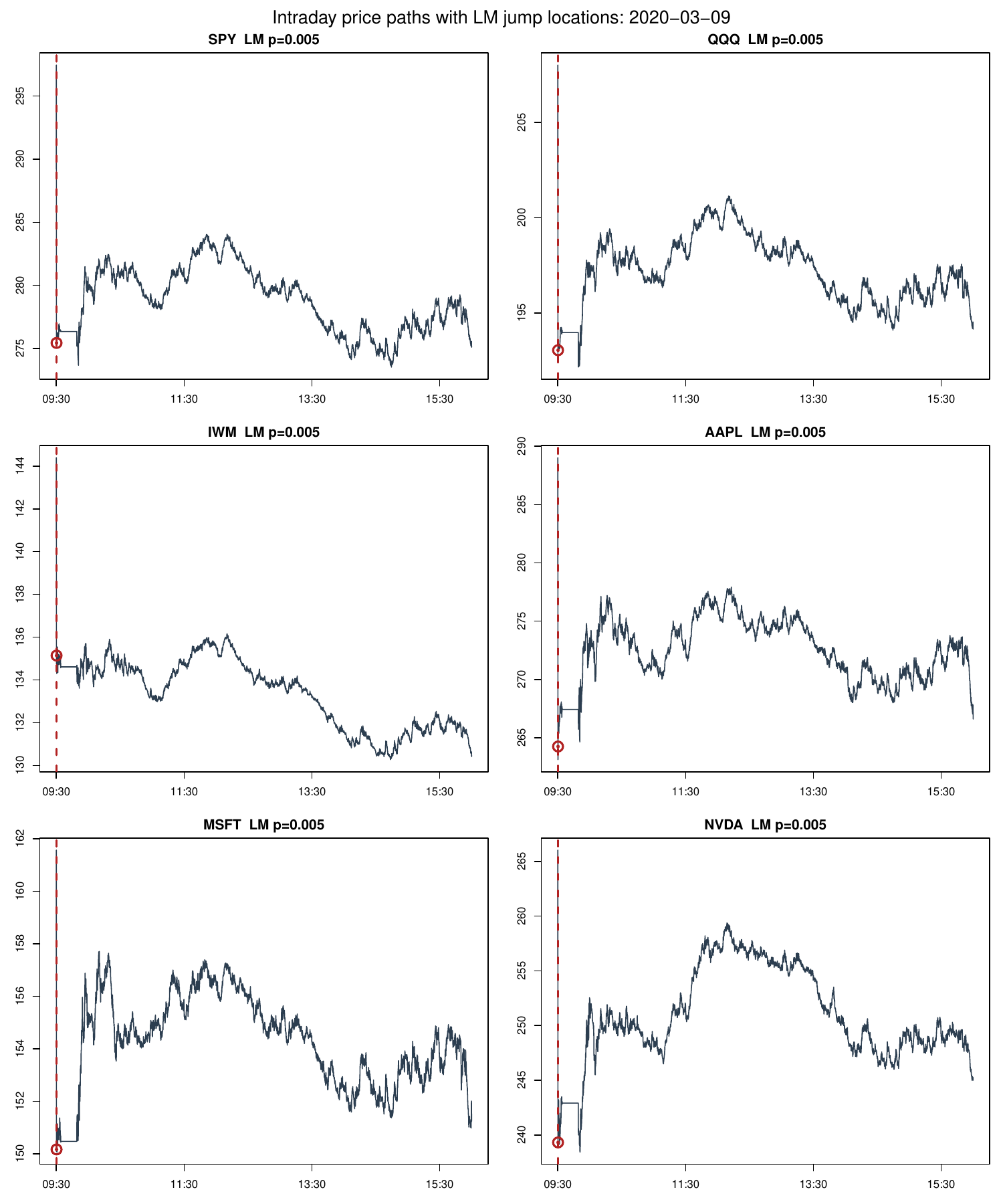}%
}{%
\fbox{\parbox{0.86\textwidth}{\centering Placeholder for the intraday price panels on March 9, 2020, with the LM-based jump location marked for each asset.}}%
}
\caption{Intraday price paths on March 9, 2020, with the LM-based jump location marked for each asset.}
\label{fig:realdata-20200309}
\end{figure}

\begin{figure}[htbp]
\centering
\IfFileExists{intraday_jump_panels_20221110_3x2.pdf}{%
\includegraphics[width=0.9\textwidth]{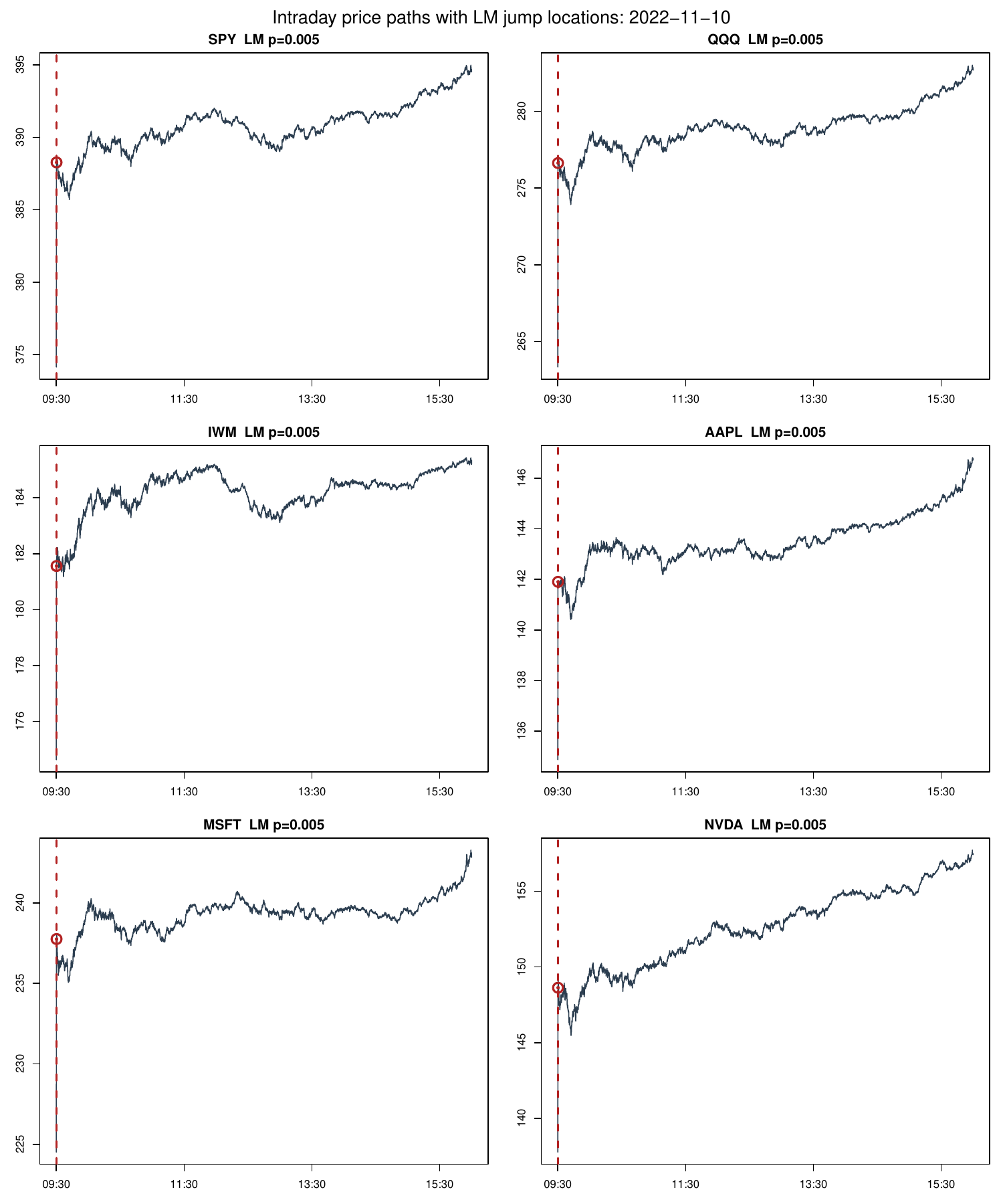}%
}{%
\fbox{\parbox{0.86\textwidth}{\centering Placeholder for the intraday price panels on November 10, 2022, with the LM-based jump location marked for each asset.}}%
}
\caption{Intraday price paths on November 10, 2022, with the LM-based jump location marked for each asset.}
\label{fig:realdata-20221110}
\end{figure}

We next revisit the empirical analysis at coarser sampling frequencies, where the effect of market microstructure noise is substantially attenuated. A large literature shows that noise is most severe at very high sampling rates and becomes much less important after temporal aggregation or sparse sampling; see, among others, \citet{AitSahaliaMyklandZhang2005,HansenLunde2006,ZuBoswijk2014,AitSahaliaJacodLi2012}. For this reason, starting from the cleaned regular-session \(3\)-second data for the six assets \(\{\text{SPY},\text{QQQ},\text{IWM},\text{AAPL},\text{MSFT},\text{NVDA}\}\), we construct \(1\)-minute, \(2\)-minute, and \(5\)-minute series by retaining the last available observation within each interval. We then apply the frictionless versions of the AJ, LM, and CC tests to the resulting intraday return sequences and compute daily rejection frequencies. Table~\ref{tab:clean-nonoise-05} reports the empirical rejection frequencies at the \(5\%\) level.

Two features stand out. First, rejection frequencies decline systematically as the sampling interval becomes coarser. This is consistent with the idea that temporal aggregation smooths short-lived intraday discontinuities and weakens the jump signal available to the tests. Second, the relative performance of AJ and LM varies with the frequency in a way that closely matches the adaptive motivation of the paper. At the \(1\)-minute frequency, LM remains competitive with AJ, which is consistent with a setting in which some discontinuities are still sufficiently sharp to be captured by a local extreme-value statistic. At the \(2\)-minute frequency, and even more clearly at the \(5\)-minute frequency, AJ becomes systematically stronger than LM, suggesting that once the data are sampled more sparsely the remaining jump evidence is less concentrated in a few extreme increments and is better captured by the ratio-based statistic.

The main message of Table~\ref{tab:clean-nonoise-05}, however, is the robustness of the Cauchy combination. At the \(1\)- and \(2\)-minute frequencies, CC is uniformly the strongest procedure across the six assets. At the \(5\)-minute frequency, although AJ slightly overtakes CC on average, the combined procedure remains very close to the stronger component and continues to dominate LM by a substantial margin for most assets. Hence the empirical evidence at coarser frequencies reinforces the central contribution of the paper: AJ is relatively more effective when jump evidence is more diffuse, LM is relatively more effective when jumps are sharper and more localized, and the proposed CC procedure combines these two sources of power and delivers the most stable overall performance across empirically relevant configurations.

\begin{table}[htbp]
\centering
\caption{Empirical rejection frequencies at the 5\% level for the no-noise AJ, LM, and CC tests}
\label{tab:clean-nonoise-05}
\begin{tabular}{lccccccccc}
\toprule
 & \multicolumn{3}{c}{1 min} & \multicolumn{3}{c}{2 min} & \multicolumn{3}{c}{5 min} \\
\cmidrule(lr){2-4}\cmidrule(lr){5-7}\cmidrule(lr){8-10}
Asset & AJ & LM & CC & AJ & LM & CC & AJ & LM & CC \\
\midrule
SPY & 0.332 & 0.366 & 0.431 & 0.287 & 0.221 & 0.334 & 0.193 & 0.112 & 0.191 \\
QQQ & 0.305 & 0.357 & 0.405 & 0.275 & 0.212 & 0.317 & 0.188 & 0.085 & 0.183 \\
IWM & 0.340 & 0.268 & 0.406 & 0.303 & 0.156 & 0.328 & 0.237 & 0.085 & 0.202 \\
AAPL & 0.392 & 0.388 & 0.491 & 0.340 & 0.229 & 0.384 & 0.250 & 0.131 & 0.251 \\
MSFT & 0.390 & 0.355 & 0.479 & 0.326 & 0.226 & 0.367 & 0.220 & 0.129 & 0.231 \\
NVDA & 0.305 & 0.258 & 0.377 & 0.261 & 0.167 & 0.296 & 0.168 & 0.078 & 0.151 \\
\bottomrule
\end{tabular}
\end{table}

\section{Conclusion}
\label{sec:conclusion}

This paper develops an adaptive jump test by combining a sum-type statistic with a max-type statistic. The key result is that the two components are asymptotically independent in both the latent and noisy settings, which makes an analytically calibrated Cauchy combination test feasible. The resulting procedure adapts well to dense and sparse alternatives and remains consistent against fixed finite-activity jumps, so it provides a simple and robust default when the form of the jump alternative is unknown.

Future research can proceed in several directions. One is to relax the i.i.d.\ noise assumption and allow serially dependent microstructure noise, building on \citet{AitSahaliaMyklandZhang2011} and \citet{LiLaevenVellekoop2020}. Another is to extend the adaptive idea to multivariate high-frequency data and study common-jump or co-jump testing under asynchronous sampling and noise; see \citet{JacodTodorov2009} and \citet{BibingerWinkelmann2015}.

\appendix
\makeatletter
\renewcommand\thetheorem{\Alph{section}.\arabic{theorem}}
\renewcommand\theproposition{\Alph{section}.\arabic{proposition}}
\renewcommand\thelemma{\Alph{section}.\arabic{lemma}}
\renewcommand\thecorollary{\Alph{section}.\arabic{corollary}}
\renewcommand\theassumption{\Alph{section}.\arabic{assumption}}
\renewcommand\thedefinition{\Alph{section}.\arabic{definition}}
\renewcommand\theremark{\Alph{section}.\arabic{remark}}
\numberwithin{equation}{section}
\renewcommand\theequation{\Alph{section}.\arabic{equation}}
\renewcommand\p@equation{}
\makeatother

\section{Auxiliary inequalities, local regularity, and stochastic-Taylor expansion}
\label{app:aux}

\begin{lemma}
\label{lem:LMregularity}
By Assumption \ref{ass:H0}, for every $\varepsilon\in(0,1/2)$,
\begin{equation*}
\max_{0\le i\le n-1}\sup_{t_i\le u\le t_{i+1}}\abs{b_u-b_{t_i}}
=O_p\paren{\Delta_n^{1/2-\varepsilon}},
\qquad
\max_{0\le i\le n-1}\sup_{t_i\le u\le t_{i+1}}\abs{\sigma_u-\sigma_{t_i}}
=O_p\paren{\Delta_n^{1/2-\varepsilon}}.
\end{equation*}
\end{lemma}

\begin{proof}
Fix $\varepsilon\in(0,1/2)$ and choose $q>1/\varepsilon$. Under Assumption \ref{ass:H0}, by the Burkholder--Davis--
Gundy inequality \citep[Chapter II]{Protter2005}, and the c\`adl\`ag local boundedness of
$(b^{(0)},\beta)$,
\begin{equation}\label{eq:bmodmom_v4}
\E\sup_{t_i\le u\le t_{i+1}}\abs{b_u-b_{t_i}}^q
\le C_q\Delta_n^{q/2},
\qquad i=0,\dots,n-1.
\end{equation}
Therefore
\begin{align*}
\Pbb\paren{
\max_{0\le i\le n-1}\sup_{t_i\le u\le t_{i+1}}\abs{b_u-b_{t_i}}>
\Delta_n^{1/2-\varepsilon}}
&\le
\sum_{i=0}^{n-1}
\frac{\E\sup_{t_i\le u\le t_{i+1}}\abs{b_u-b_{t_i}}^q}{\Delta_n^{q(1/2-\varepsilon)}}\\
&\le C_q n\Delta_n^{q\varepsilon}
= C_qT\Delta_n^{-1+q\varepsilon}
\to 0.
\end{align*}
Hence
\[
\max_{0\le i\le n-1}\sup_{t_i\le u\le t_{i+1}}\abs{b_u-b_{t_i}}
=O_p\paren{\Delta_n^{1/2-\varepsilon}}.
\]
The same argument applied to $\sigma$ gives
\[
\max_{0\le i\le n-1}\sup_{t_i\le u\le t_{i+1}}\abs{\sigma_u-\sigma_{t_i}}
=O_p\paren{\Delta_n^{1/2-\varepsilon}}.
\]
\end{proof}

\begin{lemma}[One-step stochastic-Taylor expansion under leverage]
\label{lem:taylor_increment}
For $i=1,\dots,n$, define
\[
E_i=\frac{\Delta_i^nW}{\sqrt{\Delta_n}},
\qquad
I_{i,n}^{(r,1)}=
\int_{t_{i-1}}^{t_i}\int_{t_{i-1}}^s\dd B_u^{(r)}\,\dd W_s,
\qquad r=2,\dots,d_B.
\]
Then $E_1,\dots,E_n$ are independent standard normal random variables, and under Assumption
\ref{ass:H0} there exist remainder variables $\rho_{i,n}$ such that
\begin{equation*}
\Delta_i^nX
=
\sigma_{t_{i-1}}\sqrt{\Delta_n}E_i
+b_{t_{i-1}}\Delta_n
+\frac12\gamma_{t_{i-1}}^{(1)}\paren{(\Delta_i^nW)^2-\Delta_n}
+\sum_{r=2}^{d_B}\gamma_{t_{i-1}}^{(r)}I_{i,n}^{(r,1)}
+\rho_{i,n}.
\end{equation*}
Moreover, for every $\varepsilon\in(0,1/2)$,
\begin{equation*}
\max_{1\le i\le n}\abs{\rho_{i,n}}=O_p\paren{\Delta_n^{3/2-\varepsilon}},
\qquad
\max_{1\le i\le n}
\abs{\Delta_i^nX-\sigma_{t_{i-1}}\sqrt{\Delta_n}E_i}
=O_p\paren{\Delta_n^{1-\varepsilon}}.
\end{equation*}
\end{lemma}

\begin{proof}
Under Assumption~\ref{ass:H0},
\[
\Delta_i^nX
=\int_{t_{i-1}}^{t_i}b_s\,\dd s
+\int_{t_{i-1}}^{t_i}\sigma_s\,\dd W_s.
\]
The drift term satisfies
\[
\int_{t_{i-1}}^{t_i}b_s\,\dd s
=b_{t_{i-1}}\Delta_n
+\int_{t_{i-1}}^{t_i}(b_s-b_{t_{i-1}})\,\dd s.
\]
Further,
\[
\int_{t_{i-1}}^{t_i}\sigma_s\,\dd W_s
=
\sigma_{t_{i-1}}\Delta_i^nW
+\int_{t_{i-1}}^{t_i}(\sigma_s-\sigma_{t_{i-1}})\,\dd W_s.
\]
Inserting the It\^o decomposition of $\sigma_s-\sigma_{t_{i-1}}$ gives
\begin{align*}
\int_{t_{i-1}}^{t_i}(\sigma_s-\sigma_{t_{i-1}})\,\dd W_s
&=
\int_{t_{i-1}}^{t_i}\int_{t_{i-1}}^s\sigma_u^{(0)}\,\dd u\,\dd W_s
+\sum_{r=1}^{d_B}\int_{t_{i-1}}^{t_i}\int_{t_{i-1}}^s\gamma_u^{(r)}\,\dd B_u^{(r)}\,\dd W_s
\\
&=
\sum_{r=1}^{d_B}\gamma_{t_{i-1}}^{(r)}
\int_{t_{i-1}}^{t_i}\int_{t_{i-1}}^s\dd B_u^{(r)}\,\dd W_s
+\widetilde\rho_{i,n},
\end{align*}
where
\begin{align*}
\widetilde\rho_{i,n}
&=
\int_{t_{i-1}}^{t_i}\int_{t_{i-1}}^s\sigma_u^{(0)}\,\dd u\,\dd W_s
+\sum_{r=1}^{d_B}
\int_{t_{i-1}}^{t_i}\int_{t_{i-1}}^s\paren{\gamma_u^{(r)}-\gamma_{t_{i-1}}^{(r)}}\,\dd B_u^{(r)}\,\dd W_s.
\end{align*}
For $r=1$,
\begin{equation*}
\int_{t_{i-1}}^{t_i}\int_{t_{i-1}}^s\dd W_u\,\dd W_s
=\frac12\paren{(\Delta_i^nW)^2-\Delta_n}.
\end{equation*}
Combining the previous displays yields
\begin{align*}
    \Delta_i^nX
=
\sigma_{t_{i-1}}\sqrt{\Delta_n}E_i
+b_{t_{i-1}}\Delta_n
+\frac12\gamma_{t_{i-1}}^{(1)}\paren{(\Delta_i^nW)^2-\Delta_n}
+\sum_{r=2}^{d_B}\gamma_{t_{i-1}}^{(r)}I_{i,n}^{(r,1)}
+\rho_{i,n},
\end{align*} 
with
\begin{equation*}
    \rho_{i,n}
=
\int_{t_{i-1}}^{t_i}(b_s-b_{t_{i-1}})\,\dd s
+\widetilde\rho_{i,n}.
\end{equation*}

Lemma \ref{lem:LMregularity}, the Cauchy--Schwarz inequality, and the Burkholder--Davis--Gundy
inequality imply, for every $\varepsilon\in(0,1/2)$,
\begin{align*}
\max_{1\le i\le n}\abs{\int_{t_{i-1}}^{t_i}(b_s-b_{t_{i-1}})\,\dd s}
&\le
\Delta_n\max_{0\le i\le n-1}\sup_{t_i\le s\le t_{i+1}}\abs{b_s-b_{t_i}}
=O_p\paren{\Delta_n^{3/2-\varepsilon}},
\\
\max_{1\le i\le n}
\abs{\int_{t_{i-1}}^{t_i}\int_{t_{i-1}}^s\sigma_u^{(0)}\,\dd u\,\dd W_s}
&=O_p\paren{\Delta_n^{3/2-\varepsilon}},
\\
\max_{1\le i\le n}
\abs{\int_{t_{i-1}}^{t_i}\int_{t_{i-1}}^s\paren{\gamma_u^{(r)}-\gamma_{t_{i-1}}^{(r)}}\,\dd B_u^{(r)}\,\dd W_s}
&=O_p\paren{\Delta_n^{3/2-\varepsilon}},
\qquad r=1,\dots,d_B.
\end{align*}
Therefore
\[
\max_{1\le i\le n}\abs{\rho_{i,n}}=O_p\paren{\Delta_n^{3/2-\varepsilon}}.
\]
Since
\[
\begin{aligned}
\max_{1\le i\le n}\abs{b_{t_{i-1}}\Delta_n}&=O_p(\Delta_n),\\
\max_{1\le i\le n}\abs{(\Delta_i^nW)^2-\Delta_n}
&=O_p\paren{\Delta_n\log n},\\
\max_{1\le i\le n}\abs{I_{i,n}^{(r,1)}}&=O_p\paren{\Delta_n\log n}.
\end{aligned}
\]
and $\Delta_n\log n=o(\Delta_n^{1-\varepsilon})$, we derive the remaining desired results.
\end{proof}

\begin{lemma}[Analytic inequalities]\label{lem:analytic}
Let $p>3$. Then there exists a finite constant $C_p$ such that for all $x_1,\dots,x_k,y_1,\dots,y_k\in\R$,
\begin{equation*}
\begin{aligned}
&\abs{U_{p,k}(x_1+y_1,\dots,x_k+y_k)-U_{p,k}(x_1,\dots,x_k)}\\
&\qquad\le C_p
\paren{\sum_{\ell=1}^k\abs{x_\ell}^{p-1}+\sum_{\ell=1}^k\abs{y_\ell}^{p-1}}
\sum_{\ell=1}^k\abs{y_\ell}.
\end{aligned}
\end{equation*}
\end{lemma}

\begin{proof}
The function $f(x)=\abs{x}^p$ is $C^2$ on $\R\setminus\{0\}$. By the Mean Value Theorem and the fact that
\begin{equation*}
    |x+ty|^{p-1}\le C_p(|x|^{p-1}+|y|^{p-1}),\qquad 0\le t\le 1,
\end{equation*}
we have
\begin{align}\label{eq:taylorineq1}
    |f(x+y)-f(x)|
    \le p|y|\int_0^1|x+ty|^{p-1}\dd t\le C_p(|x|^{p-1}+|y|^{p-1})|y|.
\end{align}
Applying \eqref{eq:taylorineq1} once to the block sum and once to
each coordinate term in \eqref{eq:Upk} yields
\begin{align*}
    &\abs{U_{p,k}(x_1+y_1,\dots,x_k+y_k)-U_{p,k}(x_1,\dots,x_k)}\\
    \le&C_p\left(\left|\sum_{l=1}^kx_l\right|^{p-1}+\left|\sum_{l=1}^ky_l\right|^{p-1}\right)\left|\sum_{l=1}^ky_l\right|+C_p\sum_{l=1}^k(|x_l|^{p-1}+|y_l|^{p-1})|y_l|,
\end{align*}
which, combining with the fact $\left|\sum_{l=1}^kx_l\right|^{p-1}\le C_{p,k}\sum_{l=1}^k|x_l|^{p-1}$ and $\left|\sum_{l=1}^ky_l\right|\le \sum_{l=1}^k|y_l|$, gives the desired result.
\end{proof}

\section{Proofs under the continuous-path null}
\label{app:null}

\begin{lemma}[Block martingale expansion for the sum statistic]\label{lem:AJexpansionH0}
Under Assumption \ref{ass:H0},
\begin{equation*}
\bar Z_n^{AJ}:=\frac{\Delta_n^{-1/2}\paren{R_n^{AJ}-k^{p/2-1}}}{\tau_0}
=\sum_{j=1}^{m_n}\omega_{j,n}\Psi_0\paren{E_{j,1},\dots,E_{j,k}}+o_p(1),
\end{equation*}
where
\begin{equation*}
\Psi_0(\mathbf e)=\frac{U_{p,k}(\mathbf e)}{\varsigma_{p,k}},
\qquad
\omega_{j,n}=\frac{\sqrt{k\Delta_n}\,\sigma_{(j-1)k\Delta_n}^p}{\sqrt{A_{2p}}},
\qquad
E_{j,\ell}=E_{(j-1)k+\ell}.
\end{equation*}
In particular,
\begin{equation*}
\sum_{j=1}^{m_n}\omega_{j,n}^2\toP 1,
\qquad
\max_{1\le j\le m_n}\abs{\omega_{j,n}}\toP 0.
\end{equation*}
Moreover,
\begin{equation*}
Z_n^{AJ}-\bar Z_n^{AJ}\toP 0.
\end{equation*}
\end{lemma}

\begin{proof}
By Lemma \ref{lem:taylor_increment},
\begin{equation*}
\Delta_{(j-1)k+\ell}^nX
=
\sigma_{t_{(j-1)k}}\sqrt{\Delta_n}E_{j,\ell}+\chi_{j,\ell,n},
\end{equation*}
where
\begin{align}\label{eq:block_increment_leverage}
\chi_{j,\ell,n}
=&(\sigma_{t_{(j-1)k+\ell-1}}-\sigma_{t_{(j-1)k}})\sqrt{\Delta_n}E_{j,\ell}
+b_{t_{(j-1)k+\ell-1}}\Delta_n
+\frac12\gamma_{t_{(j-1)k+\ell-1}}^{(1)}\Delta_n(E_{j,\ell}^2-1)\nonumber\\
&+\sum_{r=2}^{d_B}\gamma_{t_{(j-1)k+\ell-1}}^{(r)}I_{(j-1)k+\ell,n}^{(r,1)}
+\rho_{(j-1)k+\ell,n}.
\end{align}
For $j=1,\dots,m_n$, define 
\begin{align*}
G_{j,n}
:=&U_{p,k}(\Delta_{(j-1)k+1}^nX,\dots,\Delta_{jk}^nX)\\
=&\Delta_n^{p/2}\sigma_{t_{(j-1)k}}^pU_{p,k}(E_{j,1}+\Delta_n^{-1/2}\sigma_{t_{(j-1)k}}^{-1}\chi_{j,1,n},\dots,E_{j,k}+\Delta_n^{-1/2}\sigma_{t_{(j-1)k}}^{-1}\chi_{j,k,n})\\
H_{j,n}:=&\Delta_n^{p/2}\sigma_{t_{(j-1)k}}^pU_{p,k}(E_{j,1},\dots,E_{j,k}).
\end{align*}
Lemma \ref{lem:analytic} yields
\begin{equation*}
\abs{G_{j,n}-H_{j,n}}
\le
C\Delta_n^{(p-1)/2}
\paren{\sum_{\ell=1}^k\abs{E_{j,\ell}}^{p-1}+\Delta_n^{-(p-1)/2}\sum_{\ell=1}^k\abs{\chi_{j,\ell,n}}^{p-1}}
\sum_{\ell=1}^k\abs{\chi_{j,\ell,n}},
\end{equation*}
which, noting that the perturbations $\chi_{j,\ell,n}$ are of order $\Delta_n$ and the Gaussian moments of
$(E_{j,1},\dots,E_{j,k})$ are bounded, implies
\begin{equation}\label{eq:G-Hvariation}
\sum_{j=1}^{m_n}\E\abs{G_{j,n}-H_{j,n}}^2
=O\paren{\Delta_n^{p}}.
\end{equation}

Let $\bm E_{j}=(E_{j,1},\dots,E_{j,k})^\top$ and $\bm\chi_{j}'=\Delta_n^{-1/2}\sigma_{t_{(j-1)k}}^{-1}(\chi_{j,1,n},\dots,\chi_{j,k,n})^\top$. Taylor expansion gives
\begin{align*}
    U_{p,k}(\bm E_j+\bm\chi_j')-U_{p,k}(\bm E_j)=&\sum_{\ell=1}^{k} \partial_\ell U_{p,k}(\bm E_j)\chi'_{j,\ell}+\frac12\sum_{\ell,m=1}^{k} \partial_{\ell m}U_{p,k}(\bm E_j+\theta\bm\chi_j')\chi'_{j,\ell}\chi'_{j,m}\\
    =&:T_{1,j}+T_{2,j},
\end{align*}
for some $\theta\in(0,1)$. Using the fact that $\chi'_{j,l},\chi'_{j,m}$ are of order $\Delta_n^{1/2}$ and standard Gaussian moment bounds, 
\begin{equation}\label{eq:T2_bound}
\E_{j-1}(T_{2,j}) = O_p(\Delta_n),
\end{equation}
where $\E_{j-1}(\cdot)=\E(\cdot\mid\mathcal{F}_{t_{(j-1)k}})$. 

Next, we analyze $\E_{j-1}(T_{1,j})$. We substitute $\chi'_{j,\ell}=\sum_{a=1}^5\chi'_{j,\ell,(a)}$ via \eqref{eq:block_increment_leverage}. 
For $a=1$, conditioning directly on $\mathcal{F}_{t_{(j-1)k}}$ fails to isolate the leverage effect. Instead, by It\^{o} formula for the semimartingale $\sigma$ from $t_{(j-1)k}$ to $t_{(j-1)k+\ell-1}$,
\begin{equation*}
\sigma_{t_{(j-1)k+\ell-1}}-\sigma_{t_{(j-1)k}} = \gamma^{(1)}_{t_{(j-1)k}}\sqrt{\Delta_n}\sum_{m=1}^{\ell-1}E_{j,m} + \sum_{r=2}^{d_B}\gamma^{(r)}_{t_{(j-1)k}}(B^{(r)}_{t_{(j-1)k+\ell-1}}-B^{(r)}_{t_{(j-1)k}}) + \zeta_{j,\ell,n},
\end{equation*}
where 
\begin{equation*}
\zeta_{j,\ell,n}=\int_{t_{(j-1)k}}^{t_{(j-1)k+\ell-1}}\sigma_s^{(0)}\dd s+\int_{t_{(j-1)k}}^{t_{(j-1)k+\ell-1}}(\gamma_s-\gamma_{t_{(j-1)k}})^\top\dd B_s.
\end{equation*}
Hence,
\begin{align*}
    \E_{j-1}\{\partial_\ell U_{p,k}(\bm E_j)\chi'_{j,\ell,(1)}\} 
    =&\sigma_{t_{(j-1)k}}^{-1}\gamma^{(1)}_{t_{(j-1)k}}\sqrt{\Delta_n}\E_{j-1}\Big\{\partial_\ell U_{p,k}(\bm E_j)E_{j,\ell}\sum_{m=1}^{\ell-1}E_{j,m}\Big\} \\
    &+\sigma_{t_{(j-1)k}}^{-1}\sum_{r=2}^{d_B}\gamma^{(r)}_{t_{(j-1)k}}\E_{j-1}\{\partial_\ell U_{p,k}(\bm E_j)E_{j,\ell}\Delta B^{(r)}\} \\
    &+\sigma_{t_{(j-1)k}}^{-1}\E_{j-1}\{\partial_\ell U_{p,k}(\bm E_j)E_{j,\ell}\zeta_{j,\ell,n}\}.
\end{align*}
By the parity of $U_{p,k}$, the function $\bm E \mapsto \partial_\ell U_{p,k}(\bm E)E_\ell \sum E_m$ is an odd function of $\bm E$; thus, its expectation with respect to the standard Gaussian vector exactly vanishes. This continuous-path cancellation zeroes out the leading term. The second term vanishes because $B^{(r)}\ind W$. The third term is bounded by $C\E_{j-1}[\zeta_{j,\ell,n}^2]^{1/2} = o_p(\Delta_n^{1/2})$.

Similarly, for $a=2$, we approximate $b_{t_{(j-1)k+\ell-1}}$ by $b_{t_{(j-1)k}}$. Since $\partial_\ell U_{p,k}$ is odd, the main term $\sigma_{t_{(j-1)k}}^{-1}\Delta_n^{1/2}b_{t_{(j-1)k}}\E\{\partial_\ell U_{p,k}(\bm E_j)\} = 0$. The approximation error is bounded by $C\Delta_n^{1/2}\E_{j-1}(b_{t_{(j-1)k+\ell-1}}-b_{t_{(j-1)k}}) = o_p(\Delta_n^{1/2})$. 
For $a=3$, substituting $\gamma^{(1)}_{t_{(j-1)k+\ell-1}}$ with $\gamma^{(1)}_{t_{(j-1)k}}$, the leading term again vanishes because $\bm E \mapsto \partial_\ell U_{p,k}(\bm E)(E_\ell^2-1)$ is an odd function. The error is bounded by $C\Delta_n^{1/2}\E_{j-1}(\gamma^{(1)}_{t_{(j-1)k+\ell-1}}-\gamma^{(1)}_{t_{(j-1)k}}) = o_p(\Delta_n^{1/2})$. 
For $a=4$, approximating $\gamma^{(r)}_{t_{(j-1)k+\ell-1}}$ by $\gamma^{(r)}_{t_{(j-1)k}}$, the main term vanishes since $\E(I_{(j-1)k+\ell,n}^{(r,1)}\mid W)=0$. The error is bounded by Cauchy-Schwarz: $C\Delta_n^{-1/2}\E_{j-1}\{(\gamma^{(r)}_{t_{(j-1)k+\ell-1}}-\gamma^{(r)}_{t_{(j-1)k}})^2\}^{1/2}\E_{j-1}\{(I^{(r,1)})^2\}^{1/2} = \Delta_n^{-1/2}o_p(1)O_p(\Delta_n) = o_p(\Delta_n^{1/2})$. 
For $a=5$, $\chi'_{j,\ell,(5)} = O_p(\Delta_n)$, contributing $O_p(\Delta_n)$. 
In summary,
\begin{equation}\label{eq:T1_bound}
    \E_{j-1}(T_{1,j}) = \sum_{\ell=1}^k\sum_{a=1}^5\E_{j-1}\{\partial_\ell U_{p,k}(\bm E_j)\chi'_{j,\ell,(a)}\} = o_p(\Delta_n^{1/2}).
\end{equation}

By \eqref{eq:T2_bound} and \eqref{eq:T1_bound},
\begin{equation*}
\sum_{j=1}^{m_n}\E_{j-1}(G_{j,n}-H_{j,n})
=m_n\Delta_n^{p/2} o_p(\Delta_n^{1/2})=o_p(\Delta_n^{p/2-1/2}),
\end{equation*}
which gives
\begin{align*}
\sum_{j=1}^{m_n}\paren{G_{j,n}-H_{j,n}}
=\sum_{j=1}^{m_n}\{(G_{j,n}-H_{j,n})-\E_{j-1}(G_{j,n}-H_{j,n})\}+o_p(\Delta_n^{p/2-1/2})=o_p\paren{\Delta_n^{p/2-1/2}},
\end{align*}
where the last step follows from the fact that the martingale has predictable quadratic variation of order $O(\Delta_n^p)$ by \eqref{eq:G-Hvariation}. Hence
\begin{equation}\label{eq:numexpandnull_leverage}
B_n^{(k)}(p)-k^{p/2-1}B_n(p)
=
\Delta_n^{p/2}
\sum_{j=1}^{m_n}\sigma_{t_{(j-1)k}}^pU_{p,k}(E_{j,1},\dots,E_{j,k})
+o_p\paren{\Delta_n^{p/2-1/2}}.
\end{equation}
For the denominator, Lemma \ref{lem:taylor_increment} and \eqref{eq:taylorineq1} imply
\begin{align}
\sum_{i=1}^n\abs{\Delta_i^nX}^p
&=
\sum_{i=1}^n\abs{\sigma_{t_{i-1}}\sqrt{\Delta_n}E_i}^p
+o_p\paren{\Delta_n^{p/2-1}}
\nonumber\\
&=
\Delta_n^{p/2}\sum_{i=1}^n|\sigma_{t_{i-1}}|^p\abs{E_i}^p
+o_p\paren{\Delta_n^{p/2-1}}
\nonumber\\
&=
m_pA_p\Delta_n^{p/2-1}+o_p\paren{\Delta_n^{p/2-1}},
\label{eq:BnLLN_leverage}
\end{align}
where the last step follows from the martingale-difference decomposition
\[
\Delta_n^{p/2}\sum_{i=1}^n\sigma_{t_{i-1}}^p\paren{\abs{E_i}^p-m_p}
+
 m_p\Delta_n^{p/2}\sum_{i=1}^n\sigma_{t_{i-1}}^p,
\]
the Riemann approximation of $A_p$, and the variance bound
\[
\Var\paren{\Delta_n^{p/2}\sum_{i=1}^n\sigma_{t_{i-1}}^p\paren{\abs{E_i}^p-m_p}}
=O\paren{\Delta_n^{p-1}}.
\]
Combining \eqref{eq:numexpandnull_leverage} and \eqref{eq:BnLLN_leverage},
\begin{equation*}
\Delta_n^{-1/2}\paren{R_n^{AJ}-k^{p/2-1}}
=
\frac{\sqrt{\Delta_n}}{m_pA_p}
\sum_{j=1}^{m_n}\sigma_{t_{(j-1)k}}^pU_{p,k}(E_{j,1},\dots,E_{j,k})
+o_p(1).
\end{equation*}
Dividing by $\tau_0$ and using \eqref{eq:tau0},
\[
\bar Z_n^{AJ}
=
\sum_{j=1}^{m_n}\frac{\sqrt{k\Delta_n}\,\sigma_{t_{(j-1)k}}^p}{\sqrt{A_{2p}}}
\frac{U_{p,k}(E_{j,1},\dots,E_{j,k})}{\varsigma_{p,k}}+o_p(1).
\]
Moreover,
\[
\sum_{j=1}^{m_n}\omega_{j,n}^2
=
\frac{k\Delta_n\sum_{j=1}^{m_n}\sigma_{(j-1)k\Delta_n}^{2p}}{A_{2p}}\toP 1,
\qquad
\max_{1\le j\le m_n}\abs{\omega_{j,n}}
\le
\frac{\overline\sigma^p\sqrt{k\Delta_n}}{\sqrt{A_{2p}}}\to 0.
\]
The same denominator argument with $r=p$ and $r=2p$ gives
\[
\hat A_{p,n}\toP A_p,
\qquad
\hat A_{2p,n}\toP A_{2p},
\qquad
\hat\tau_{n,0}\toP\tau_0.
\]
Hence, by Slutsky's theorem, and the fact that $\bar{Z}^{AJ}=O_p(1)$ shown in Lemma \ref{lem:finiteremove}, $Z_n^{AJ}-\bar Z_n^{AJ}\toP 0$.
\end{proof}

\begin{lemma}[Leverage-robust Gaussian approximation for the max statistic]\label{lem:LMgauss}
Under Assumption \ref{ass:H0} and \eqref{eq:Kncond},
\begin{equation*}
\max_{K_n\le i\le n}
\abs{\frac{\hat V_{n,i}}{\sigma_{t_{i-1}}^2\Delta_n}-1}
=o_p\paren{\an^2}.
\end{equation*}
Consequently,
\begin{equation*}
\xi_n-\tilde\xi_n\toP 0,
\qquad
\tilde\xi_n=\frac{\max_{K_n\le i\le n}\abs{E_i}-C_n}{\an}.
\end{equation*}
\end{lemma}

\begin{proof}
Let
\[
R_{i,n}=\Delta_i^nX-\sigma_{t_{i-1}}\sqrt{\Delta_n}E_i.
\]
Lemma \ref{lem:taylor_increment} yields, for every $\varepsilon\in(0,1/2)$,
\begin{equation}\label{eq:Ribound_max}
\max_{1\le i\le n}\abs{R_{i,n}}=O_p\paren{\Delta_n^{1-\varepsilon}}.
\end{equation}
For $i=K_n,\dots,n$, write
\[
\hat V_{n,i}=A_{n,i}+B_{n,i},
\]
where
\begin{align*}
A_{n,i}
&=
\frac{\pi\Delta_n}{2(K_n-1)}\sum_{j=i-K_n+2}^{i}
\sigma_{t_{j-2}}^2\abs{E_j}\abs{E_{j-1}},\\
B_{n,i}
&=
\hat V_{n,i}-A_{n,i}.
\end{align*}
Using
\[
\Delta_j^nX
=\sigma_{t_{j-1}}\sqrt{\Delta_n}E_j+R_{j,n},
\qquad
\Delta_{j-1}^nX
=\sigma_{t_{j-2}}\sqrt{\Delta_n}E_{j-1}+R_{j-1,n},
\]
and the inequality
\[
\abs{\abs{a+b}\abs{c+d}-\abs{a}\abs{c}}
\le \abs{b}\abs{c+d}+\abs{d}\abs{a},
\]
we obtain
\begin{align}\label{eq:Bni_leverage}
\abs{B_{n,i}}
&\le
\frac{C}{K_n}\sum_{j=i-K_n+2}^{i}
\Big(
\Delta_n\abs{\sigma_{t_{j-1}}-\sigma_{t_{j-2}}}\abs{E_jE_{j-1}}
+\sqrt{\Delta_n}\abs{R_{j,n}}\abs{E_{j-1}}
+\sqrt{\Delta_n}\abs{R_{j-1,n}}\abs{E_j}
\Big).
\end{align}
We use \eqref{eq:Ribound_max} only with a fixed small exponent. Choose once and for all
$\varepsilon\in(0,1/6)$. Since
\[
\max_{K_n\le i\le n}\frac1{K_n}\sum_{j=i-K_n+2}^{i}\abs{E_j}
=O_p(1)
\]
by the same Bernstein--union-bound argument used below for $Z_j$, \eqref{eq:Bni_leverage} gives
\begin{align*}
\max_{K_n\le i\le n}\abs{B_{n,i}}
&=O_p\paren{\sqrt{\Delta_n}\,\Delta_n^{1-\varepsilon}}
  +O_p\paren{\Delta_n^{2-2\varepsilon}}  \\
&=O_p\paren{\Delta_n^{3/2-\varepsilon}}.
\end{align*}
Because $\varepsilon<1/2$ already implies
$\Delta_n^{1/2-\varepsilon}=o((\log n)^{-1})$, this is more than enough for the present normalization:
\begin{equation}\label{eq:Bni_leverage2}
\max_{K_n\le i\le n}\abs{B_{n,i}}=o_p\paren{\Delta_n\an^2}.
\end{equation}
The stronger restriction $\varepsilon<1/6$ will also be used after division by
$\sqrt{\hat V_{n,i}}$, where the error has size
$O_p(\Delta_n^{1/2-\varepsilon})=o_p(a_n)$. Next,
\[
\frac{A_{n,i}}{\sigma_{t_{i-1}}^2\Delta_n}-1=D_{n,i}^{(1)}+D_{n,i}^{(2)},
\]
where
\begin{align*}
D_{n,i}^{(1)}
&=
\frac{1}{K_n-1}\sum_{j=i-K_n+2}^{i}
\paren{\frac{\sigma_{t_{j-2}}^2}{\sigma_{t_{i-1}}^2}-1},\\
D_{n,i}^{(2)}
&=
\frac{\pi}{2(K_n-1)}\sum_{j=i-K_n+2}^{i}
\frac{\sigma_{t_{j-2}}^2}{\sigma_{t_{i-1}}^2}
\paren{\abs{E_j}\abs{E_{j-1}}-\frac{2}{\pi}}.
\end{align*}
Lemma \ref{lem:LMregularity} implies
\begin{equation}\label{eq:D1_leverage}
\max_{K_n\le i\le n}\abs{D_{n,i}^{(1)}}
\le C\sup_{\abs{u-v}\le K_n\Delta_n}\abs{\sigma_u-\sigma_v}
=o_p\paren{\an^2}.
\end{equation}
Let
\[
Z_j=\abs{E_j}\abs{E_{j-1}}-\frac{2}{\pi}.
\]
The sequence $(Z_j)$ is $1$-dependent and has exponential moments. Splitting into odd and even
subsequences and applying Bernstein's inequality to each subsequence \citep[Chapter 2]{Vershynin2018}
gives
\[
\Pbb\paren{\abs{\sum_{j=r}^{r+K_n-1}Z_j}>x}
\le 4\exp\paren{-c_1\min\set{\frac{x^2}{K_n},x}}
\]
for finite constants $c_1,c_2>0$. Taking $x=M\sqrt{K_n\log n}$ and using the union bound over the
$O(n)$ windows,
\begin{equation}\label{eq:D2_leverage}
\max_{K_n\le i\le n}
\abs{\frac{1}{K_n-1}\sum_{j=i-K_n+2}^{i}Z_j}
=O_p\paren{\sqrt{\frac{\log n}{K_n}}}
=o_p\paren{\an^2}.
\end{equation}
Since $\underline\sigma\le\sigma_t\le\overline\sigma$, \eqref{eq:D2_leverage} implies
\begin{equation}\label{eq:D2_leverage2}
\max_{K_n\le i\le n}\abs{D_{n,i}^{(2)}}=o_p\paren{\an^2}.
\end{equation}
Combining \eqref{eq:Bni_leverage2}, \eqref{eq:D1_leverage}, and \eqref{eq:D2_leverage2},
\begin{equation}\label{eq:Vconsistency}
\max_{K_n\le i\le n}
\abs{\frac{\hat V_{n,i}}{\sigma_{t_{i-1}}^2\Delta_n}-1}
=o_p\paren{\an^2}.    
\end{equation}
Finally,
\[
L_{n,i}-E_i
=
E_i\paren{\frac{\sigma_{t_{i-1}}\sqrt{\Delta_n}}{\sqrt{\hat V_{n,i}}}-1}
+\frac{R_{i,n}}{\sqrt{\hat V_{n,i}}}.
\]
The preceding choice $\varepsilon<1/6$ gives
\[
\max_i\frac{\abs{R_{i,n}}}{\sqrt{\Delta_n}}=O_p(\Delta_n^{1/2-\varepsilon})=o_p(\an),
\]
and \eqref{eq:Vconsistency}, together with $\max_{1\le i\le n}\abs{E_i}=O_p(1/\an)$, gives
\[
\max_i\abs{E_i}
\abs{\frac{\sigma_{t_{i-1}}\sqrt{\Delta_n}}{\sqrt{\hat V_{n,i}}}-1}
=O_p(\an^{-1})o_p(\an^2)=o_p(\an).
\]
Consequently
\[
\max_{K_n\le i\le n}\abs{L_{n,i}-E_i}=o_p(\an).
\]
Therefore
\[
\abs{\xi_n-\tilde\xi_n}
\le
\frac{\max_{K_n\le i\le n}\abs{L_{n,i}-E_i}}{\an}
\toP 0.
\]
\end{proof}

\begin{lemma}[Finite removal and stable Gaussian limit for the AJ block sum]\label{lem:finiteremove}
Let $\mathcal J_n\subset\{1,\dots,m_n\}$ be deterministic with $\sup_n\#\mathcal J_n<\infty$, and let
\[
\mathcal H_n=\sigma\set{E_{j,1},\dots,E_{j,k}:j\in\mathcal J_n}.
\]
Under Assumption \ref{ass:H0},
\begin{equation*}
\sum_{j\in\mathcal J_n}\omega_{j,n}\Psi_0(E_{j,1},\dots,E_{j,k})\toP 0,
\end{equation*}
and
\begin{equation*}
\sum_{j\notin\mathcal J_n}\omega_{j,n}\Psi_0(E_{j,1},\dots,E_{j,k})
\toD N(0,1)
\end{equation*}
stably with respect to $\mathcal H_n$.
\end{lemma}

\begin{proof}
Since $\Psi_0$ has mean zero, variance one, and finite $(2+\delta)$-moment for some $\delta>0$, 
\[
\E\abs{\sum_{j\in\mathcal J_n}\omega_{j,n}\Psi_0(E_{j,1},\dots,E_{j,k})}^2
=\sum_{j\in\mathcal J_n}\omega_{j,n}^2
\le \#\mathcal J_n\max_{1\le j\le m_n}\omega_{j,n}^2\toP 0,
\]
which proves 
\[\sum_{j\in\mathcal J_n}\omega_{j,n}\Psi_0(E_{j,1},\dots,E_{j,k})\toP 0.\]

For the complement, define the enlarged block filtration
\[
\F_{j,n}^{\mathcal J}=\mathcal H_n\vee\F_{jk\Delta_n},
\qquad j=0,1,\dots,m_n,
\]
and the array
\[
\zeta_{j,n}=\1\set{j\notin\mathcal J_n}\omega_{j,n}\Psi_0(E_{j,1},\dots,E_{j,k}).
\]
For $j\notin\mathcal J_n$, the weight $\omega_{j,n}$ is $\F_{j-1,n}^{\mathcal J}$-measurable, while the block
vector $(E_{j,1},\dots,E_{j,k})$ is independent of $\F_{j-1,n}^{\mathcal J}$. Hence
\[
\E(\zeta_{j,n}\mid \F_{j-1,n}^{\mathcal J})=0,
\qquad j=1,\dots,m_n,
\]
so $(\zeta_{j,n})$ is a martingale-difference array. Its predictable quadratic variation is
\[
\sum_{j=1}^{m_n}\E(\zeta_{j,n}^2\mid \F_{j-1,n}^{\mathcal J})
=\sum_{j\notin\mathcal J_n}\omega_{j,n}^2
=1+o_p(1)
\]
by Lemma \ref{lem:AJexpansionH0}. Further,
\[
\sum_{j=1}^{m_n}\E\abs{\zeta_{j,n}}^{2+\delta}
\le C\sum_{j=1}^{m_n}\abs{\omega_{j,n}}^{2+\delta}
\le C\max_{1\le j\le m_n}\abs{\omega_{j,n}}^{\delta}
\sum_{j=1}^{m_n}\omega_{j,n}^2
\toP 0.
\]
The stable martingale central limit theorem \citep[Chapter~3]{HallHeyde1980} therefore yields
\[\sum_{j\notin\mathcal J_n}\omega_{j,n}\Psi_0(E_{j,1},\dots,E_{j,k})
\toD N(0,1),\]
stably with respect to $\mathcal H_n$.
\end{proof}

\begin{proof}[Proof of Theorem \ref{thm:H0indep}]
By Lemmas \ref{lem:AJexpansionH0} and \ref{lem:LMgauss}, the vector difference between $(Z_n^{AJ},\xi_n)$ and its oracle counterpart converges to zero in probability. Slutsky's theorem in product spaces therefore reduces the joint convergence problem to the oracle pair. It is enough to prove
\begin{equation}\label{eq:H0oraclegoal}
\Pbb\paren{\mcS_n\le x,\ \tilde\xi_n\le y}\to \Phi(x)e^{-e^{-y}},
\end{equation}
where
\[
\mcS_n=\sum_{j=1}^{m_n}\omega_{j,n}\Psi_0(E_{j,1},\dots,E_{j,k}),\qquad
\tilde\xi_n=\frac{\max_{K_n\le i\le n}\abs{E_i}-C_n}{\an}.
\]

Fix $x,y\in\R$ and define
\[
\nu_n(y)=C_n+\an y,
\qquad
N_n(y)=\sum_{i=K_n}^{n}\1\set{\abs{E_i}>\nu_n(y)}.
\]
Then $\{\tilde\xi_n\le y\}=\{N_n(y)=0\}$. Let
\[
q_n(y)=\Pbb\paren{\abs{E_1}>\nu_n(y)}.
\]
Mills' ratio gives
\begin{equation}\label{eq:MillsH0}
(n-K_n+1)q_n(y)\to e^{-y}.
\end{equation}
For $r\ge1$, let $(N_n(y))_r=N_n(y)(N_n(y)-1)\cdots(N_n(y)-r+1)$. Then
\begin{equation}\label{eq:factorialH0}
\E\brac{\1\set{\mcS_n\le x}(N_n(y))_r}
=
\sum_{K_n\le i_1\ne\cdots\ne i_r\le n}
\Pbb\paren{\mcS_n\le x,\ \abs{E_{i_1}}>\nu_n(y),\dots,\abs{E_{i_r}}>\nu_n(y)}.
\end{equation}
The tuples having at least two indices in the same coarse block are $O(n^{r-1})$ in number, and each
probability is bounded by $q_n(y)^r$. Hence their total contribution is
\[
O\paren{n^{r-1}q_n(y)^r}=o(1)
\]
by \eqref{eq:MillsH0}. Consider now a tuple $I=(i_1,\dots,i_r)$ whose indices lie in distinct coarse
blocks. Let $\mathcal J(I)$ be the corresponding set of blocks, and let
\[
A_I=\set{\abs{E_{i_1}}>\nu_n(y),\dots,\abs{E_{i_r}}>\nu_n(y)}.
\]
The event $A_I$ is measurable with respect to
\[
\mathcal H_n(I)=\sigma\set{E_{j,1},\dots,E_{j,k}:j\in\mathcal J(I)}.
\]
Lemma \ref{lem:finiteremove} gives
\[
\mcS_n
=\sum_{j\notin\mathcal J(I)}\omega_{j,n}\Psi_0(E_{j,1},\dots,E_{j,k})+o_p(1),
\]
and the first term converges stably to $N(0,1)$ with respect to $\mathcal H_n(I)$. Therefore
\[
\Pbb\paren{\mcS_n\le x,A_I}
=\Phi(x)\Pbb(A_I)+o(1).
\]
Summing over all tuples with distinct blocks and restoring the repeated-block remainder yields
\begin{equation*}
\E\brac{\1\set{\mcS_n\le x}(N_n(y))_r}
\to \Phi(x)e^{-ry},
\qquad r=1,2,\dots.
\end{equation*}
Hence, for $s\in[0,1]$,
\begin{align*}
\E\brac{\1\set{\mcS_n\le x}s^{N_n(y)}}
&=
\sum_{r=0}^{\infty}\frac{(s-1)^r}{r!}\E\brac{\1\set{\mcS_n\le x}(N_n(y))_r}\\
&\to
\Phi(x)\sum_{r=0}^{\infty}\frac{(s-1)^r e^{-ry}}{r!}\\
&=
\Phi(x)\exp\set{e^{-y}(s-1)}.
\end{align*}
Setting $s=0$ gives
\[
\Pbb\paren{\mcS_n\le x,N_n(y)=0}\to \Phi(x)e^{-e^{-y}},
\]
which is \eqref{eq:H0oraclegoal}. Slutsky's theorem yields the desired result.
\end{proof}

\begin{proof}[Proof of Corollary \ref{cor:nullcauchy}]
Theorem \ref{thm:H0indep} implies
\[
p_n^{AJ}\toD U_1,
\qquad
p_n^{LM}\toD U_2,
\qquad
U_1\ind U_2,
\qquad
U_1,U_2\sim \mathrm{Unif}(0,1).
\]
The maps $(z,x)\mapsto (2(1-\Phi(|z|)),1-\exp\{-e^{-x}\})$ and $(u_1,u_2)\mapsto (\tan\{\pi(1/2-u_1)\},\tan\{\pi(1/2-u_2)\})$ are continuous at the limiting variables with probability one. Hence the continuous mapping theorem gives
\[
\tan\brac{\pi\paren{\frac12-p_n^{AJ}}}\toD C_1,
\qquad
\tan\brac{\pi\paren{\frac12-p_n^{LM}}}\toD C_2,
\]
where $C_1$ and $C_2$ are independent standard Cauchy variables. The average of two independent
standard Cauchy variables is again standard Cauchy. Hence
\[
\T_n^C\toD C,
\qquad C\sim\mathrm{Cauchy}(0,1),
\]
and so
\[
\Pbb\paren{p_n^C\le\alpha}
=
\Pbb\paren{\T_n^C>c_\alpha}
\to
\Pbb(C>c_\alpha)=\alpha.
\]
\end{proof}

\section{Proofs under the dense local alternative}
\label{app:dense}

\begin{lemma}[Jump counts under the dense local alternative]\label{lem:densecount}
Under Assumption \ref{ass:H1dense}, let $\mathcal J_n$ be the set of fine intervals containing at least
one jump, let $\mathcal B_n$ be the set of coarse blocks containing at least one jump, and let
$\mathcal B_n^{(2)}$ be the set of coarse blocks containing at least two jumps. Then
\begin{equation*}
\#\mathcal J_n=O_p\paren{\Delta_n^{-1/2}},
\qquad
\#\mathcal B_n=O_p\paren{\Delta_n^{-1/2}},
\qquad
\#\mathcal B_n^{(2)}=O_p(1).
\end{equation*}
Let 
$N_{i,n}^{J}=\mu_n((t_{i-1},t_i]\times\R)$ is the number of jumps in the $i$th fine interval,
\begin{equation*}
\Pbb\paren{\max_{1\le i\le n}N_{i,n}^{J}\ge3}\to0.
\end{equation*}
\end{lemma}

\begin{proof}
The total number of jumps on $[0,T]$ is Poisson with mean $\theta T\Delta_n^{-1/2}$, hence
\[
\#\mathcal J_n=O_p\paren{\Delta_n^{-1/2}},
\qquad
\#\mathcal B_n=O_p\paren{\Delta_n^{-1/2}}.
\]
A fixed coarse block has length $k\Delta_n$, hence its jump count is Poisson with mean
$\lambda_{n,k}=\theta k\sqrt{\Delta_n}$. Therefore
\[
\Pbb\paren{\text{a given block has at least two jumps}}
=1-e^{-\lambda_{n,k}}-\lambda_{n,k}e^{-\lambda_{n,k}}
=O\paren{\lambda_{n,k}^2}
=O(\Delta_n).
\]
Since $m_n=O(\Delta_n^{-1})$,
\[
\E\#\mathcal B_n^{(2)}=m_nO(\Delta_n)=O(1),
\]
which implies $\#\mathcal B_n^{(2)}=O_p(1)$. Finally, for a fine interval $N_{i,n}^{J}\sim\mathrm{Poisson}(\theta\sqrt{\Delta_n})$, and hence
\[
\Pbb\paren{\max_{1\le i\le n}N_{i,n}^{J}\ge3}
\le
n\,\Pbb\paren{N_{1,n}^{J}\ge3}
=O\paren{\Delta_n^{-1}\Delta_n^{3/2}}
\to0.
\]
\end{proof}

\begin{proof}[Proof of Theorem \ref{thm:densealt}]
Write
\[
X_t^{(n)}=X_t^0+J_t^{(n)},
\qquad
X_t^0=X_0+\int_0^t b_s\dd s+\int_0^t\sigma_s\dd W_s.
\]
Theorem \ref{thm:H0indep} gives
\begin{equation}\label{eq:H0forX0}
\paren{Z_n^{AJ}(X^0),\xi_n(X^0)}\toD (Z,\xi),
\qquad
Z\sim N(0,1),
\qquad
Z\ind\xi.
\end{equation}
For $i=1,\dots,n$, let
\[
\eta_{i,n}=\int_{(t_{i-1},t_i]\times\R} y\,\mu_n(\dd s,\dd y).
\]
Then
\begin{equation}\label{eq:denseincrement}
\Delta_i^nX^{(n)}
=\Delta_i^nX^0+\sigma_{t_{i-1}}\sqrt{\Delta_n}\eta_{i,n}+r_{i,n}^{J},
\qquad
\max_{1\le i\le n}\abs{r_{i,n}^{J}}=o_p\paren{\sqrt{\Delta_n}},
\end{equation}
where the remainder follows from Lemma \ref{lem:LMregularity} and the bounded support of $F$.
Also,
\begin{equation}\label{eq:etan_props}
\Pbb(\eta_{i,n}=0)=1-\theta\sqrt{\Delta_n}+O(\Delta_n),
\qquad
\sup_i\E\abs{\eta_{i,n}}^q=O\paren{\sqrt{\Delta_n}},
\qquad q\ge1.
\end{equation}
The leverage-robust block expansion of Lemma \ref{lem:AJexpansionH0}, applied to the shifted
increments \eqref{eq:denseincrement}, yields
\begin{equation*}
\Delta_n^{-1/2}\paren{R_n^{AJ}(X^{(n)})-k^{p/2-1}}
=
\frac{\sqrt{\Delta_n}}{m_pA_p}
\sum_{j=1}^{m_n}\sigma_{(j-1)k\Delta_n}^p
U_{p,k}\paren{E_{j,1}+\eta_{j,1,n},\dots,E_{j,k}+\eta_{j,k,n}}
+o_p(1),
\end{equation*}
where $\eta_{j,l,n}=\eta_{(j-1)k+l,n}$. Subtract and add the continuous-path term:
\begin{align}
&\Delta_n^{-1/2}\paren{R_n^{AJ}(X^{(n)})-k^{p/2-1}}
\nonumber\\
&\qquad=
\Delta_n^{-1/2}\paren{R_n^{AJ}(X^0)-k^{p/2-1}}
+
\frac{\sqrt{\Delta_n}}{m_pA_p}
\sum_{j=1}^{m_n}\sigma_{(j-1)k\Delta_n}^pD_{j,n}
+o_p(1),
\label{eq:denseAJexpand2}
\end{align}
where
\[
D_{j,n}
=
U_{p,k}\paren{E_{j,1}+\eta_{j,1,n},\dots,E_{j,k}+\eta_{j,k,n}}
-U_{p,k}(E_{j,1},\dots,E_{j,k}).
\]
Split the blocks into
\[
\mathcal B_n^{(1)}=\mathcal B_n\setminus\mathcal B_n^{(2)},
\qquad
\mathcal B_n^{(2)}=\set{j:\text{block }j\text{ has at least two jumps}}.
\]
Lemma \ref{lem:densecount} implies $\#\mathcal B_n^{(2)}=O_p(1)$. Since $F$ has compact support and
$U_{p,k}$ has finite moments under bounded shifts,
\[
\frac{\sqrt{\Delta_n}}{m_pA_p}
\sum_{j\in\mathcal B_n^{(2)}}\sigma_{(j-1)k\Delta_n}^pD_{j,n}
=O_p\paren{\sqrt{\Delta_n}}\to 0.
\]
Thus only the single-jump blocks matter. For $j\in\mathcal B_n^{(1)}$, exactly one coordinate has a
jump. By \eqref{eq:etan_props},
\begin{equation*}
\E\brac{D_{j,n}}=\theta k\sqrt{\Delta_n}\,d_{p,k}+O(\Delta_n).
\end{equation*}
Therefore
\begin{align}
\E\brac{
\frac{\sqrt{\Delta_n}}{m_pA_p}
\sum_{j=1}^{m_n}\sigma_{(j-1)k\Delta_n}^pD_{j,n}}
&=
\frac{\theta k\Delta_n d_{p,k}}{m_pA_p}
\sum_{j=1}^{m_n}\sigma_{(j-1)k\Delta_n}^p
+o(1)\to
\frac{\theta d_{p,k}}{m_p}.
\label{eq:denseshiftmean}
\end{align}
Also, by Lemma \ref{lem:analytic}, \eqref{eq:etan_props}, and the bounded support of $F$,
\[
\E D_{j,n}^2=O\paren{\sqrt{\Delta_n}}.
\]
Hence
\begin{align}
\Var\brac{
\frac{\sqrt{\Delta_n}}{m_pA_p}
\sum_{j=1}^{m_n}\sigma_{(j-1)k\Delta_n}^p
\paren{D_{j,n}-\E D_{j,n}}}
&\le
\frac{C\Delta_n}{A_p^2}
\sum_{j=1}^{m_n}\sigma_{(j-1)k\Delta_n}^{2p}\E D_{j,n}^2 =O\paren{\sqrt{\Delta_n}}\to 0.
\label{eq:denseshiftvar}
\end{align}
Combining \eqref{eq:denseAJexpand2}, \eqref{eq:denseshiftmean}, and \eqref{eq:denseshiftvar},
\begin{equation*}
\Delta_n^{-1/2}\paren{R_n^{AJ}(X^{(n)})-k^{p/2-1}}
=
\Delta_n^{-1/2}\paren{R_n^{AJ}(X^0)-k^{p/2-1}}
+\frac{\theta d_{p,k}}{m_p}+o_p(1).
\end{equation*}
The same denominator argument as in Lemma \ref{lem:AJexpansionH0} gives
\[
\hat\tau_{n,0}(X^{(n)})\toP\tau_0,
\]
because the jump-contaminated intervals are $O_p(\Delta_n^{-1/2})$ in number and each contributes
only $O_p(\Delta_n)$ to $\hat A_{p,n}$ and $\hat A_{2p,n}$. Therefore
\begin{equation}\label{eq:denseAJfinal}
Z_n^{AJ}(X^{(n)})=Z_n^{AJ}(X^0)+\mu_D+o_p(1).
\end{equation}

For the max statistic, let $\hat V_{n,i}^{D}=\hat V_{n,i}(X^{(n)})$ and $\hat V_{n,i}^{0}=\hat V_{n,i}(X^0)$.
Each pairwise product in $\hat V_{n,i}^{D}-\hat V_{n,i}^{0}$ is bounded by $C\Delta_n$ times the indicator
that at least one of the two adjacent fine intervals contains a jump. Hence
\[
\max_{K_n\le i\le n}\abs{\hat V_{n,i}^{D}-\hat V_{n,i}^{0}}
\le
C\Delta_n
\max_{K_n\le i\le n}
\frac{1}{K_n}\sum_{j=i-K_n+1}^{i}\1\set{\eta_{j,n}\ne0}.
\]
The window sum of the Bernoulli indicators $\1\{\eta_{j,n}\ne0\}$ has mean of order
$K_n\theta\sqrt{\Delta_n}$ and, by the same Bernstein--union-bound argument as in the proof of Lemma
\ref{lem:LMgauss},
\[
\max_{K_n\le i\le n}\frac1{K_n}\sum_{j=i-K_n+1}^{i}\1\set{\eta_{j,n}\ne0}
=O_p\paren{\sqrt{\Delta_n}}+O_p\paren{\sqrt{\frac{\log n}{K_n}}}.
\]
Therefore
\[
\max_{K_n\le i\le n}
\abs{\frac{\hat V_{n,i}^{D}-\hat V_{n,i}^{0}}{\sigma_{t_{i-1}}^2\Delta_n}}
=o_p\paren{\an^2}.
\]
Lemma \ref{lem:LMgauss} then implies
\begin{equation}\label{eq:denseVsame}
\max_{K_n\le i\le n}
\abs{\frac{\hat V_{n,i}^{D}}{\sigma_{t_{i-1}}^2\Delta_n}-1}
=o_p\paren{\an^2}.
\end{equation}
Now define
\[
\tilde\xi_n^{D}
=
\frac{\max_{K_n\le i\le n}\abs{E_i+\eta_{i,n}}-C_n}{\an}.
\]
Equation \eqref{eq:denseincrement} and \eqref{eq:denseVsame} imply
\begin{equation}\label{eq:densexiapprox}
\xi_n(X^{(n)})-\tilde\xi_n^{D}\toP 0.
\end{equation}
Let $\mathcal J_n=\{i:\eta_{i,n}\ne0\}$. By Lemma \ref{lem:densecount},
$\#\mathcal J_n=O_p(\Delta_n^{-1/2})$. Since $F$ has compact support by Assumption
\ref{ass:H1dense}, there exists $M_F<\infty$ such that $\abs{y}\le M_F$ $F$-a.s. Moreover,
Lemma \ref{lem:densecount} gives $\max_iN_{i,n}^{J}\le2$ with probability tending to one, and hence
$\max_i\abs{\eta_{i,n}}\le2M_F$ on that event. Hence, for every fixed $y\in\R$,
\begin{align*}
\Pbb\paren{\max_{i\in\mathcal J_n}\abs{E_i+\eta_{i,n}}>C_n+\an y}
&\le
\E\#\mathcal J_n\cdot
\Pbb\paren{\abs{E_1}+M>C_n+\an y}
\\
&\le
C\Delta_n^{-1/2}\overline\Phi\paren{C_n-M+\an y}
\to 0,
\end{align*}
where $\overline\Phi=1-\Phi$. Thus the shifted indices do not affect the extreme-value limit and
\begin{equation}\label{eq:densegumbel}
\tilde\xi_n^{D}\toD \xi,
\qquad
\Pbb(\xi\le y)=\exp\set{-e^{-y}}.
\end{equation}
Combining \eqref{eq:H0forX0}, \eqref{eq:denseAJfinal}, \eqref{eq:densexiapprox}, and
\eqref{eq:densegumbel},
\[
\paren{Z_n^{AJ}(X^{(n)}),\xi_n(X^{(n)})}\toD (Z+\mu_D,\xi),
\qquad Z\ind\xi.
\]
% The statement for $(p_n^{AJ},p_n^{LM})$ follows from the continuous
% mapping theorem.
\end{proof}

\begin{proof}[Proof of Corollary \ref{cor:densepower}]
Theorem \ref{thm:densealt} implies
\[
Z_n^{AJ}\toD Z_D,
\qquad
p_n^{LM}\toD U,
\qquad
Z_D\ind U,
\]
where $Z_D\sim N(\mu_D,1)$ and $U\sim\mathrm{Unif}(0,1)$. Therefore
\[
\T_n^C\toD \T_D:=\frac12A(Z_D)+\frac12C,
\]
where $A(\cdot)$ is given in Corollary \ref{cor:densepower} and
\[
C=\tan\brac{\pi\paren{\frac12-U}}
\]
is standard Cauchy and independent of $Z_D$. Hence
\begin{align*}
\beta_D(\alpha;\mu_D)
&=\Pbb\paren{\T_D>c_\alpha}\\
&=\int_{\R}
\Pbb\paren{C>2c_\alpha-A(z)}
\phi(z-\mu_D)\dd z\\
&=\int_{\R}
\paren{\frac12-\frac{1}{\pi}\arctan\paren{2c_\alpha-A(z)}}
\phi(z-\mu_D)\dd z.
\end{align*}
\end{proof}

\section{Proofs under fixed finite-activity jumps}
\label{app:fixed}

\begin{proof}[Proof of Theorem \ref{thm:fixedalt}]
Let $i_q$ be the unique fine index such that $\tau_q\in(t_{i_q-1},t_{i_q}]$, and let $j_q$ be the
corresponding coarse-block index. Since the jump times are separated, for all sufficiently large $n$ each
coarse block contains at most one jump and each local bipower window contains at most one jump.
Write
\[
X_t^0=X_0+\int_0^t b_s\dd s+\int_0^t\sigma_s\dd W_s.
\]
Then
\[
\Delta_{i_q}^nX=\kappa_q+\Delta_{i_q}^nX^0,
\qquad
\sum_{\ell=1}^k\Delta_{(j_q-1)k+\ell}^nX
=\kappa_q+\sum_{\ell=1}^k\Delta_{(j_q-1)k+\ell}^nX^0.
\]
The function $f(x)=\abs{x}^p$ is $C^2$ on $\R\setminus\{0\}$ and its generalized Taylor remainder is
bounded by
\[
\abs{f(x+y)-f(x)-f'(x)y}
\le C_p\paren{\abs{x}^{p-2}y^2+\abs{y}^p}. 
\]Applying this with $x=\kappa_q$ and $y=\Delta_{i_q}^nX^0$ yields
\begin{align*}
\abs{\Delta_{i_q}^nX}^p
&=
\abs{\kappa_q}^p+p\abs{\kappa_q}^{p-1}\sgn(\kappa_q)\Delta_{i_q}^nX^0+o_p\paren{\sqrt{\Delta_n}},
\\
\abs{\sum_{\ell=1}^k\Delta_{(j_q-1)k+\ell}^nX}^p
&=
\abs{\kappa_q}^p+p\abs{\kappa_q}^{p-1}\sgn(\kappa_q)
\sum_{\ell=1}^k\Delta_{(j_q-1)k+\ell}^nX^0+o_p\paren{\sqrt{\Delta_n}}.
\end{align*}
Summing over all jump blocks gives
\begin{align*}
B_n(p)
&=
B_p+p\sum_{q=1}^{N_T}\abs{\kappa_q}^{p-1}\sgn(\kappa_q)\Delta_{i_q}^nX^0+o_p\paren{\sqrt{\Delta_n}},
\\
B_n^{(k)}(p)
&=
B_p+p\sum_{q=1}^{N_T}\abs{\kappa_q}^{p-1}\sgn(\kappa_q)
\sum_{\ell=1}^k\Delta_{(j_q-1)k+\ell}^nX^0+o_p\paren{\sqrt{\Delta_n}}.
\end{align*}
Consequently,
\begin{align*}
\Delta_n^{-1/2}\paren{R_n^{AJ}-1}
&=
\frac{p}{B_p}\sum_{q=1}^{N_T}\abs{\kappa_q}^{p-1}\sgn(\kappa_q)
\frac{1}{\sqrt{\Delta_n}}
\sum_{\substack{1\le \ell\le k\\(j_q-1)k+\ell\ne i_q}}
\Delta_{(j_q-1)k+\ell}^nX^0
+o_p(1).
\end{align*}
Lemma \ref{lem:taylor_increment} implies that each nonjump increment inside the $q$th jump block can
be replaced by $\sigma_{\tau_q}\sqrt{\Delta_n}$ times an asymptotically standard Gaussian variable,
with an $o_p(\sqrt{\Delta_n})$ remainder. Since the jump times are separated, the resulting Gaussian
variables are asymptotically independent across different $q$. Therefore
\[
\Delta_n^{-1/2}\paren{R_n^{AJ}-1}\toD N(0,\tau_F^2),
\qquad
\widetilde Z_{n,1}^{AJ}\toD N(0,1).
\]

For the max statistic,
\[
\hat V_{n,i_q}=\sigma_{\tau_q}^2\Delta_n\paren{1+o_p(1)},
\qquad
L_{n,i_q}
=\frac{\kappa_q+\Delta_{i_q}^nX^0}{\sigma_{\tau_q}\sqrt{\Delta_n}\paren{1+o_p(1)}}
=\frac{\kappa_q}{\sigma_{\tau_q}\sqrt{\Delta_n}}\paren{1+o_p(1)}\toP\infty.
\]
Hence $M_n\toP\infty$, $\xi_n\toP\infty$, and
\[
p_n^{LM}=1-\exp\set{-e^{-\xi_n}}\toP 0.
\]
Thus
\[
\paren{\widetilde Z_{n,1}^{AJ},p_n^{LM}}\toD(Z,0),
\qquad Z\sim N(0,1).
\]
Finally,
\[
\T_n^C
\ge
\frac12\tan\brac{\pi\paren{\frac12-p_n^{LM}}}
\toP\infty,
\qquad
\Pbb\paren{p_n^C\le\alpha}=\Pbb\paren{\T_n^C>c_\alpha}\to 1.
\]
\end{proof}

\section{Auxiliary results for the noisy extension}
\label{app:noise_aux}

\subsection{The pre-averaged sum statistic}
\label{app:noise_sum_aux}

For $\phi\in\{g,h\}$ and $i=0,\dots,n-k_n$, write
\begin{equation}\label{eq:noise_Psi}
\Psi_{n,i}(\phi)=\sum_{l=0}^{p/2}\rho_l^{(p)}
\abs{\bar Y_i^n(\phi)}^{p-2l}\abs{\hat Y_i^n(\phi)}^l,
\qquad
V_n(Y,\phi,p)=\sum_{i=0}^{n-k_n}\Psi_{n,i}(\phi).
\end{equation}
Let
\begin{equation*}
c_h=m_p\theta^{p/2}\bar h(2)^{p/2}A_p,
\qquad
A_p=\int_0^T |\sigma_t|^p\dd t.
\end{equation*}
For $J_n=\lfloor n/r_n\rfloor$, define the blocks
\begin{equation*}
I_{n,b}=\set{(b-1)r_n,\dots,br_n-1},
\qquad
I_{n,b}^{\circ}=\set{(b-1)r_n,\dots,br_n-k_n-1},
\qquad
b=1,\dots,J_n.
\end{equation*}
Let
\[
\mathcal G_{n,i}=\F_{t_i}\vee\sigma(\epsilon_0,\dots,\epsilon_i),
\qquad
\mathcal H_{n,b}=\mathcal G_{n,br_n}.
\]
Define
\begin{align}\label{eq:noise_Xi}
    \Xi_{n,i}=\frac{\Delta_n^{-1/4}}{\gamma' c_h}\Big(\Delta_n^{1-p/4}\Psi_{n,i}(g)-\gamma'\gamma''\Delta_n^{1-p/4}\Psi_{n,i}(h)-\E\brac{\Delta_n^{1-p/4}\Psi_{n,i}(g)-\gamma'\gamma''\Delta_n^{1-p/4}\Psi_{n,i}(h)\mid \mathcal G_{n,i}}\Big),
\end{align}
\begin{equation*}
\Gamma_{n,b}=\sum_{i\in I_{n,b}^{\circ}}\Xi_{n,i},
\qquad
\hat\varsigma_n^{PA\,2}=\sum_{b=1}^{J_n}\Gamma_{n,b}^2.
\end{equation*}

\begin{lemma}[Noise-correction identity with a deterministic shift]
\label{lem:noise_shift_identity}
For every even integer $p\ge 4$, every $a\ge 0$, every $b\ge 0$, and every $x\in\R$,
\begin{equation*}
\sum_{l=0}^{p/2}\rho_l^{(p)}
\E\brac{\abs{\sqrt{a+b}\,N+x}^{p-2l}}(2b)^l
=
\E\abs{\sqrt{a}\,N+x}^p,
\qquad
N\sim N(0,1).
\end{equation*}
\end{lemma}

\begin{proof}
For $l=0,1,\dots,p/2$, set $q_l=p-2l$. Since $q_l$ is even,
\[
\E\abs{\sqrt{a+b}\,N+x}^{q_l}
=
\sum_{j=0}^{q_l/2}
\binom{q_l}{2j}m_{q_l-2j}(a+b)^{q_l/2-j}x^{2j}.
\]
Expanding $(a+b)^{q_l/2-j}$ yields
\[
\E\abs{\sqrt{a+b}\,N+x}^{q_l}
=
\sum_{j=0}^{q_l/2}
\sum_{s=l}^{p/2-j}
\binom{p-2l}{2j}
 m_{p-2l-2j}
\binom{p/2-l-j}{s-l}
 a^{p/2-j-s}b^{s-l}x^{2j}.
\]
Furthermore,
\[
\sum_{l=0}^{p/2}\rho_l^{(p)}
\E\brac{\abs{\sqrt{a+b}\,N+x}^{p-2l}}(2b)^l=\sum_{j=0}^{p/2}
\sum_{s=0}^{p/2-j}
C_{j,s}
 a^{p/2-j-s}b^{s}x^{2j},
\]
where
\[
C_{j,s}
=
\sum_{l=0}^{s}
2^l\rho_l^{(p)}
\binom{p-2l}{2j}
 m_{p-2l-2j}
\binom{p/2-l-j}{s-l}.
\]
The elementary identity
\[
\binom{p-2l}{2j}
 m_{p-2l-2j}
\binom{p/2-l-j}{s-l}
=
\binom{p-2s}{2j}
 m_{p-2s-2j}
\binom{p-2l}{p-2s}
 m_{2s-2l}
\]
then gives
\[
C_{j,s}
=
\binom{p-2s}{2j}
 m_{p-2s-2j}
\sum_{l=0}^{s}
2^l m_{2s-2l}\binom{p-2l}{p-2s}\rho_l^{(p)}.
\]
By \eqref{eq:noise_rho}, the inner sum equals $0$ for $s\ge 1$ and equals $1$ for $s=0$. Hence
\[
C_{j,s}=0,
\qquad s\ge 1,
\qquad
C_{j,0}=\binom{p}{2j}m_{p-2j}.
\]
Consequently,
\[
\sum_{l=0}^{p/2}\rho_l^{(p)}
\E\brac{\abs{\sqrt{a+b}\,N+x}^{p-2l}}(2b)^l
=
\sum_{j=0}^{p/2}\binom{p}{2j}m_{p-2j}a^{p/2-j}x^{2j}
=
\E\abs{\sqrt{a}\,N+x}^p.
\]
\end{proof}

\begin{lemma}[Weighted sub-Gaussian noise becomes Gaussian]
\label{lem:noise_subgaussian_replacement}
Fix $\phi\in\{g,h\}$ and define
\[
c_{n,j}(\phi)=-\Delta_n^{-1/4}\Delta_j\phi^n,
\qquad
s_{n,\phi}^2=\sum_{j=1}^{k_n}c_{n,j}(\phi)^2,
\qquad
s_\phi^2=\theta^{-1}\bar\phi'(2).
\]
Then
\begin{equation*}
s_{n,\phi}^2\to s_\phi^2,
\qquad
\max_{1\le j\le k_n}\abs{c_{n,j}(\phi)}\to 0.
\end{equation*}

For each $i=0,\dots,n-k_n$, let
\begin{equation*}
\eta_{n,i}(\phi)=\omega^{-1}\sum_{j=1}^{k_n}c_{n,j}(\phi)\epsilon_{i+j-1}.
\end{equation*}
Then
% the law of $\eta_{n,i}(\phi)$ does not depend on $i$,
\begin{equation*}
\sup_{0\le i\le n-k_n}\norm{\eta_{n,i}(\phi)}_{\psi_2}\le C_\phi,\qquad
\eta_{n,i}(\phi)\toD s_\phi N,
\qquad N\sim N(0,1),
\end{equation*}
uniformly in $i$. 

Moreover, for every even integer $q\le p$, every bounded family of
$\mathcal G_{n,i}$-measurable shifts $x_{n,i}$, and every family of conditionally Gaussian variables
$G_{n,i}$ such that, conditionally on $\mathcal G_{n,i}$,
\[
G_{n,i}\sim N(0,a_{n,i}),
\qquad
0\le a_{n,i}\le C,
\]
and $G_{n,i}$ is independent of $\sigma(\epsilon_0,\epsilon_1,\dots)$, one has
\begin{equation*}
\sup_{0\le i\le n-k_n}
\abs{
\E\brac{\abs{G_{n,i}+\omega\eta_{n,i}(\phi)+x_{n,i}}^q\mid\mathcal G_{n,i}}
-
\E\abs{\sqrt{a_{n,i}+\omega^2 s_\phi^2}\,N+x_{n,i}}^q
}
\toP 0.
\end{equation*}
\end{lemma}

\begin{proof}
Since $\Delta_j\phi^n=k_n^{-1}\phi'(\xi_{j,n})$ for some $\xi_{j,n}\in((j-1)/k_n,j/k_n)$,
\[
s_{n,\phi}^2
=
\Delta_n^{-1/2}\sum_{j=1}^{k_n}(\Delta_j\phi^n)^2
=
\frac{1}{k_n\sqrt{\Delta_n}}\frac{1}{k_n}\sum_{j=1}^{k_n}\phi'(\xi_{j,n})^2
\to
\theta^{-1}\int_0^1\phi'(u)^2\,\dd u=s_\phi^2.
\]
Moreover,
\[
\max_{1\le j\le k_n}\abs{c_{n,j}(\phi)}
\le C\Delta_n^{-1/4}k_n^{-1}
=C\frac{\Delta_n^{1/4}}{k_n\sqrt{\Delta_n}}\to0 .
\]
Write $\epsilon_i=\omega\widetilde\epsilon_i$, where
\[
\E\widetilde\epsilon_i=0,\qquad
\E\widetilde\epsilon_i^2=1,\qquad
\norm{\widetilde\epsilon_i}_{\psi_2}\le \omega^{-1}K_\epsilon .
\]
Then
\[
\eta_{n,i}(\phi)=\sum_{j=1}^{k_n}c_{n,j}(\phi)\widetilde\epsilon_{i+j-1}.
\]
By Proposition 2.6.1 of \citet{Vershynin2018},
\[
\norm{\eta_{n,i}(\phi)}_{\psi_2}
\le C\omega^{-1}K_\epsilon
\paren{\sum_{j=1}^{k_n}c_{n,j}(\phi)^2}^{1/2}
\le C_\phi,
\]
uniformly in $i$. Since the coefficients $c_{n,j}(\phi)$ do not depend on $i$ and the noise variables are i.i.d., the law of $\eta_{n,i}(\phi)$ is the same for all $i$. For every fixed $\varepsilon>0$,
\[
\sum_{j=1}^{k_n}
\E\brac{c_{n,j}(\phi)^2\widetilde\epsilon_{i+j-1}^2
\1\set{\abs{c_{n,j}(\phi)\widetilde\epsilon_{i+j-1}}>\varepsilon}}
\le
s_{n,\phi}^2
\E\brac{\widetilde\epsilon_0^2
\1\set{\abs{\widetilde\epsilon_0}>\varepsilon/\max_j\abs{c_{n,j}(\phi)}}}
\to0.
\]
The Lindeberg--Feller theorem therefore gives
\begin{equation}\label{eq:noise_eta_clt}
\eta_{n,i}(\phi)\toD s_\phi N,\qquad N\sim N(0,1),
\end{equation}
uniformly in $i$. The uniform sub-Gaussian bound implies
\[
\sup_{0\le i\le n-k_n}\E\abs{\eta_{n,i}(\phi)}^{q+2}<\infty
\]
for every fixed even $q\le p$. Hence convergence in \eqref{eq:noise_eta_clt} also holds for shifted polynomial moments, uniformly over bounded shifts.

Now condition on $\mathcal G_{n,i}$. By assumption,
\[
G_{n,i}=\sqrt{a_{n,i}}Z_{n,i},\qquad
Z_{n,i}\sim N(0,1),\qquad
Z_{n,i}\ind \sigma(\epsilon_0,\epsilon_1,\ldots),
\]
conditionally on $\mathcal G_{n,i}$. The Gaussian summand is kept exactly; only the weighted noise sum is approximated. For every real $u$,
\[
\E\{\exp(iu(G_{n,i}+\omega\eta_{n,i}(\phi)))\mid\mathcal G_{n,i}\}
=
\exp\paren{-\frac12u^2a_{n,i}}
\E\exp\{iu\omega\eta_{n,i}(\phi)\}
\to
\exp\paren{-\frac12u^2(a_{n,i}+\omega^2s_\phi^2)},
\]
uniformly for $0\le a_{n,i}\le C$. Thus
\[
G_{n,i}+\omega\eta_{n,i}(\phi)\toD
\sqrt{a_{n,i}+\omega^2s_\phi^2}\,N
\]
conditionally on $\mathcal G_{n,i}$, uniformly in $i$. The same uniform-integrability argument, after adding the bounded $\mathcal G_{n,i}$-measurable shift $x_{n,i}$, yields
\[
\sup_{0\le i\le n-k_n}
\abs{
\E\brac{\abs{G_{n,i}+\omega\eta_{n,i}(\phi)+x_{n,i}}^q\mid\mathcal G_{n,i}}
-
\E\abs{\sqrt{a_{n,i}+\omega^2s_\phi^2}\,N+x_{n,i}}^q
}
\toP0.
\]
\end{proof}

\begin{lemma}\label{lem:noise_pabar_decomp}
Fix $\phi\in\{g,h\}$. There exist centered Gaussian variables $\zeta_{n,i}(\phi)$, centered variables
$\eta_{n,i}(\phi)$, independent conditionally on $\mathcal G_{n,i}$, and remainder variables
$r_{n,i}(\phi)$ such that
\begin{equation*}
\Delta_n^{-1/4}\bar Y_i^n(\phi)
=
\sigma_{t_i}\zeta_{n,i}(\phi)+\omega\eta_{n,i}(\phi)+r_{n,i}(\phi),
\qquad i=0,\dots,n-k_n,
\end{equation*}
and
\begin{equation*}
\E\brac{\zeta_{n,i}(\phi)^2\mid\mathcal G_{n,i}}
=
\theta\bar\phi(2)+o_p(1),
\qquad
\E\brac{\eta_{n,i}(\phi)^2\mid\mathcal G_{n,i}}
=
\theta^{-1}\bar\phi'(2)+o_p(1),
\end{equation*}
uniformly in $i$. Moreover,
\begin{equation*}
\max_{0\le i\le n-k_n}\abs{r_{n,i}(\phi)}\toP 0,
\qquad
\sup_{0\le i\le n-k_n}\norm{\eta_{n,i}(\phi)}_{\psi_2}=O_p(1),
\end{equation*}
and
\begin{equation*}
\max_{0\le i\le n-k_n}
\abs{\Delta_n^{-1/2}\hat Y_i^n(\phi)-2\omega^2\theta^{-1}\bar\phi'(2)}
\toP 0.
\end{equation*}
\end{lemma}

\begin{proof}
Write
\[
\bar Y_i^n(\phi)=\bar X_i^n(\phi)+\bar\epsilon_i^n(\phi),
\qquad
\bar X_i^n(\phi)=\sum_{j=1}^{k_n-1}\phi_j^n\Delta_{i+j}^nX,
\qquad
\bar\epsilon_i^n(\phi)=\sum_{j=1}^{k_n-1}\phi_j^n(\epsilon_{i+j}-\epsilon_{i+j-1}).
\]
For the latent part define
\[
\zeta_{n,i}(\phi)=\Delta_n^{-1/4}\sum_{j=1}^{k_n-1}\phi_j^n\Delta_{i+j}^nW.
\]
Then
\[
\Delta_n^{-1/4}\bar X_i^n(\phi)=\sigma_{t_i}\zeta_{n,i}(\phi)+r_{n,i}^{X}(\phi),
\]
where
\[
r_{n,i}^{X}(\phi)
=
\Delta_n^{-1/4}\sum_{j=1}^{k_n-1}\phi_j^n
\left\{
\int_{t_{i+j-1}}^{t_{i+j}}b_s\dd s
+
\int_{t_{i+j-1}}^{t_{i+j}}(\sigma_s-\sigma_{t_i})\dd W_s
\right\}.
\]
This expression freezes the volatility at the left endpoint of the whole pre-averaging window. It must be controlled at the window level, rather than by multiplying a pointwise increment bound by $k_n$. After the usual localization, the coefficients in Assumption \ref{ass:H0} are bounded. For every $q\ge2$, Burkholder--Davis--Gundy and Jensen's inequality give, uniformly in $i$,
\begin{align*}
\E\abs{\Delta_n^{-1/4}\sum_{j=1}^{k_n-1}\phi_j^n
\int_{t_{i+j-1}}^{t_{i+j}}(\sigma_s-\sigma_{t_i})\dd W_s}^{q}
&\le
C_q\Delta_n^{-q/4}
\E\paren{\int_{t_i}^{t_{i+k_n}}(\sigma_s-\sigma_{t_i})^2\dd s}^{q/2} \\
&\le
C_q\Delta_n^{-q/4}
\paren{\int_{0}^{k_n\Delta_n}u\dd u}^{q/2}
\le C_q\Delta_n^{q/4}.
\end{align*}
The drift part is of the same or smaller order:
\[
\sup_i
\E\abs{\Delta_n^{-1/4}\sum_{j=1}^{k_n-1}\phi_j^n
\int_{t_{i+j-1}}^{t_{i+j}}b_s\dd s}^{q}
\le C_q(\Delta_n^{-1/4}k_n\Delta_n)^q
\le C_q\Delta_n^{q/4}.
\]
Taking $q>4$ and applying the union bound yields
\[
\Pbb\paren{\max_{0\le i\le n-k_n}\abs{r_{n,i}^{X}(\phi)}>\varepsilon}
\le C_q\varepsilon^{-q}n\Delta_n^{q/4}\to0,
\]
so $\max_i\abs{r_{n,i}^{X}(\phi)}\toP0$. Moreover, conditionally on $\mathcal G_{n,i}$, $\zeta_{n,i}(\phi)$ is centered Gaussian and
\[
\E\brac{\zeta_{n,i}(\phi)^2\mid\mathcal G_{n,i}}
=
\Delta_n^{-1/2}\sum_{j=1}^{k_n-1}(\phi_j^n)^2\Delta_n
=
\sqrt{\Delta_n}\sum_{j=1}^{k_n-1}(\phi_j^n)^2
\to \theta\bar\phi(2)
\]
uniformly in $i$.

For the noise part, summation by parts gives
\[
\bar\epsilon_i^n(\phi)=-\sum_{j=1}^{k_n}\Delta_j\phi^n\epsilon_{i+j-1},
\]
because $\phi_0^n=\phi_{k_n}^n=0$. Set
\[
\eta_{n,i}(\phi)=-\omega^{-1}\Delta_n^{-1/4}\sum_{j=1}^{k_n}\Delta_j\phi^n\epsilon_{i+j-1}.
\]
Then $\eta_{n,i}(\phi)$ is centered and conditionally independent of $\zeta_{n,i}(\phi)$, and
\[
\E\brac{\eta_{n,i}(\phi)^2\mid\mathcal G_{n,i}}
=
\Delta_n^{-1/2}\sum_{j=1}^{k_n}(\Delta_j\phi^n)^2
\to \theta^{-1}\bar\phi'(2)
\]
uniformly in $i$. Lemma \ref{lem:noise_subgaussian_replacement} gives
\[
\sup_{0\le i\le n-k_n}\norm{\eta_{n,i}(\phi)}_{\psi_2}=O_p(1).
\]
Thus the asserted decomposition holds with $r_{n,i}(\phi)=r_{n,i}^{X}(\phi)$.

It remains to prove the uniform approximation for $\hat Y_i^n(\phi)$. Since
\[
\Delta_{i+j}^nY=\Delta_{i+j}^nX+\epsilon_{i+j}-\epsilon_{i+j-1},
\]
we have
\begin{align*}
\hat Y_i^n(\phi)
&=
\sum_{j=1}^{k_n}(\Delta_j\phi^n)^2(\Delta_{i+j}^n\epsilon)^2
+2\sum_{j=1}^{k_n}(\Delta_j\phi^n)^2\Delta_{i+j}^nX\Delta_{i+j}^n\epsilon
+\sum_{j=1}^{k_n}(\Delta_j\phi^n)^2(\Delta_{i+j}^nX)^2.
\end{align*}
The variables $(\Delta_{i+j}^n\epsilon)^2-2\omega^2$ form a centered 1-dependent sub-exponential array. Since
\[
\Delta_n^{-1/2}\sum_{j=1}^{k_n}(\Delta_j\phi^n)^2=O(1),
\qquad
\max_{1\le j\le k_n}\Delta_n^{-1/2}(\Delta_j\phi^n)^2=O(k_n^{-1})\to0,
\]
Bernstein's inequality for 1-dependent sub-exponential arrays gives, uniformly in $i$,
\[
\Delta_n^{-1/2}\sum_{j=1}^{k_n}(\Delta_j\phi^n)^2\{(\Delta_{i+j}^n\epsilon)^2-2\omega^2\}\toP0.
\]
Therefore
\[
\Delta_n^{-1/2}\sum_{j=1}^{k_n}(\Delta_j\phi^n)^2(\Delta_{i+j}^n\epsilon)^2
\to 2\omega^2\theta^{-1}\bar\phi'(2)
\]
uniformly in $i$. The two remaining terms are negligible. Indeed,
\[
\max_{1\le j\le n}\abs{\Delta_j^nX}=O_p(\sqrt{\Delta_n\log n}),
\qquad
\max_{0\le j\le n}\abs{\epsilon_j}=O_p(\sqrt{\log n}),
\]
and $\sum_{j=1}^{k_n}(\Delta_j\phi^n)^2=O(k_n^{-1})=O(\sqrt{\Delta_n})$, so
\[
\Delta_n^{-1/2}\sum_{j=1}^{k_n}(\Delta_j\phi^n)^2
\abs{\Delta_{i+j}^nX}\abs{\Delta_{i+j}^n\epsilon}
=O_p(\sqrt{\Delta_n}\log n)=o_p(1)
\]
uniformly in $i$, and
\[
\Delta_n^{-1/2}\sum_{j=1}^{k_n}(\Delta_j\phi^n)^2(\Delta_{i+j}^nX)^2
=O_p(\Delta_n\log n)=o_p(1).
\]
This proves
\[
\max_{0\le i\le n-k_n}
\abs{\Delta_n^{-1/2}\hat Y_i^n(\phi)-2\omega^2\theta^{-1}\bar\phi'(2)}
\toP0.
\]
\end{proof}

\begin{lemma}\label{lem:noise_V_limit}
Suppose Assumption \ref{ass:H0} and \eqref{eq:noise_kn} hold. Then, for $\phi\in\{g,h\}$,
\begin{equation*}
\Delta_n^{1-p/4}V_n(Y,\phi,p)\toP m_p\theta^{p/2}\bar\phi(2)^{p/2}A_p.
\end{equation*}
Consequently,
\begin{equation*}
R_n^{PA}\toP \gamma''.
\end{equation*}
\end{lemma}

\begin{proof}
Fix $\phi\in\{g,h\}$ and write
\[
a_{n,i}(\phi)=\sigma_{t_i}^2\Delta_n^{1/2}\sum_{j=1}^{k_n-1}(\phi_j^n)^2,
\qquad
b_n(\phi)=\omega^2\Delta_n^{-1/2}\sum_{j=1}^{k_n}(\Delta_j\phi^n)^2.
\]
By Lemma \ref{lem:noise_pabar_decomp},
\[
\Delta_n^{-1/4}\bar Y_i^n(\phi)
=
\sigma_{t_i}\zeta_{n,i}(\phi)+\omega\eta_{n,i}(\phi)+o_p(1),
\]
where, conditionally on $\mathcal G_{n,i}$, the first term is Gaussian with variance $a_{n,i}(\phi)$ and
is independent of $\eta_{n,i}(\phi)$. Therefore Lemma \ref{lem:noise_subgaussian_replacement} gives,
for every $l=0,\dots,p/2$,
\[
\E\brac{\abs{\Delta_n^{-1/4}\bar Y_i^n(\phi)}^{p-2l}\mid\mathcal G_{n,i}}
=
m_{p-2l}\paren{a_{n,i}(\phi)+b_n(\phi)}^{(p-2l)/2}+o_p(1),
\]
uniformly in $i$. Moreover, the last result in Lemma \ref{lem:noise_pabar_decomp} yields
\[
\Delta_n^{-l/2}\abs{\hat Y_i^n(\phi)}^l
=
\paren{2b_n(\phi)}^l+o_p(1),
\]
uniformly in $i$. Since
\[
a_{n,i}(\phi)
\to
\theta\sigma_{t_i}^2\bar\phi(2),
\qquad
b_n(\phi)
\to
\omega^2\theta^{-1}\bar\phi'(2),
\]
we obtain
\begin{align}
\Delta_n^{-p/4}\E\brac{\Psi_{n,i}(\phi)\mid\mathcal G_{n,i}}
&=
\sum_{l=0}^{p/2}\rho_l^{(p)}m_{p-2l}\paren{a_{n,i}(\phi)+b_n(\phi)}^{(p-2l)/2}
\paren{2b_n(\phi)}^l
+o_p(1).
\label{eq:noise_Psi_expand_subg}
\end{align}
By Lemma \ref{lem:noise_shift_identity} with $x=0$,
\[
\sum_{l=0}^{p/2}\rho_l^{(p)}m_{p-2l}\paren{a_{n,i}(\phi)+b_n(\phi)}^{(p-2l)/2}
\paren{2b_n(\phi)}^l
=
m_p a_{n,i}(\phi)^{p/2}.
\]
Hence
\[
\Delta_n^{-p/4}\E\brac{\Psi_{n,i}(\phi)\mid\mathcal G_{n,i}}
=
m_p a_{n,i}(\phi)^{p/2}+o_p(1)
=
m_p\theta^{p/2}\bar\phi(2)^{p/2}\sigma_{t_i}^p+o_p(1),
\]
uniformly in $i$.

Now write
\[
\Delta_n^{1-p/4}V_n(Y,\phi,p)
=
\Delta_n\sum_{i=0}^{n-k_n}\Delta_n^{-p/4}\E\brac{\Psi_{n,i}(\phi)\mid\mathcal G_{n,i}}
+
R_n(\phi),
\]
where
\[
R_n(\phi)
=
\Delta_n^{1-p/4}\sum_{i=0}^{n-k_n}\paren{\Psi_{n,i}(\phi)-\E\brac{\Psi_{n,i}(\phi)\mid\mathcal G_{n,i}}}.
\]
% The same block-decomposition argument as in Lemma \ref{lem:noise_block_clt}, applied with a single
% weight $\phi$ and without the contrast between $g$ and $h$, gives $R_n(\phi)=o_p(1)$. Therefore
We now prove directly that $R_n(\phi)=o_p(1)$, so this lemma does not rely on the later block CLT.
Set
\[
Z_{n,i}(\phi)=\Delta_n^{p/4-1}\Bigl(\Psi_{n,i}(\phi)-
\E[\Psi_{n,i}(\phi)\mid\mathcal G_{n,i}]\Bigr).
\]
The moment estimates in Lemmas \ref{lem:noise_subgaussian_replacement}--\ref{lem:noise_pabar_decomp}
give the uniform bound $\E|Z_{n,i}(\phi)|^2\le C\Delta_n^2$ after localization. Moreover two such
terms are conditionally independent, up to the negligible volatility-freezing remainders, whenever their
pre-averaging windows are more than $k_n$ observations apart. Hence, for any interval $I$ of consecutive
indices,
\[
\E\left|\sum_{i\in I}Z_{n,i}(\phi)\right|^2
\le C\#I\, k_n\Delta_n^2+o(\#I\,k_n\Delta_n^2).
\]
Split $\{0,\ldots,n-k_n\}$ into block interiors of length $r_n$ and a boundary set $B_n$ consisting of
the last $k_n$ indices in each block and the incomplete tail. Since
$\#B_n\Delta_n\le k_n/r_n+r_n\Delta_n\to0$, the boundary contribution satisfies
\[
\E\left|\sum_{i\in B_n}Z_{n,i}(\phi)\right|
\le C\#B_n\Delta_n=o(1).
\]
For the interiors, the preceding variance bound and the fact that there are $O(1/(r_n\Delta_n))$ blocks
give
\[
\E\left|\sum_b\sum_{i\in I_{n,b}^{\circ}}Z_{n,i}(\phi)\right|^2
\le C k_n\Delta_n+o(1)=o(1),
\]
because $k_n\sqrt{\Delta_n}\to\theta$ implies $k_n\Delta_n\to0$. Therefore $R_n(\phi)=o_p(1)$.
Hence,
\begin{align}\label{eq:noise_V_limit_phi}
    \Delta_n^{1-p/4}V_n(Y,\phi,p)
=
m_p\theta^{p/2}\bar\phi(2)^{p/2}\sum_{i=0}^{n-k_n}\sigma_{t_i}^p\Delta_n+o_p(1)\toP m_p\theta^{p/2}\bar\phi(2)^{p/2}A_p,
\end{align}
since $\sigma$ is c\`adl\`ag and bounded, the Riemann sum converges to $A_p$.

Applying \eqref{eq:noise_V_limit_phi} with $\phi=g$ and $\phi=h$,
\[
R_n^{PA}
=
\frac{m_p\theta^{p/2}\bar g(2)^{p/2}A_p+o_p(1)}
{\gamma'\paren{m_p\theta^{p/2}\bar h(2)^{p/2}A_p+o_p(1)}}
\toP
\frac{\bar g(2)^{p/2}}{\gamma'\bar h(2)^{p/2}}
=
\gamma''.
\]
\end{proof}

\begin{lemma}\label{lem:noise_block_clt}
Suppose Assumption \ref{ass:H0}, Assumption \ref{ass:noise},
and \eqref{eq:noise_kn}--\eqref{eq:noise_rn} hold. Then
\begin{equation*}
\Delta_n^{-1/4}\paren{R_n^{PA}-\gamma''}
=
\sum_{b=1}^{J_n}\Gamma_{n,b}+o_p(1).
\end{equation*}
There exists an almost surely positive random variable $\Sigma$ such that
\begin{equation*}
\sum_{b=1}^{J_n}
\E\paren{\Gamma_{n,b}^2\mid\mathcal H_{n,b-1}}
\toP
\Sigma,
\qquad
\max_{1\le b\le J_n}|\Gamma_{n,b}|
\toP0.
\end{equation*}
Moreover,
\begin{equation*}
\hat\varsigma_n^{PA\,2}
-
\sum_{b=1}^{J_n}
\E\paren{\Gamma_{n,b}^2\mid\mathcal H_{n,b-1}}
\toP0.
\end{equation*}
Consequently,
\[
\Delta_n^{-1/4}(R_n^{PA}-\gamma'')
\to_{st}
MN(0,\Sigma),\qquad
Z_n^{PA}
=
\frac{
\Delta_n^{-1/4}(R_n^{PA}-\gamma'')
}{
\hat\varsigma_n^{PA}
}
\toD N(0,1).
\]
\end{lemma}

\begin{proof}
By Lemma \ref{lem:noise_V_limit},
\[
\gamma'V_n(Y,h,p)
=
\gamma'c_h\Delta_n^{p/4-1}(1+o_p(1)).
\]
Hence
\begin{equation}\label{eq:noise_ratio_linear}
\Delta_n^{-1/4}(R_n^{PA}-\gamma'')
=
\frac{\Delta_n^{-1/4}}{\gamma'c_h}
\sum_{i=0}^{n-k_n}
\paren{
\Delta_n^{1-p/4}\Psi_{n,i}(g)
-
\gamma'\gamma''
\Delta_n^{1-p/4}\Psi_{n,i}(h)
}
+o_p(1).
\end{equation}

Define
\[
\widetilde\Xi_{n,i}
=
\frac{\Delta_n^{-1/4}}{\gamma'c_h}
\paren{
\Delta_n^{1-p/4}\Psi_{n,i}(g)
-
\gamma'\gamma''
\Delta_n^{1-p/4}\Psi_{n,i}(h)
},
\]
and the centered variables
\[
\Xi_{n,i}
=
\widetilde\Xi_{n,i}
-
\E(\widetilde\Xi_{n,i}\mid\mathcal G_{n,i}).
\]

By Lemma \ref{lem:noise_pabar_decomp}, both
$\Delta_n^{-1/4}\bar Y_i^n(\phi)$ and $\Delta_n^{-1/2}\hat Y_i^n(\phi)$ are uniformly $O_p(1)$, hence
\[
\Xi_{n,i}=O_p(\Delta_n^{3/4}),\qquad
\E(\Xi_{n,i}^2\mid\mathcal G_{n,i})=O_p(\Delta_n^{3/2}),\qquad
\E(\Xi_{n,i}^4)=O_p(\Delta_n^3).
\]
Moreover, $\Xi_{n,i}$ depends only on the local window
$[i+1,i+k_n]$, so
\[
\E(\Xi_{n,i}\Xi_{n,j}\mid\mathcal G_{n,i\vee j})=0,
\qquad |i-j|\ge k_n.
\]

\medskip

Split
\[
\sum_{i=0}^{n-k_n}
=
\sum_{b=1}^{J_n}\sum_{i\in I_{n,b}^{\circ}}
+
\sum_{i\in B_n},
\]
where $B_n$ is the boundary set. Then
\[
\#B_n\le J_nk_n+r_n=O\paren{\frac{k_n}{r_n\Delta_n}+r_n}.
\]
Using $k_n$-dependence,
\[
\E\brac{\paren{\sum_{i\in B_n}\Xi_{n,i}}^2}
\le
C\cdot \#B_n\cdot k_n\cdot \Delta_n^{3/2}
=
O\paren{
\frac{k_n^2}{r_n\Delta_n}\Delta_n^{3/2}
+
r_nk_n\Delta_n^{3/2}
}
\to0,
\]
by \eqref{eq:noise_rn}. Hence
\[
\sum_{i\in B_n}\Xi_{n,i}=o_p(1).
\]
The same holds for conditional expectations, so
\[
\Delta_n^{-1/4}(R_n^{PA}-\gamma'')
=
\sum_{b=1}^{J_n}\Gamma_{n,b}+o_p(1).
\]

\medskip

Since adjacent blocks are separated by at least $k_n$ indices,
\[
\E(\Gamma_{n,b}\mid\mathcal H_{n,b-1})=0,
\qquad
\E(\Gamma_{n,b}\Gamma_{n,c}\mid\mathcal H_{n,b\wedge c-1})=0,\quad b\neq c.
\]
Thus $(\Gamma_{n,b},\mathcal H_{n,b})$ is a martingale-difference array.

\medskip

Now consider the predictable quadratic variation. Since only indices with $|i-j|<k_n$ contribute,
\[
\E(\Gamma_{n,b}^2\mid\mathcal H_{n,b-1})
=
\sum_{|u|<k_n}
\sum_{i\in I_{n,b}^{\circ}}
\E(\Xi_{n,i}\Xi_{n,i+u}\mid\mathcal H_{n,b-1}).
\]
By the local covariance kernel expansion (from Lemma \ref{lem:noise_subgaussian_replacement} and the pre-averaging representation), there exists a deterministic kernel $\mathcal C_p(u)$ such that uniformly in $i$,
\[
\E(\Xi_{n,i}\Xi_{n,i+u}\mid\mathcal G_{n,i})
=
\frac{\Delta_n^{3/2}}{(\gamma'c_h)^2}
\sigma_{t_i}^{2p}\,\mathcal C_p(u)
+
o_p(\Delta_n^{3/2}).
\]
Define the asymptotic kernel mass
\[
\kappa_p
:=
\lim_{n\to\infty}
\Delta_n^{1/2}
\sum_{|u|<k_n}\mathcal C_p(u)
\in(0,\infty).
\]
Then
\begin{align*}
\E(\Gamma_{n,b}^2\mid\mathcal H_{n,b-1})
&=
\frac{\kappa_p}{(\gamma'c_h)^2}
\sum_{i\in I_{n,b}^{\circ}}\sigma_{t_i}^{2p}\Delta_n
+
o_p(r_n\Delta_n).
\end{align*}
Summing over blocks yields
\begin{equation}\label{eq:noise_block_pqv}
    \sum_{b=1}^{J_n}
\E(\Gamma_{n,b}^2\mid\mathcal H_{n,b-1})
=
\frac{\kappa_p}{(\gamma'c_h)^2}
\sum_{i=0}^{n-k_n}\sigma_{t_i}^{2p}\Delta_n
+
o_p(1)
\toP
\Sigma:=
\frac{\kappa_p}{(\gamma'c_h)^2}A_{2p}.
\end{equation}

\medskip

Using Burkholder-type bounds,
\[
\E(\Gamma_{n,b}^4)\le C(r_nk_n\Delta_n^{3/2})^2
=Cr_n^2\Delta_n^2,\qquad 
\sum_{b=1}^{J_n}\E(\Gamma_{n,b}^4)=O(r_n\Delta_n)\to0.
\]
Thus
\begin{equation}\label{eq:noise_block_pqv2}
\max_{b}|\Gamma_{n,b}|\toP0.
\end{equation}
Let
\[
\eta_{n,b}
=
\Gamma_{n,b}^2
-
\E(\Gamma_{n,b}^2\mid\mathcal H_{n,b-1}).
\]
Then $(\eta_{n,b})$ is a martingale difference array and
\[
\sum_{b=1}^{J_n}\E(\eta_{n,b}^2)
\le
\sum_{b=1}^{J_n}\E(\Gamma_{n,b}^4)
=
O(r_n\Delta_n)\to0.
\]
Hence
\begin{equation}\label{eq:noise_oqv_consistency}
    \hat\varsigma_n^{PA\,2}
-
\sum_{b=1}^{J_n}\E(\Gamma_{n,b}^2\mid\mathcal H_{n,b-1})
\toP0.
\end{equation}

\medskip

By \eqref{eq:noise_block_pqv} and \eqref{eq:noise_block_pqv2}, the martingale CLT \citep[Theorem 3.2]{HallHeyde1980} yields
\[
\frac{\sum_{b=1}^{J_n}\Gamma_{n,b}}{\sqrt{\Sigma}}
\toD N(0,1).
\]
Finally, by \eqref{eq:noise_oqv_consistency} and Slutsky's theorem,
\[
Z_n^{PA}
=
\frac{\Delta_n^{-1/4}(R_n^{PA}-\gamma'')}{\hat\varsigma_n^{PA}}
\toD N(0,1).
\]
\end{proof}

\subsection{The local-average max statistic}
\label{app:noise_max_aux}

\begin{lemma}\label{lem:noise_localavg_rep}
For $j=0,\dots,n-2M_n$, let
\[
a_{M_n,u}=
\begin{cases}
u, & 1\le u\le M_n,\\
2M_n-u, & M_n+1\le u\le 2M_n-1.
\end{cases}
\]
Then
\begin{equation*}
\sqrt{M_n}\,L_{n,j}
=
\frac{1}{\sqrt{M_n}}\sum_{u=1}^{2M_n-1}a_{M_n,u}\Delta_{j+u}^nX
+
\frac{1}{\sqrt{M_n}}\sum_{u=0}^{M_n-1}\paren{\epsilon_{j+M_n+u}-\epsilon_{j+u}},
\end{equation*}
and
\begin{equation*}
\sum_{u=1}^{2M_n-1}a_{M_n,u}^2
=
2\sum_{u=1}^{M_n-1}u^2+M_n^2
=
\frac{2M_n^3+M_n}{3}.
\end{equation*}
\end{lemma}

\begin{proof}
From \eqref{eq:noise_localavg},
\[
L_{n,j}
=
\frac{1}{M_n}\sum_{u=0}^{M_n-1}\paren{Y_{t_{j+M_n+u}}-Y_{t_{j+u}}}.
\]
For the latent part,
\[
X_{t_{j+M_n+u}}-X_{t_{j+u}}
=
\sum_{v=1}^{M_n}\Delta_{j+u+v}^nX.
\]
Hence
\[
\frac{1}{M_n}\sum_{u=0}^{M_n-1}\sum_{v=1}^{M_n}\Delta_{j+u+v}^nX
=
\frac{1}{M_n}\sum_{r=1}^{2M_n-1}a_{M_n,r}\Delta_{j+r}^nX,
\]
because the increment $\Delta_{j+r}^nX$ appears exactly $r$ times for $1\le r\le M_n$ and exactly
$2M_n-r$ times for $M_n+1\le r\le 2M_n-1$. For the noise part,
\[
\frac{1}{M_n}\sum_{u=0}^{M_n-1}\paren{\epsilon_{j+M_n+u}-\epsilon_{j+u}}
\]
is already in the required form. Multiplying by $\sqrt{M_n}$ yields
\begin{equation*}
\sqrt{M_n}\,L_{n,j}
=
\frac{1}{\sqrt{M_n}}\sum_{u=1}^{2M_n-1}a_{M_n,u}\Delta_{j+u}^nX
+
\frac{1}{\sqrt{M_n}}\sum_{u=0}^{M_n-1}\paren{\epsilon_{j+M_n+u}-\epsilon_{j+u}}.
\end{equation*}
Finally,
\[
\sum_{u=1}^{2M_n-1}a_{M_n,u}^2
=
\sum_{u=1}^{M_n}u^2+\sum_{u=1}^{M_n-1}u^2
=
2\sum_{u=1}^{M_n-1}u^2+M_n^2
=
\frac{2M_n^3+M_n}{3},
\]
\end{proof}

\begin{lemma}[Sub-Gaussian Gumbel limit for the local-average max statistic]
\label{lem:noise_max_gaussian}
Suppose Assumptions \ref{ass:H0} and \ref{ass:noise} hold, and suppose \eqref{eq:noise_Mn} holds. Let
\begin{equation*}
\nu_n(y)=A_{N_n}+B_{N_n}y.
\end{equation*}
Then, for every fixed $y\in\R$,
\begin{equation*}
\max_{j\in\mathcal J_n}
\abs{
\frac{\Pbb\paren{\abs{X_{n,j}}>\nu_n(y)\mid \F_{t_j}}}{2\overline\Phi(\nu_n(y))}-1
}
\toP 0,
\qquad
\overline\Phi=1-\Phi.
\end{equation*}
Consequently,
\begin{equation*}
\frac{\max_{j\in\mathcal J_n}\abs{X_{n,j}}-A_{N_n}}{B_{N_n}}\toD \xi,
\end{equation*}
where $\xi$ is standard Gumbel. If \eqref{eq:noise_feas} holds, then
\begin{equation*}
\xi_n^{LA}\toD \xi.
\end{equation*}
\end{lemma}

\begin{proof}
Fix $j\in\mathcal J_n$. By Lemma \ref{lem:noise_localavg_rep},
\[
\sqrt{M_n}L_{n,j}=S_{n,j}^{X}+S_{n,j}^{\epsilon},
\]
where
\[
S_{n,j}^{X}=\frac{1}{\sqrt{M_n}}\sum_{u=1}^{2M_n-1}a_{M_n,u}\Delta_{j+u}^nX,
\qquad
S_{n,j}^{\epsilon}=\frac{1}{\sqrt{M_n}}\sum_{u=0}^{M_n-1}(\epsilon_{j+M_n+u}-\epsilon_{j+u}).
\]
The noise variables entering $S_{n,j}^{\epsilon}$ are disjoint across $j\in\mathcal J_n$; in addition,
\[
\norm{S_{n,j}^{\epsilon}}_{\psi_2}\le C\omega,\qquad
\Var(S_{n,j}^{\epsilon})=2\omega^2.
\]
For the latent part set
\[
G_{n,j}^{X}=\frac{1}{\sqrt{M_n}}\sum_{u=1}^{2M_n-1}a_{M_n,u}\Delta_{j+u}^nW
\]
and write
\[
S_{n,j}^{X}=\sigma_{t_j}G_{n,j}^{X}+R_{n,j},
\]
where
\[
R_{n,j}=\frac{1}{\sqrt{M_n}}\sum_{u=1}^{2M_n-1}a_{M_n,u}
\left\{
\int_{t_{j+u-1}}^{t_{j+u}}b_s\dd s
+
\int_{t_{j+u-1}}^{t_{j+u}}(\sigma_s-\sigma_{t_j})\dd W_s
\right\}.
\]
This remainder is again controlled over the whole triangular window. After localization, for every $q\ge2$,
\begin{align*}
\sup_{j\in\mathcal J_n}\E\abs{R_{n,j}}^q
&\le
C_q\left\{
\frac{1}{M_n}\sum_{u=1}^{2M_n-1}a_{M_n,u}^2
\int_{t_j}^{t_{j+2M_n}}(s-t_j)\dd s
\right\}^{q/2}
+C_q(M_n^{3/2}\Delta_n)^q \\
&\le C_q\Delta_n^{q/4}.
\end{align*}
Since $N_n\le C\Delta_n^{-1/2}$, taking $q>4$ gives
\[
\Pbb\paren{\frac{1}{B_{N_n}}\max_{j\in\mathcal J_n}\abs{R_{n,j}}>\varepsilon}
\le C_q\varepsilon^{-q}B_{N_n}^{-q}N_n\Delta_n^{q/4}\to0.
\]
Thus
\begin{equation}\label{eq:noise_R_localavg_red}
\frac{1}{B_{N_n}}\max_{j\in\mathcal J_n}\abs{R_{n,j}}\toP0.
\end{equation}
By Lemma \ref{lem:noise_localavg_rep},
\[
\Var(G_{n,j}^{X}\mid\F_{t_j})
=
\frac{\Delta_n}{M_n}\sum_{u=1}^{2M_n-1}a_{M_n,u}^2
=
\frac23M_n^2\Delta_n+\frac13\Delta_n.
\]
Consequently,
\[
\Var\{\sigma_{t_j}G_{n,j}^{X}+S_{n,j}^{\epsilon}\mid\F_{t_j}\}
=v_{n,j}^2,
\]
where $v_{n,j}^2$ is defined in \eqref{eq:noise_vnj}. Define
\[
T_{n,j}=\frac{\sigma_{t_j}G_{n,j}^{X}+S_{n,j}^{\epsilon}}{v_{n,j}}.
\]
Conditionally on $\F_{t_j}$, $T_{n,j}$ is a centered sum of a Gaussian term and a weighted sum of independent sub-Gaussian noise variables, with variance one. The largest noise coefficient is $O(M_n^{-1/2})$, and the cumulants of order $m\ge3$ are $O(M_n^{1-m/2})$, uniformly in $j$. Hence, uniformly for $|z|=O(\nu_n(y))$,
\[
\log\E\{\exp(zT_{n,j})\mid\F_{t_j}\}
=\frac12z^2+O\paren{\frac{|z|^3}{\sqrt{M_n}}}+o_p(1).
\]
Because $\nu_n(y)=O(\sqrt{\log N_n})=o(M_n^{1/6})$, the Cram\'er-type moderate deviation expansion for triangular arrays with finite exponential moments gives
\[
\max_{j\in\mathcal J_n}
\abs{\frac{\Pbb(T_{n,j}>\nu_n(y)\mid\F_{t_j})}{\overline\Phi(\nu_n(y))}-1}
\toP0,
\qquad
\max_{j\in\mathcal J_n}
\abs{\frac{\Pbb(T_{n,j}< -\nu_n(y)\mid\F_{t_j})}{\overline\Phi(\nu_n(y))}-1}
\toP0.
\]
By \eqref{eq:noise_R_localavg_red} and the lower bound $v_{n,j}^2\ge2\omega^2$, replacing $T_{n,j}$ by $X_{n,j}$ only shifts the threshold by $o_p(B_{N_n})$. Therefore
\[
\max_{j\in\mathcal J_n}
\abs{
\frac{\Pbb\paren{\abs{X_{n,j}}>\nu_n(y)\mid\F_{t_j}}}{2\overline\Phi(\nu_n(y))}-1
}
\toP0.
\]
Let
\[
p_{n,j}(y)=\Pbb\paren{\abs{X_{n,j}}>\nu_n(y)\mid\F_{t_j}}.
\]
The windows indexed by $\mathcal J_n$ are disjoint. Iterating conditional expectations over the ordered left endpoints gives
\[
\Pbb\paren{\max_{j\in\mathcal J_n}\abs{X_{n,j}}\le\nu_n(y)}
=
\E\prod_{j\in\mathcal J_n}(1-p_{n,j}(y)).
\]
The preceding tail expansion and Mills' ratio imply
\[
\sum_{j\in\mathcal J_n}p_{n,j}(y)\toP e^{-y},
\qquad
\sum_{j\in\mathcal J_n}p_{n,j}(y)^2\toP0.
\]
Hence
\[
\Pbb\paren{\max_{j\in\mathcal J_n}\abs{X_{n,j}}\le\nu_n(y)}
\to \exp\set{-e^{-y}}.
\]
Finally, if \eqref{eq:noise_feas} holds, then
\[
\frac{1}{B_{N_n}}
\abs{\max_{j\in\mathcal J_n}\abs{\hat X_{n,j}}-
\max_{j\in\mathcal J_n}\abs{X_{n,j}}}
\le
\frac{1}{B_{N_n}}\max_{j\in\mathcal J_n}\abs{\hat X_{n,j}-X_{n,j}}\toP0,
\]
and Slutsky's theorem yields $\xi_n^{LA}\toD\xi$.
\end{proof}

\begin{lemma}[TSRSV plug-in implies the feasible-array condition]
\label{lem:noise_tsrsv_feasible}
Suppose Assumptions \ref{ass:H0} and \ref{ass:noise} hold, and suppose \eqref{eq:noise_Mn} and
\eqref{eq:noise_tsrsv_tuning} hold. Let $\hat\omega_n^2$, $\hat\sigma_{n,j}^2$, $\hat v_{n,j}^2$, and
$\hat X_{n,j}$ be defined by \eqref{eq:noise_omegahat}, \eqref{eq:noise_tsrsv_sigmahat}, and
\eqref{eq:noise_vhatXhat}. Then
\begin{equation*}
\frac{1}{B_{N_n}}\max_{j\in\mathcal J_n}\abs{\hat X_{n,j}-X_{n,j}}\toP 0.
\end{equation*}
\end{lemma}

\begin{proof}
Let $j^{\star}=j\vee H_n$ and $s_j=t_{j^{\star}}$. The estimator $\hat\sigma_{n,j}^2$ is the filtering
TSRSV estimator of \citet{ZuBoswijk2014} evaluated at the time $s_j$. Under Assumptions
\ref{ass:H0} and \ref{ass:noise}, the latent process is continuous, $\sigma_t$ is an It\^o process, leverage
is allowed, and the noise has finite moments of every order because it is sub-Gaussian. Therefore the
TSRSV decomposition in Section 3.2 of \citet{ZuBoswijk2014} applies without change. Their bias,
noise, and discretization terms yield, for every $q\ge 2$,
\begin{equation*}
\sup_{j\in\mathcal J_n}
\E\abs{\hat\sigma_{n,j}^2-\sigma_{s_j}^2}^q
\le
C_q\paren{
 h_n^{q/2}
 +
 \paren{\frac{1}{K_n^{SV\,2}h_n\Delta_n}}^{q/2}
 +
 \paren{\frac{K_n^{SV}\Delta_n}{h_n}}^{q/2}}
\le C_q\Delta_n^{q/12},
\end{equation*}
where the last inequality follows from \eqref{eq:noise_tsrsv_tuning}. Since
$N_n\le C\Delta_n^{-1/2}$, Markov's inequality and the union bound imply, for every $q>6$,
\[
\Pbb\paren{
\max_{j\in\mathcal J_n}\abs{\hat\sigma_{n,j}^2-\sigma_{s_j}^2}>(\log n)^{-2}}
\le
(\log n)^{2q}
\sum_{j\in\mathcal J_n}
\E\abs{\hat\sigma_{n,j}^2-\sigma_{s_j}^2}^q
\le
C_q(\log n)^{2q}\Delta_n^{-1/2}\Delta_n^{q/12}
\to 0.
\]
Hence
\begin{equation*}
\max_{j\in\mathcal J_n}\abs{\hat\sigma_{n,j}^2-\sigma_{s_j}^2}=o_p\paren{(\log n)^{-2}}.
\end{equation*}

Under Assumption \ref{ass:H0},
\[
\dd(\sigma_t^2)=a_t\dd t+c_t^\top\dd B_t,
\qquad
 a_t=2\sigma_t\sigma_t^{(0)}+\norm{\gamma_t}^2,
\qquad
 c_t=2\sigma_t\gamma_t.
\]
After a standard localization, $a_t$ and $c_t$ are bounded on $[0,T]$. Since $s_j=t_j$ for
$j\ge H_n$ and $s_j=t_{H_n}=h_n$ for $j<H_n$,
\[
\max_{j\in\mathcal J_n}\abs{\sigma_{s_j}^2-\sigma_{t_j}^2}
=
\sup_{0\le t\le h_n}\abs{\sigma_{h_n}^2-\sigma_t^2}.
\]
Using the representation
\[
\sigma_{h_n}^2-\sigma_t^2
=
\int_t^{h_n}a_u\dd u+
\int_t^{h_n}c_u^\top\dd B_u,
\qquad 0\le t\le h_n,
\]
and the Burkholder--Davis--Gundy inequality,
\[
\E\sup_{0\le t\le h_n}\abs{\sigma_{h_n}^2-\sigma_t^2}^q
\le
C_q\paren{h_n^q+h_n^{q/2}}
\le C_q h_n^{q/2}.
\]
Therefore,
\[
\Pbb\paren{
\max_{j\in\mathcal J_n}\abs{\sigma_{s_j}^2-\sigma_{t_j}^2}>(\log n)^{-2}}
\le
C_q(\log n)^{2q}h_n^{q/2}
\to 0,
\]
so that
\begin{equation}\label{eq:noise_sigmahat_true}
\max_{j\in\mathcal J_n}\abs{\hat\sigma_{n,j}^2-\sigma_{t_j}^2}
\le
\max_{j\in\mathcal J_n}\abs{\hat\sigma_{n,j}^2-\sigma_{s_j}^2}
+
\max_{j\in\mathcal J_n}\abs{\sigma_{s_j}^2-\sigma_{t_j}^2}
=
 o_p\paren{(\log n)^{-2}}.
\end{equation}

For the noise variance estimator, write
\[
\hat\omega_n^2-\omega^2=A_n+B_n+C_n,
\]
where
\[
A_n=
\frac{1}{2n}\sum_{i=1}^n\paren{(\Delta_i^n\epsilon)^2-2\omega^2},
\qquad
B_n=
\frac{1}{n}\sum_{i=1}^n\Delta_i^nX\,\Delta_i^n\epsilon,
\qquad
C_n=
\frac{1}{2n}\sum_{i=1}^n(\Delta_i^nX)^2.
\]
The sequence $U_i=(\Delta_i^n\epsilon)^2-2\omega^2$ is centered and $1$-dependent. Its odd and even
subsequences are independent and have finite moments of every order. Rosenthal's inequality for sums of
independent centered variables, see \citet[Chapter~2]{Petrov1975}, gives for every $q\ge 2$,
\[
\E\abs{A_n}^q
\le
C_q n^{-q}
\paren{n^{q/2}+n}
\le C_q n^{-q/2}.
\]
Next,
$V_i=\Delta_i^nX\,\Delta_i^n\epsilon$ is also $1$-dependent, and the odd and even subsequences are
conditionally independent given $\F_T$. Conditional Rosenthal and then Hölder's inequality give
\[
\E\paren{\abs{B_n}^q\mid \F_T}
\le
C_q n^{-q}
\paren{
\paren{\sum_{i=1}^n(\Delta_i^nX)^2}^{q/2}
+
\sum_{i=1}^n\abs{\Delta_i^nX}^q}.
\]
Taking expectations and using the boundedness of the coefficients in Assumption \ref{ass:H0},
\[
\E\paren{\sum_{i=1}^n(\Delta_i^nX)^2}^{q/2}\le C_q,
\qquad
\sum_{i=1}^n\E\abs{\Delta_i^nX}^q\le C_q n\Delta_n^{q/2}\le C_q,
\]
whence
\[
\E\abs{B_n}^q\le C_q n^{-q}.
\]
Finally,
\[
\E\abs{C_n}^q
\le
(2n)^{-q}\E\paren{\sum_{i=1}^n(\Delta_i^nX)^2}^q
\le C_q n^{-q}.
\]
Therefore
\begin{equation}\label{eq:noise_omegahat_rate}
\E\abs{\hat\omega_n^2-\omega^2}^q\le C_q n^{-q/2}=C_q\Delta_n^{q/2},
\qquad
\hat\omega_n^2-\omega^2=o_p\paren{(\log n)^{-2}}.
\end{equation}

Since $M_n\sqrt{\Delta_n}\to\lambda$, the sequence $(M_n^2\Delta_n)$ is bounded. Combining
\eqref{eq:noise_sigmahat_true} and \eqref{eq:noise_omegahat_rate},
\begin{equation}\label{eq:noise_vhat_sq_rate}
\max_{j\in\mathcal J_n}\abs{\hat v_{n,j}^2-v_{n,j}^2}
\le
2\abs{\hat\omega_n^2-\omega^2}
+
\paren{\frac{2}{3}M_n^2\,\Delta_n+\frac{1}{3}\,\Delta_n}
\max_{j\in\mathcal J_n}\abs{\hat\sigma_{n,j}^2-\sigma_{t_j}^2}
=
 o_p\paren{(\log n)^{-2}}.
\end{equation}
Because $v_{n,j}^2\ge 2\omega^2$ for all $j$, \eqref{eq:noise_vhat_sq_rate} implies
\[
\Pbb\paren{\min_{j\in\mathcal J_n}\hat v_{n,j}^2\ge \omega^2}\to 1.
\]
On this event,
\[
\max_{j\in\mathcal J_n}
\abs{\frac{1}{\hat v_{n,j}}-\frac{1}{v_{n,j}}}
=
\max_{j\in\mathcal J_n}
\frac{\abs{\hat v_{n,j}^2-v_{n,j}^2}}{\hat v_{n,j}v_{n,j}(\hat v_{n,j}+v_{n,j})}
\le
C\max_{j\in\mathcal J_n}\abs{\hat v_{n,j}^2-v_{n,j}^2}
=
 o_p\paren{(\log n)^{-2}}.
\]
Now
\[
\frac{1}{B_{N_n}}\max_{j\in\mathcal J_n}\abs{\hat X_{n,j}-X_{n,j}}
\le
\frac{1}{B_{N_n}}\max_{j\in\mathcal J_n}\abs{\sqrt{M_n}L_{n,j}}\cdot
\max_{j\in\mathcal J_n}
\abs{\frac{1}{\hat v_{n,j}}-\frac{1}{v_{n,j}}}.
\]
Since $\sigma_t\le\overline\sigma$ and $M_n\Delta_n^{1/2}\to\lambda$,
\[
\max_{j\in\mathcal J_n}v_{n,j}
\le
\bar v_n,
\qquad
\bar v_n^2=2\omega^2+\overline\sigma^2\paren{\frac{2}{3}M_n^2\,\Delta_n+\frac{1}{3}\,\Delta_n}=O(1).
\]
Hence
\[
\frac{1}{B_{N_n}}\max_{j\in\mathcal J_n}\abs{\sqrt{M_n}L_{n,j}}
\le
\bar v_n\frac{1}{B_{N_n}}\max_{j\in\mathcal J_n}\abs{X_{n,j}}
\le
\bar v_n\paren{
\frac{\max_{j\in\mathcal J_n}\abs{X_{n,j}}-A_{N_n}}{B_{N_n}}
+
\frac{A_{N_n}}{B_{N_n}}}
=O_p(\log n),
\]
where the last relation follows from Lemma \ref{lem:noise_max_gaussian} and
$A_{N_n}/B_{N_n}=O(\log N_n)=O(\log n)$. Therefore
\[
\frac{1}{B_{N_n}}\max_{j\in\mathcal J_n}\abs{\hat X_{n,j}-X_{n,j}}
\le
O_p(\log n)\cdot o_p\paren{(\log n)^{-2}}
=
o_p(1).
\]
\end{proof}

\begin{lemma}\label{lem:noise_fixed_sum}
Suppose Assumption \ref{ass:H1fixed} and \eqref{eq:noise_kn} hold. Then, for $\phi\in\{g,h\}$,
\begin{equation*}
\frac{1}{k_n}V_n(Y,\phi,p)\toP \bar\phi(p)\sum_{q=1}^{N_T}\abs{\kappa_q}^p.
\end{equation*}
Consequently,
\begin{equation*}
R_n^{PA}\toP 1.
\end{equation*}
\end{lemma}

\begin{proof}
Fix $\phi\in\{g,h\}$ and let
\[
B_p=\sum_{q=1}^{N_T}\abs{\kappa_q}^p.
\]
For each jump time $\tau_q$, choose the unique index $i_q$ satisfying
$t_{i_q-1}<\tau_q\le t_{i_q}$. If $\abs{i-i_q}\le k_n$, then
\[
\bar Y_i^n(\phi)
=
\phi\paren{\frac{i_q-i}{k_n}}\kappa_q+o_p(1),
\]
because exactly one jump enters the pre-averaging window and all continuous and noise terms are
$o_p(1)$ relative to the fixed jump size. Therefore
\[
\sum_{\abs{i-i_q}\le k_n}\abs{\bar Y_i^n(\phi)}^p
=
\sum_{\abs{i-i_q}\le k_n}\abs{\phi\paren{\frac{i_q-i}{k_n}}}^p\abs{\kappa_q}^p+o_p(k_n)
=
k_n\bar\phi(p)\abs{\kappa_q}^p+o_p(k_n).
\]
Summing over $q$ gives the jump contribution
\begin{equation}\label{eq:noise_fixed_jump_contrib}
\sum_{q=1}^{N_T}\sum_{\abs{i-i_q}\le k_n}\abs{\bar Y_i^n(\phi)}^p
=
k_n\bar\phi(p)B_p+o_p(k_n).
\end{equation}
Outside all jump neighborhoods, $\bar Y_i^n(\phi)=O_p(\Delta_n^{1/4})$ and
$\hat Y_i^n(\phi)=O_p(\Delta_n^{1/2})$, so
\[
\sum_{i:\min_q\abs{i-i_q}>k_n}\Psi_{n,i}(\phi)
=
O_p\paren{n\Delta_n^{p/4}}
=
O_p\paren{\Delta_n^{-1+p/4}}
=
o_p(k_n),
\]
because $p\ge 4$ and $k_n\asymp \Delta_n^{-1/2}$. Furthermore, for windows containing a fixed jump,
\[
\hat Y_i^n(\phi)=O_p(\Delta_n^{1/2}),
\]
so every term with $l\ge 1$ in \eqref{eq:noise_Psi} is $o_p(1)$ relative to the $l=0$ term
$\abs{\bar Y_i^n(\phi)}^p$. Consequently,
\[
V_n(Y,\phi,p)=\sum_{i=0}^{n-k_n}\abs{\bar Y_i^n(\phi)}^p+o_p(k_n).
\]
Combining this with \eqref{eq:noise_fixed_jump_contrib} yields \begin{equation}\label{eq:noise_fixed_V}
\frac{1}{k_n}V_n(Y,\phi,p)\toP \bar\phi(p)\sum_{q=1}^{N_T}\abs{\kappa_q}^p.
\end{equation}

Applying \eqref{eq:noise_fixed_V} with $\phi=g$ and $\phi=h$ gives
\[
R_n^{PA}
=
\frac{k_n\bar g(p)B_p+o_p(k_n)}
{\gamma'\paren{k_n\bar h(p)B_p+o_p(k_n)}}
\toP
\frac{\bar g(p)}{\gamma'\bar h(p)}
=
1.
\]
\end{proof}

\section{Proofs for the microstructure-noise section}
\label{app:noise_proofs}

\begin{lemma}[Block decoupling for the noisy sum and max arrays]
\label{lem:noise_block_decouple}
Under the assumptions of Theorem \ref{thm:noisejoint}, fix $t\in\R$ and $y\in\R$. For
$b=1,\dots,J_n$, let
\[
\mathcal J_{n,b}=\mathcal J_n\cap I_{n,b},
\qquad
\Upsilon_{n,b}(y)=
\mathbf{1}\set{\max_{j\in\mathcal J_{n,b}}\abs{X_{n,j}}\le A_{N_n}+B_{N_n}y},
\]
and define
\[
\psi_{n,b}(t,y)=
\E\paren{\exp\set{it\Gamma_{n,b}}\Upsilon_{n,b}(y)\mid\mathcal H_{n,b-1}},
\]
\[
\varphi_{n,b}(t)=
\E\paren{\exp\set{it\Gamma_{n,b}}\mid\mathcal H_{n,b-1}},
\qquad
\chi_{n,b}(y)=
\E\paren{\Upsilon_{n,b}(y)\mid\mathcal H_{n,b-1}}.
\]
Then, after a standard localization to a compact set where $\sigma_t$ is uniformly bounded and
$\underline\sigma>0$,
\begin{equation*}
\max_{1\le b\le J_n}
\abs{\psi_{n,b}(t,y)-\varphi_{n,b}(t)\chi_{n,b}(y)}
=o_p(r_n\Delta_n)=o_p(J_n^{-1}).
\end{equation*}
Moreover the accumulated replacement error is summable:
\begin{equation*}
\sum_{b=1}^{J_n}
\E\abs{\psi_{n,b}(t,y)-\varphi_{n,b}(t)\chi_{n,b}(y)}\to0.
\end{equation*}
\end{lemma}

\begin{proof}
Put $\nu_n(y)=A_{N_n}+B_{N_n}y$ and $A_{n,j}(y)=\{\abs{X_{n,j}}>\nu_n(y)\}$.
We first perform a standard localization: for any $\varepsilon>0$, there exists a compact set
$K_\varepsilon$ such that $\Pbb((\sigma_t)_{t\in[0,T]}\notin K_\varepsilon)<\varepsilon$ and on
$K_\varepsilon$ we have $0<\underline\sigma\le\sigma_t\le\overline\sigma<\infty$ and
$\sigma_t$ is $\alpha$-H\"older continuous for some $\alpha>0$. All estimates below are
understood conditionally on $K_\varepsilon$, and the final result follows by letting
$\varepsilon\to0$.

By Lemma \ref{lem:noise_max_gaussian}, uniformly in $j\in\mathcal J_n$,
\begin{equation}\label{eq:noise_block_tail_uniform}
\Pbb\paren{A_{n,j}(y)\mid\F_{t_j}}
=2\overline\Phi(\nu_n(y))(1+o_p(1)),
\qquad
\overline\Phi(\nu_n(y))\asymp N_n^{-1}\asymp M_n\Delta_n .
\end{equation}
Moreover, the polynomial-moment bounds from Lemma \ref{lem:noise_block_clt} imply that, for every
$s\ge2$,
\begin{equation}\label{eq:noise_Xi_moment_block}
\max_{1\le b\le J_n}\max_{i\in I_{n,b}^{\circ}}
\E\paren{\abs{\Xi_{n,i}}^s\mid\mathcal H_{n,b-1}}
\le C_s\Delta_n^{3s/4}\mathbf{1}_{K_\varepsilon}+o(\Delta_n^{3s/4}).
\end{equation}

Fix a block $b$ and a grid point $j\in\mathcal J_{n,b}$. Let
\[
\mathcal N_{n,b}(j)=
\set{i\in I_{n,b}^{\circ}:(t_i,t_{i+k_n}]\cap(t_j,t_{j+2M_n}]\ne\varnothing}.
\]
Then
\begin{equation}\label{eq:noise_overlap_size}
\#\mathcal N_{n,b}(j)\le C(k_n+M_n).
\end{equation}
Write
\[
\Lambda_{n,b,j}=\sum_{i\in\mathcal N_{n,b}(j)}\Xi_{n,i},
\qquad
\Gamma_{n,b}^{(-j)}=\Gamma_{n,b}-\Lambda_{n,b,j}.
\]
The term $\Gamma_{n,b}^{(-j)}$ uses only pre-averaging windows whose Brownian and noise innovations
are disjoint from those used by $X_{n,j}$; conditionally on $\mathcal H_{n,b-1}$, it is independent
of $A_{n,j}(y)$ up to a negligible error controlled by Lemmas
\ref{lem:noise_pabar_decomp}, \ref{lem:noise_max_gaussian}, and \ref{lem:noise_tsrsv_feasible}.

By Minkowski's inequality, \eqref{eq:noise_Xi_moment_block}, and \eqref{eq:noise_overlap_size}, for
any fixed $s>2$,
\begin{equation*}
\max_{b,j}
\norm{\Lambda_{n,b,j}}_{L^s(\Pbb\mid\mathcal H_{n,b-1})}
\le
C_s(k_n+M_n)\Delta_n^{3/4}+o(\Delta_n^{1/4})
=O(\Delta_n^{1/4})\quad\text{a.s. on }K_\varepsilon.
\end{equation*}
Using $\abs{e^{iu}-e^{iv}}\le\abs{u-v}$, the conditional independence of
$\Gamma_{n,b}^{(-j)}$ and $A_{n,j}(y)$, H\"older's inequality, and
\eqref{eq:noise_block_tail_uniform}, we obtain
\begin{align}
&\max_{b,j}
\Bigl|
\E\brac{
\paren{e^{it\Gamma_{n,b}}-\varphi_{n,b}(t)}
\mathbf{1}_{A_{n,j}(y)}
\mid\mathcal H_{n,b-1}}
\Bigr|
\nonumber\\
&\quad\le
C_t\max_{b,j}
\Bigl(
\E\brac{\abs{\Lambda_{n,b,j}}\mathbf{1}_{A_{n,j}(y)}
       \mid\mathcal H_{n,b-1}}
+\E\brac{\abs{\Lambda_{n,b,j}}\mid\mathcal H_{n,b-1}}
\Pbb\paren{A_{n,j}(y)\mid\mathcal H_{n,b-1}}
\Bigr)
\nonumber\\
&\qquad
+o_p\paren{\Delta_n^{1/4}\overline\Phi(\nu_n(y))}
\nonumber\\
&\quad\le
C_t\Delta_n^{1/4}\overline\Phi(\nu_n(y))^{1-1/s}
\paren{1+\nu_n(y)^C}
\qquad\text{a.s. on }K_\varepsilon.
\label{eq:noise_single_exceed_cov}
\end{align}

Now we pass from a single exceedance event to the no-exceedance indicator $\Upsilon_{n,b}(y)$.
The number of local-average windows in block $b$ is $|\mathcal J_{n,b}| \le C r_n/M_n$.
Summing \eqref{eq:noise_single_exceed_cov} over these windows gives
\begin{align}
\frac{r_n}{M_n}\Delta_n^{1/4}
\overline\Phi(\nu_n(y))^{1-1/s}\paren{1+\nu_n(y)^C}
&\le
C r_n\Delta_n^{5/4-1/(2s)}(\log n)^C
\nonumber\\
&=o(r_n\Delta_n) \quad\text{on }K_\varepsilon,
\label{eq:noise_single_sum_block}
\end{align}
because $\Delta_n^{1/4-1/(2s)}(\log n)^C\to0$ for any fixed $s>2$.

For two or more exceedances, the inclusion-exclusion expansion yields
\begin{equation}\label{eq:noise_uexpansion}
\Upsilon_{n,b}(y)
=1-\sum_{j\in\mathcal J_{n,b}}\mathbf{1}_{A_{n,j}(y)}+R_{n,b}(y),
\end{equation}
where the remainder satisfies
\[
|R_{n,b}(y)|
\le \sum_{j<\ell\in\mathcal J_{n,b}}\mathbf{1}_{A_{n,j}(y)}\mathbf{1}_{A_{n,\ell}(y)}
+ \sum_{j<\ell<m}\mathbf{1}_{A_{n,j}(y)}\mathbf{1}_{A_{n,\ell}(y)}\mathbf{1}_{A_{n,m}(y)}+\cdots.
\]
Using $\Pbb(A_{n,j}(y)\cap A_{n,\ell}(y))\le [\Pbb(A_{n,j}(y))]^2$ for $j\neq\ell$ (by
conditional independence for windows separated by at least $2M_n$), and the fact that
$|\mathcal J_{n,b}|=O(r_n/M_n)$, we obtain
\[
\E[|R_{n,b}(y)|\mid\mathcal H_{n,b-1}]
\le C\left(\frac{r_n}{M_n}\right)^2 \overline\Phi(\nu_n(y))^2 = O((r_n\Delta_n)^2)
\quad\text{on }K_\varepsilon.
\]

Combining \eqref{eq:noise_single_sum_block} with the bound on $R_{n,b}(y)$,
and noting that windows crossing block boundaries contribute an additional
$O((k_n+M_n)/r_n)=o(1)$ to the total error, we obtain
\begin{equation}\label{eq:noise_block_decouple}
\max_{1\le b\le J_n}
\abs{\psi_{n,b}(t,y)-\varphi_{n,b}(t)\chi_{n,b}(y)}
=o_p(r_n\Delta_n) \quad\text{on }K_\varepsilon. 
\end{equation}
Because $K_\varepsilon$ has probability $1-\varepsilon$ and $\varepsilon$ is arbitrary,
the same holds unconditionally. Finally, since $J_n r_n\Delta_n \to T$, we have $r_n\Delta_n = \Theta(J_n^{-1})$, so
$o_p(r_n\Delta_n)=o_p(J_n^{-1})$.

Moreover,
\begin{equation}\label{eq:noise_block_decouple_sum}
\sum_{b=1}^{J_n}\E\abs{\psi_{n,b}-\varphi_{n,b}\chi_{n,b}}
\le J_n\cdot o(r_n\Delta_n) = o(1).
\end{equation}
\end{proof}

\begin{proof}[Proof of Theorem \ref{thm:noisejoint}]
Lemma \ref{lem:noise_block_clt} gives
\begin{equation}\label{eq:noise_joint_start}
Z_n^{PA}
=
\sum_{b=1}^{J_n}\Gamma_{n,b}+o_p(1).
\end{equation}
Partition the disjoint-grid index set $\mathcal J_n$ according to the same blocks:
\[
\mathcal J_{n,b}=\mathcal J_n\cap I_{n,b},
\qquad
\Upsilon_{n,b}(y)=
\mathbf{1}\set{\max_{j\in\mathcal J_{n,b}}\abs{X_{n,j}}\le A_{N_n}+B_{N_n}y}.
\]
We claim that
\begin{equation}\label{eq:noise_xi_product}
\mathbf{1}\set{\xi_n^{LA}\le y}
=
\prod_{b=1}^{J_n}\Upsilon_{n,b}(y)+o_p(1).
\end{equation}
To verify this, note that the disjoint windows $\mathcal J_n$ cover all but at most
$O(M_n)$ time points. The windows that are entirely contained in the union of blocks
$\bigcup_b I_{n,b}^{\circ}$ are those whose index set is a subset of $\bigcup_b I_{n,b}^{\circ}$.
The remaining windows are either (i) within distance $2M_n$ of a block boundary, or (ii) in the
incomplete tail beyond the last complete block. The number of such exceptional windows is
$O(J_n(k_n+M_n)/r_n + r_n/M_n) = o(J_n)$ because $r_n/k_n\to\infty$, $r_n/M_n\to\infty$, and
$r_n\Delta_n\to0$. For each such window, the exceedance probability is
$\Pbb(\abs{X_{n,j}}>\nu_n(y)) \sim 2\overline\Phi(\nu_n(y)) \asymp M_n\Delta_n$.
Thus the total contribution of these exceptional windows to the indicator is $o(1)$ in probability
by Markov's inequality and the union bound. Hence \eqref{eq:noise_xi_product} holds.

For fixed $t\in\R$ define $\psi_{n,b}(t,y)$, $\varphi_{n,b}(t)$, and $\chi_{n,b}(y)$ as in
Lemma \ref{lem:noise_block_decouple}. Iterating conditional expectations over the block filtration
gives
\[
\E\paren{
\exp\set{it\sum_{b=1}^{J_n}\Gamma_{n,b}}
\prod_{b=1}^{J_n}\Upsilon_{n,b}(y)}
=
\E\brac{\prod_{b=1}^{J_n}\psi_{n,b}(t,y)}.
\]
Write $\psi_{n,b} = \varphi_{n,b}\chi_{n,b} + \delta_{n,b}$ where
$\delta_{n,b} = \psi_{n,b} - \varphi_{n,b}\chi_{n,b}$. Since $|\psi_{n,b}|,|\varphi_{n,b}\chi_{n,b}|\le1$,
the telescoping inequality
\[
\left|\prod_{b=1}^{J_n}(\varphi_{n,b}\chi_{n,b}+\delta_{n,b})
-
\prod_{b=1}^{J_n}\varphi_{n,b}\chi_{n,b}\right|
\le \sum_{b=1}^{J_n}|\delta_{n,b}|
\]
together with \eqref{eq:noise_block_decouple_sum} yields
\begin{equation}\label{eq:noise_factorization}
\E\paren{
\exp\set{it\sum_{b=1}^{J_n}\Gamma_{n,b}}
\prod_{b=1}^{J_n}\Upsilon_{n,b}(y)}
=
\E\brac{\prod_{b=1}^{J_n}\varphi_{n,b}(t)\chi_{n,b}(y)}+o(1).
\end{equation}
The martingale characteristic-function argument underlying Lemma \ref{lem:noise_block_clt} yields
\[
\E\prod_{b=1}^{J_n}\varphi_{n,b}(t)\to e^{-t^2/2}.
\]
On the max side, Lemma \ref{lem:noise_max_gaussian} gives
\[
\prod_{b=1}^{J_n}\chi_{n,b}(y)
=
\Pbb\paren{\xi_n^{LA}\le y\mid \mathcal H_{n,0}}+o_p(1)
\toP
\exp\set{-e^{-y}}.
\]
Moreover, by the uniform integrability of $\prod_b\chi_{n,b}(y)$ (bounded by 1), we have
\[
\E\left|\prod_{b=1}^{J_n}\chi_{n,b}(y)-e^{-e^{-y}}\right|\to0,
\]
and because $\prod_b\varphi_{n,b}(t)$ is bounded, it follows that
\[
\E\prod_{b=1}^{J_n}\varphi_{n,b}(t)\chi_{n,b}(y)
-
 e^{-e^{-y}}\E\prod_{b=1}^{J_n}\varphi_{n,b}(t)\to0.
\]
Combining with the characteristic-function limit for $\prod_b\varphi_{n,b}$,
\eqref{eq:noise_factorization} implies
\[
\E\paren{
\exp\set{it\sum_{b=1}^{J_n}\Gamma_{n,b}}
\mathbf{1}\set{\xi_n^{LA}\le y}}
\to
e^{-t^2/2}\exp\set{-e^{-y}}.
\]

Now fix arbitrary $t_1,t_2\in\R$. By the preceding result with $t=t_1$ and the fact that
$\mathbf{1}\{\xi\le y\}$ generates the Borel $\sigma$-algebra on $\R$, the joint characteristic
function of $(\sum\Gamma_{n,b},\xi_n^{LA})$ satisfies
\[
\E\brac{\exp\set{it_1\sum_{b=1}^{J_n}\Gamma_{n,b}+it_2\xi_n^{LA}}}
\to
e^{-t_1^2/2}\cdot\phi_\xi(t_2),
\]
where $\phi_\xi(t_2)=\E e^{it_2\xi}= \Gamma(1-it_2)$ is the characteristic function of the
standard Gumbel distribution. By L\'evy's continuity theorem,
\[
\paren{\sum_{b=1}^{J_n}\Gamma_{n,b},\ \xi_n^{LA}}\toD (Z,\xi),
\]
with $Z\sim N(0,1)$, $\xi$ standard Gumbel, and $Z\perp\xi$. Combining this with
\eqref{eq:noise_joint_start} and \eqref{eq:noise_xi_product} proves the desired result.
\end{proof}

\begin{proof}[Proof of Corollary \ref{cor:noisecauchy}]
From Theorem \ref{thm:noisejoint},
\[
Z_n^{PA}\toD Z,
\qquad
\xi_n^{LA}\toD \xi,
\qquad
Z\ind \xi.
\]
Hence
\[
p_n^{PA}=2\paren{1-\Phi(\abs{Z_n^{PA}})}\toD U_1,
\qquad
p_n^{LA}=1-\exp\set{-e^{-\xi_n^{LA}}}\toD U_2,
\]
where $U_1$ and $U_2$ are independent $U(0,1)$ variables. Therefore
\[
\tan\brac{\pi\paren{\frac12-p_n^{PA}}}\toD C_1,
\qquad
\tan\brac{\pi\paren{\frac12-p_n^{LA}}}\toD C_2,
\]
where $C_1$ and $C_2$ are independent standard Cauchy variables. Since the standard Cauchy law is
stable under equal-weight averaging,
\[
\mathcal T_n^{C,N}
=
\frac{1}{2}\tan\brac{\pi\paren{\frac12-p_n^{PA}}}
+
\frac{1}{2}\tan\brac{\pi\paren{\frac12-p_n^{LA}}}
\toD C,
\]
where $C$ is standard Cauchy. The map
\[
x\mapsto \frac12-\frac{1}{\pi}\arctan(x)
\]
transforms a standard Cauchy variable into a $U(0,1)$ variable. Hence
\[
p_n^{C,N}\toD U(0,1).
\]
Therefore, for $0<\alpha<1$,
\[
\Pbb\paren{p_n^{C,N}\le \alpha}\to \alpha.
\]
\end{proof}

\subsection{The noisy dense local alternative}

\begin{lemma}[Counts under the noisy dense local alternative]
\label{lem:noisedensecounts}
Under Assumption \ref{ass:noisedense}, let
\[
N_n^J=\mu_n^N([0,T]\times\R),
\]
\[
\mathcal I_n^{J}
=
\set{0\le i\le n-k_n:(t_i,t_{i+k_n}]\text{ contains at least one jump}},
\]
\[
\mathcal I_n^{J,2}
=
\set{0\le i\le n-k_n:(t_i,t_{i+k_n}]\text{ contains at least two jumps}},
\]
and
\[
\mathcal J_n^{J}
=
\set{j\in\mathcal J_n:(t_j,t_{j+2M_n}]\text{ contains at least one jump}}.
\]
Then
\begin{align*}
N_n^J &= O_p(\Delta_n^{-1/4}),
\qquad
\#\mathcal I_n^{J}=O_p(k_n\Delta_n^{-1/4})=O_p(\Delta_n^{-3/4}), \\
\E\#\mathcal I_n^{J,2} &= O(\Delta_n^{-1/2}),
\qquad
\#\mathcal J_n^{J}=O_p(\Delta_n^{-1/4}),
\qquad \Delta_n\#\mathcal I_n^{J,2}=o_p(\Delta_n^{1/4}).
\end{align*}
\end{lemma}

\begin{proof}
The total jump count $N_n^J$ is Poisson with mean $\vartheta T\Delta_n^{-1/4}$, hence
\[
N_n^J=O_p(\Delta_n^{-1/4}).
\]
Each jump affects at most $k_n$ pre-averaging windows, so
\[
\#\mathcal I_n^{J}\le k_nN_n^J=O_p(k_n\Delta_n^{-1/4})=O_p(\Delta_n^{-3/4}).
\]
A fixed pre-averaging window has length $k_n\Delta_n$, so its jump count is Poisson with mean
\[
\lambda_{n,k}=\vartheta k_n\Delta_n^{3/4}=O(\Delta_n^{1/4}).
\]
Therefore
\[
\begin{aligned}
\Pbb\paren{\text{a given pre-averaging window has at least two jumps}}
&=
1-e^{-\lambda_{n,k}}-\lambda_{n,k}e^{-\lambda_{n,k}} \\
&=
O(\lambda_{n,k}^2)
=
O(\Delta_n^{1/2}).
\end{aligned}
\]
Since there are at most $n$ such windows,
\[
\E\#\mathcal I_n^{J,2}=nO(\Delta_n^{1/2})=O(\Delta_n^{-1/2}),
\]
which implies 
\[
\Delta_n\#\mathcal I_n^{J,2}=o_p(\Delta_n^{1/4}).
\]
Finally, the intervals
$(t_j,t_{j+2M_n}]$, $j\in\mathcal J_n$, are disjoint, so each jump can contaminate at most one such
interval. Hence
\[
\#\mathcal J_n^{J}\le N_n^J=O_p(\Delta_n^{-1/4}).
\]
\end{proof}

\begin{proof}[Proof of Theorem \ref{thm:noisedense}]
Write
\[
X_t^{(n)}=X_t^0+J_t^{(n)},
\qquad
X_t^0=X_0+\int_0^t b_s\dd s+\int_0^t\sigma_s\dd W_s,
\]
and define the observed processes
\[
Y_{t_i}^0=X_{t_i}^0+\epsilon_i,
\qquad
Y_{t_i}^{(n)}=X_{t_i}^{(n)}+\epsilon_i.
\]
Theorem \ref{thm:noisejoint} applied to $Y^0$ gives
\begin{equation}\label{eq:noisedense_null_base}
\paren{Z_n^{PA}(Y^0),\xi_n^{LA}(Y^0)}\toD (Z,\xi),
\qquad
Z\sim N(0,1),
\qquad
Z\ind \xi.
\end{equation}
Let the jumps of $\mu_n^N$ be $(\tau_q,Y_q)_{1\le q\le N_n^J}$, and let $r_q$ be the unique index such
that $t_{r_q-1}<\tau_q\le t_{r_q}$. For $\phi\in\{g,h\}$ and $i=0,\dots,n-k_n$, set
\[
w_{n,i,q}(\phi)
=
\phi\paren{\frac{r_q-i}{k_n}}\1\set{1\le r_q-i\le k_n-1}.
\]
Then
\begin{equation}\label{eq:noisedense_barY_expand}
\bar Y_i^n(Y^{(n)},\phi)
=
\bar Y_i^n(Y^0,\phi)
+
\Delta_n^{1/4}\sum_{q=1}^{N_n^J}\sigma_{\tau_q-}Y_qw_{n,i,q}(\phi)
+\bar r_{n,i}(\phi),
\end{equation}
where
\[
\max_{0\le i\le n-k_n}\abs{\bar r_{n,i}(\phi)}=o_p(\Delta_n^{1/4}).
\]
Indeed, every affected window has width $k_n\Delta_n\to 0$, the path $\sigma$ is c\`adl\`ag and bounded,
and the number of jumps in one such window is $O_p(1)$ by Assumption \ref{ass:noisedense}. Moreover,
on $\mathcal I_n^{J}\setminus\mathcal I_n^{J,2}$ there is exactly one jump in the $i$th pre-averaging
window, say $(\tau_{q(i)},Y_{q(i)})$, and Lemma \ref{lem:noise_pabar_decomp} gives
\[
\Delta_n^{-1/4}\bar Y_i^n(Y^0,\phi)
=
\sigma_{t_i}\zeta_{n,i}(\phi)+\omega\eta_{n,i}(\phi)+o_p(1),
\]
while the last result in Lemma \ref{lem:noise_pabar_decomp} yields
\[
\Delta_n^{-1/2}\hat Y_i^n(Y^{(n)},\phi)
=
2\omega^2\theta^{-1}\bar\phi'(2)+o_p(1)
\]
uniformly on $\mathcal I_n^{J}\setminus\mathcal I_n^{J,2}$ because a single jump of size
$\Delta_n^{1/4}$ changes $\hat Y_i^n(\phi)$ only by $O_p(\Delta_n^{1/2})$. Therefore, using \eqref{eq:noisedense_barY_expand}, Lemma \ref{lem:noise_subgaussian_replacement},
and Lemma \ref{lem:noise_shift_identity} with
\[
a=\theta\bar\phi(2)\sigma_{t_i}^2,
\qquad
b=\omega^2\theta^{-1}\bar\phi'(2),
\qquad
x=\sigma_{t_i}Y_{q(i)}w_{n,i,q(i)}(\phi),
\]
we obtain
\begin{align}
\Delta_n^{1-p/4}\Psi_{n,i}(Y^{(n)},\phi)
&=
\Delta_n\sigma_{t_i}^p
\E\abs{\sqrt{\theta\bar\phi(2)}E+Y_{q(i)}w_{n,i,q(i)}(\phi)}^p
+o_p(\Delta_n),
\label{eq:noisedense_single1}
\\
\Delta_n^{1-p/4}\Psi_{n,i}(Y^{0},\phi)
&=
\Delta_n m_p\theta^{p/2}\bar\phi(2)^{p/2}\sigma_{t_i}^p
+o_p(\Delta_n),
\label{eq:noisedense_single0}
\end{align}
uniformly over $i\in\mathcal I_n^{J}\setminus\mathcal I_n^{J,2}$. Subtracting these relations yields
\begin{equation}\label{eq:noisedense_single_diff}
\Delta_n^{1-p/4}\paren{\Psi_{n,i}(Y^{(n)},\phi)-\Psi_{n,i}(Y^0,\phi)}
=
\Delta_n\sigma_{t_i}^p
D_\phi\paren{Y_{q(i)},\frac{r_{q(i)}-i}{k_n}}
+o_p(\Delta_n).
\end{equation}
If $i\notin\mathcal I_n^{J}$, then the two windows coincide and the difference is zero. If
$i\in\mathcal I_n^{J,2}$, boundedness of $\sigma$, compact support of $F$, and the polynomial growth of
$\Psi_{n,i}$ imply
\[
\Delta_n^{1-p/4}\abs{\Psi_{n,i}(Y^{(n)},\phi)-\Psi_{n,i}(Y^0,\phi)}=O_p(\Delta_n),
\]
so Lemma \ref{lem:noisedensecounts} gives
\[
\sum_{i\in\mathcal I_n^{J,2}}
\Delta_n^{1-p/4}\abs{\Psi_{n,i}(Y^{(n)},\phi)-\Psi_{n,i}(Y^0,\phi)}
=
O_p\paren{\Delta_n\#\mathcal I_n^{J,2}}
=
o_p(\Delta_n^{1/4}).
\]
Consequently,
\begin{align}
\Delta_n^{1-p/4}\paren{V_n(Y^{(n)},\phi,p)-V_n(Y^0,\phi,p)}
&=
\Delta_n\sum_{q=1}^{N_n^J}
\sum_{i=r_q-k_n+1}^{r_q-1}
\sigma_{t_i}^pD_\phi\paren{Y_q,\frac{r_q-i}{k_n}}
+o_p(\Delta_n^{1/4}).
\label{eq:noisedense_Vpre}
\end{align}
Since $\sigma$ is c\`adl\`ag and bounded and $k_n\Delta_n\to0$,
\[
\sum_{i=r_q-k_n+1}^{r_q-1}\sigma_{t_i}^pD_\phi\paren{Y_q,\frac{r_q-i}{k_n}}
=
\sigma_{\tau_q}^p\sum_{i=r_q-k_n+1}^{r_q-1}D_\phi\paren{Y_q,\frac{r_q-i}{k_n}}+o_p(k_n).
\]
The Riemann-sum approximation then gives
\[
\frac{1}{k_n}\sum_{i=r_q-k_n+1}^{r_q-1}D_\phi\paren{Y_q,\frac{r_q-i}{k_n}}
\to
\int_0^1D_\phi(Y_q,u)\dd u.
\]
Writing
\[
H_\phi(y)=\int_0^1D_\phi(y,u)\dd u,
\]
we obtain from \eqref{eq:noisedense_Vpre}
\begin{equation}\label{eq:noisedense_Vpre2}
\Delta_n^{1-p/4}\paren{V_n(Y^{(n)},\phi,p)-V_n(Y^0,\phi,p)}
=
k_n\Delta_n\sum_{q=1}^{N_n^J}\sigma_{\tau_q}^pH_\phi(Y_q)+o_p(\Delta_n^{1/4}).
\end{equation}
Because $k_n\sqrt{\Delta_n}\to\theta$,
\[
k_n\Delta_n=\theta\Delta_n^{1/2}+o(\Delta_n^{1/2}).
\]
Conditional on $\F_T$, the Poisson compensation formula yields
\[
\E\brac{\sum_{q=1}^{N_n^J}\sigma_{\tau_q}^pH_\phi(Y_q)\mid\F_T}
=
\vartheta\Delta_n^{-1/4}A_p d_\phi^N
\]
and
\[
\Var\brac{\sum_{q=1}^{N_n^J}\sigma_{\tau_q}^pH_\phi(Y_q)\mid\F_T}
=
\vartheta\Delta_n^{-1/4}
\int_0^T\sigma_t^{2p}\dd t
\int_{\R}H_\phi(y)^2F(\dd y)
=
O_p(\Delta_n^{-1/4}).
\]
Multiplying by $k_n\Delta_n$ therefore gives
\begin{equation}\label{eq:noisedense_Vshift}
\Delta_n^{1-p/4}V_n(Y^{(n)},\phi,p)
=
\Delta_n^{1-p/4}V_n(Y^0,\phi,p)
+
\vartheta\theta A_p d_\phi^N\Delta_n^{1/4}
+
o_p(\Delta_n^{1/4}).
\end{equation}
Applying \eqref{eq:noisedense_Vshift} with $\phi=g$ and $\phi=h$, and using
Lemma \ref{lem:noise_V_limit},
\begin{equation}\label{eq:noisedense_ratio_shift}
\Delta_n^{-1/4}\paren{R_n^{PA}(Y^{(n)})-\gamma''}
=
\Delta_n^{-1/4}\paren{R_n^{PA}(Y^0)-\gamma''}
+
\frac{\vartheta\theta A_p}{\gamma'c_h}\paren{d_g^N-\gamma'\gamma''d_h^N}
+
o_p(1).   
\end{equation}
Let $\mathcal B_n^{J}$ be the collection of blocks $I_{n,b}$ that intersect $\mathcal I_n^{J}$. Since
$r_n/k_n\to\infty$, every jump contaminates at most two such blocks, and Lemma
\ref{lem:noisedensecounts} yields
\[
\#\mathcal B_n^{J}=O_p(\Delta_n^{-1/4}).
\]
For $b\notin\mathcal B_n^{J}$, we have $\Gamma_{n,b}(Y^{(n)})=\Gamma_{n,b}(Y^0)$. For
$b\in\mathcal B_n^{J}$, relation \eqref{eq:noisedense_ratio_shift} on a single block gives
\[
\abs{\Gamma_{n,b}(Y^{(n)})-\Gamma_{n,b}(Y^0)}=O_p(\Delta_n^{1/4}).
\]
Hence
\begin{align*}
\abs{\hat\varsigma_n^{PA}(Y^{(n)})^2-\hat\varsigma_n^{PA}(Y^0)^2}
&\le
\sum_{b\in\mathcal B_n^{J}}
\abs{\Gamma_{n,b}(Y^{(n)})-\Gamma_{n,b}(Y^0)}
\abs{\Gamma_{n,b}(Y^{(n)})+\Gamma_{n,b}(Y^0)}
\\
&=
O_p(\#\mathcal B_n^{J}\Delta_n^{1/2})
+
O_p\paren{\Delta_n^{1/4}\sum_{b\in\mathcal B_n^{J}}\abs{\Gamma_{n,b}(Y^0)}}
\\
&=
o_p(1),
\end{align*}
because $\#\mathcal B_n^{J}=O_p(\Delta_n^{-1/4})$ and Lemma \ref{lem:noise_block_clt} gives
$\max_b\abs{\Gamma_{n,b}(Y^0)}\toP0$. Since the same lemma yields
$\hat\varsigma_n^{PA}(Y^0)\toP\Sigma^{1/2}$, we obtain
\begin{equation}\label{eq:noisedense_selfnorm}
\hat\varsigma_n^{PA}(Y^{(n)})\toP\Sigma^{1/2}.
\end{equation}
Combining \eqref{eq:noisedense_ratio_shift} and \eqref{eq:noisedense_selfnorm} gives
\begin{equation}\label{eq:noisedense_Zshift}
Z_n^{PA}(Y^{(n)})=Z_n^{PA}(Y^0)+\mu_{D,N}+o_p(1).
\end{equation}

Now turn to the max statistic. If $j\notin\mathcal J_n^{J}$, then
$L_{n,j}(Y^{(n)})=L_{n,j}(Y^0)$. If $j\in\mathcal J_n^{J}$, each jump in
$(t_j,t_{j+2M_n}]$ contributes at most $C\Delta_n^{1/4}$ to $L_{n,j}$, so
\[
\sqrt{M_n}\abs{L_{n,j}(Y^{(n)})-L_{n,j}(Y^0)}
\le
CN_{n,j}^{J},
\]
where $N_{n,j}^{J}$ is the jump count on $(t_j,t_{j+2M_n}]$. Because
\[
\E N_{n,j}^{J}=\vartheta 2M_n\Delta_n^{3/4}=O(\Delta_n^{1/4}),
\]
we have $N_{n,j}^{J}=O_p(1)$ uniformly on $\mathcal J_n^{J}$. Lemma
\ref{lem:noisedensecounts} implies $\#\mathcal J_n^{J}=O_p(\Delta_n^{-1/4})$. Therefore, for every
fixed $y\in\R$, the tail approximation in Lemma \ref{lem:noise_max_gaussian} gives
\begin{align}
\Pbb\paren{\max_{j\in\mathcal J_n^{J}}\abs{X_{n,j}(Y^{(n)})}>A_{N_n}+B_{N_n}y}
&\le
\E\#\mathcal J_n^{J}\cdot
2\overline\Phi\paren{A_{N_n}+B_{N_n}y-C}(1+o(1))
+o(1)
\nonumber\\
&\le
C\Delta_n^{-1/4}\overline\Phi\paren{A_{N_n}-C+B_{N_n}y}+o(1)
\to 0,
\label{eq:noisedense_max_negl}
\end{align}
where $\overline\Phi=1-\Phi$ and the last limit uses $N_n\asymp\Delta_n^{-1/2}$. Thus the
contaminated disjoint windows do not affect the extreme-value limit. It remains only to check that
the feasible spot-volatility and noise-variance estimators are not spoiled by the dense small jumps. For
$\hat\omega_n^2$, the jump contribution is bounded by
\[
\frac{C}{n}\sum_{i=1}^n\{(\Delta_i^nJ^{(n)})^2+|\Delta_i^nJ^{(n)}|\,|\Delta_i^nY^0|\}=o_p((\log n)^{-2}),
\]
because $N_n^J=O_p(\Delta_n^{-1/4})$ and $|\Delta_i^nJ^{(n)}|\le C\Delta_n^{1/4}$ on the compact-support
localization. For the filtering TSRSV estimator, a spot window has length
$h_n\asymp\Delta_n^{1/6}$ and contains $O_p(h_n\Delta_n^{-1/4}+|\log\Delta_n|)=O_p(\Delta_n^{-1/12}|
\log\Delta_n|)$ jumps uniformly over the disjoint grid. The jump quadratic variation in such a window is
therefore $O_p(h_n\Delta_n^{1/4}|\log\Delta_n|)$; after division by $h_n$ it is
$O_p(\Delta_n^{1/4}|\log\Delta_n|)=o_p((\log n)^{-2})$. The cross terms with the continuous part and the
noise part are smaller by Cauchy--Schwarz and the sub-exponential bounds used in Lemma
\ref{lem:noise_tsrsv_feasible}. Consequently the feasible-array condition \eqref{eq:noise_feas} also holds
for $Y^{(n)}$, and hence
\begin{equation}\label{eq:noisedense_xi_same}
\xi_n^{LA}(Y^{(n)})-\xi_n^{LA}(Y^0)\toP 0.
\end{equation}
Combining \eqref{eq:noisedense_null_base}, \eqref{eq:noisedense_Zshift}, and
\eqref{eq:noisedense_xi_same} gives
\[
\paren{Z_n^{PA}(Y^{(n)}),\xi_n^{LA}(Y^{(n)})}\toD (Z+\mu_{D,N},\xi),
\qquad
Z\ind\xi.
\]
% The statement about $p$-values follows from the
% continuous mapping theorem.
\end{proof}

\begin{proof}[Proof of Corollary \ref{cor:noisedensepower}]
Theorem \ref{thm:noisedense} implies
\[
Z_n^{PA}\toD Z_{D,N},
\qquad
p_n^{LA}\toD U,
\qquad
Z_{D,N}\ind U,
\]
where $Z_{D,N}\sim N(\mu_{D,N},1)$ and $U\sim\mathrm{Unif}(0,1)$. Therefore
\[
\mathcal T_n^{C,N}\toD \mathcal T_{D,N}:=\frac12A(Z_{D,N})+\frac12C,
\]
where $A(\cdot)$ is given in Corollary \ref{cor:densepower} and
\[
C=\tan\brac{\pi\paren{\frac12-U}}
\]
is standard Cauchy and independent of $Z_{D,N}$. Hence
\begin{align*}
\beta_{D,N}(\alpha;\mu_{D,N})
&=
\Pbb\paren{\mathcal T_{D,N}>c_\alpha}
\\
&=
\int_{\R}
\Pbb\paren{C>2c_\alpha-A(z)}
\phi(z-\mu_{D,N})\dd z
\\
&=
\int_{\R}
\paren{\frac12-\frac{1}{\pi}\arctan\paren{2c_\alpha-A(z)}}
\phi(z-\mu_{D,N})\dd z.
\end{align*}
\end{proof}

\begin{proof}[Proof of Theorem \ref{thm:noisefixed}]
Lemma \ref{lem:noise_fixed_sum} gives
\[
R_n^{PA}\toP1.
\]
We next prove that the local-average maximum detects the fixed jump. Let $q_n$ and $j_n\in\mathcal J_n$ be measurable indices attaining the maximum in the definition of $D_n$, and set
\[
u_n=i_{q_n,n}-j_n,
\qquad
1\le u_n\le 2M_n-1.
\]
Since the fixed jump times are separated and $M_n\Delta_n\to0$, the local-average window $(t_{j_n},t_{j_n+2M_n}]$ contains no fixed jump other than $\kappa_{q_n}$ with probability tending to one. Lemma \ref{lem:noise_localavg_rep} gives the exact triangular-weight representation
\[
\sqrt{M_n}L_{n,j_n}
=
\frac{a_{M_n,u_n}}{\sqrt{M_n}}\kappa_{q_n}
+\frac{1}{\sqrt{M_n}}\sum_{u=1}^{2M_n-1}a_{M_n,u}\Delta_{j_n+u}^nX^0
+\frac{1}{\sqrt{M_n}}\sum_{u=0}^{M_n-1}(\epsilon_{j_n+M_n+u}-\epsilon_{j_n+u}),
\]
where $X^0$ is the continuous component of $X$. By Lemma \ref{lem:noise_max_gaussian}, applied to the continuous-plus-noise part,
\[
\max_{j\in\mathcal J_n}
\abs{\frac{1}{v_{n,j}}
\left\{
\frac{1}{\sqrt{M_n}}\sum_{u=1}^{2M_n-1}a_{M_n,u}\Delta_{j+u}^nX^0
+\frac{1}{\sqrt{M_n}}\sum_{u=0}^{M_n-1}(\epsilon_{j+M_n+u}-\epsilon_{j+u})
\right\}}
=O_p(\sqrt{\log n}).
\]
Moreover $0<c\le v_{n,j}\le C<\infty$ uniformly in $j$, and the finite fixed jumps do not affect the plug-in variance at the selected left endpoint $j_n$: the estimator $\hat\omega_n^2$ receives only $O_p(n^{-1})$ from finitely many squared fixed jumps, and the TSRSV spot window ending at $t_{j_n}$ contains no fixed jump for all large $n$ because the jump times are separated and $h_n\to0$. Thus
\[
\max_{j\in\mathcal J_n}\abs{\hat v_{n,j}^{-1}-v_{n,j}^{-1}}=o_p(1)
\]
on the selected jump-detecting windows. Since $\inf_q\abs{\kappa_q}>0$ a.s., condition \eqref{eq:noisy_fixed_grid_condition} implies
\[
\abs{\hat X_{n,j_n}}
\ge
c\frac{D_n}{\sqrt{M_n}}-O_p(\sqrt{\log n})
\]
and therefore
\[
\frac{\mathcal M_n^{LA}}{\sqrt{\log n}}\toP\infty.
\]
Because $A_{N_n}=O(\sqrt{\log n})$ and $B_{N_n}^{-1}=O(\sqrt{\log n})$, it follows that
\[
\frac{\xi_n^{LA}}{\log n}\toP\infty.
\]
Consequently, for every fixed $L>0$,
\[
p_n^{LA}=1-\exp\{-e^{-\xi_n^{LA}}\}=o_p(\Delta_n^L),
\]
and hence
\[
\tan\brac{\pi\paren{\frac12-p_n^{LA}}}
=\cot(\pi p_n^{LA})
\]
diverges faster than any fixed power of $\Delta_n^{-1}$.

It remains to check that the other Cauchy component cannot cancel this divergence. Under fixed finite-activity jumps, the variables entering the block self-normalizer are at most polynomial in $\Delta_n^{-1}$ after localization and after using the sub-Gaussian maximal bound
$\max_{0\le i\le n}\abs{\epsilon_i}=O_p(\sqrt{\log n})$. Hence there exists a finite constant $C_p$ such that
\[
\hat\varsigma_n^{PA}=O_p(\Delta_n^{-C_p}).
\]
Since $R_n^{PA}\toP1$ and $\gamma''>1$, the numerator $\Delta_n^{-1/4}(R_n^{PA}-\gamma'')$ is bounded away from zero in absolute value times $\Delta_n^{-1/4}$ with probability tending to one. Thus
\[
\abs{Z_n^{PA}}\ge \Delta_n^{C_p}
\]
with probability tending to one, after increasing $C_p$ if necessary. The elementary inequalities
\[
1-p_n^{PA}=2\Phi(\abs{Z_n^{PA}})-1\ge c(\abs{Z_n^{PA}}\wedge1),
\qquad
\tan\brac{\pi\paren{\frac12-u}}
\ge -\frac{C}{1-u},\quad 0<u<1,
\]
then imply
\[
\tan\brac{\pi\paren{\frac12-p_n^{PA}}}=O_p^{-}(\Delta_n^{-C_p}),
\]
where $O_p^{-}(a_n)$ means that the negative part is $O_p(a_n)$. Choosing $L>C_p$ in the preceding bound for $p_n^{LA}$ gives
\[
\mathcal T_n^{C,N}\toP\infty.
\]
Therefore, for every $0<\alpha<1$,
\[
\Pbb\paren{p_n^{C,N}\le\alpha}
=
\Pbb\paren{\mathcal T_n^{C,N}>c_\alpha}\to1.
\]
\end{proof}
\bibliographystyle{elsarticle-harv}
\bibliography{refs.bib}

\end{document}